\documentclass[11pt, letterpaper]{amsart}
\usepackage[english]{babel}
\usepackage{amssymb,amsthm,amsmath}
\usepackage[numbers,sort&compress]{natbib}
\usepackage{color}
\usepackage{graphicx}
\usepackage{tikz}
\usepackage[colorlinks,
            linkcolor=blue,
            anchorcolor=blue,
            citecolor=blue
            ]{hyperref}
\usepackage{prerex}
\usepackage{mathabx}

\numberwithin{equation}{section}

\def\R{\mathbb{R}}
\def\C{\mathbb{C}}
\def\Z{\mathbb{Z}}
\def\Q{\mathbb{Q}}
\def\N{\mathbb{N}}

\newcommand{\dist}{\mathrm{dist_H}}
\newcommand{\dis}{\mathrm{dist}}

\newcommand{\openrm}{\mathrm{(}}
\newcommand{\closerm}{\mathrm{)}}
\newcommand{\op}{H_{\alpha, \lambda, \theta}}

\newcommand{\ids}{\mathcal{N}_{\alpha, \lambda}}

\newcommand{\specU}{S(\alpha, \lambda)}

\newcommand{\sqnp}{ S\left(\frac{\widetilde{p}}{\widetilde{q}}, \lambda\right)}

\newcommand{\Itau}{I^j_\tau}

\newcommand{\Itauxi}{I^j_{\tau, \xi}}
\newcommand{\Itausxi}{I^j_{\tau', \xi}}

\newcommand{\specUq}{S\left(\frac{p}{q}, \lambda\right)}
\newcommand{\specIq}{S_{-}\left(\frac{p}{q}, \lambda\right)}
\newcommand{\Discr}{ \Delta_{\frac{p}{q}, \lambda}(E)}

\def\bm{\begin{matrix}}
\def\em{\end{matrix}}
\newcommand{\dimension}{\operatorname{dim}_H}
\newcommand{\cont}{\mathcal{H}^{\omega}_\infty}
\newcommand{\measphi}{\mathcal{H}^{\omega}}

\def\uc{\operatorname{UC}}

\theoremstyle{plain}
\newtheorem{theorem}{Theorem}
\theoremstyle{plain}
\newtheorem{proposition}{Proposition}[section]
\newtheorem{defin}[proposition]{Definition}
\newtheorem{lemma}[proposition]{Lemma}
\newtheorem{rmk}[proposition]{Remark} 
\newtheorem{corollary}[proposition]{Corollary}
\newtheorem{claim}[proposition]{Claim}

\newtheorem*{claim*}{Claim}

\newtheoremstyle{named}{}{}{\itshape}{}{\bfseries}{.}{.5em}{\thmnote{#3}#1}
\theoremstyle{named}
\newtheorem*{namedtheorem}{}

\begin{document}

\title[]{On the abominable properties of the Almost Mathieu operator with well approximated frequencies}
\author{Artur Avila}
\address{Institut f\"{u}r Mathematik, Universit\"{a}t Z\"{u}rich Winterthurerstrasse 190, CH-8057 Z\"{u}rich, Switzerland \& IMPA, Estrada Dona Castorina 110, 22460-320, Rio
de Janeiro, Brazil}
 \email{artur.avila@math.uzh.ch}

\author {Yoram Last}
\address{
Institute of Mathematics, The Hebrew University of Jerusalem, 91904 Jerusalem, Israel.
} \email{ylast@math.huji.ac.il}

\author {Mira Shamis}
\address{School of Mathematical Sciences,
Queen Mary University of London, 
Mile End Road, London E1 4NS, England} \email{m.shamis@qmul.ac.uk}

\author{Qi Zhou}
\address{Chern Institute of Mathematics and LPMC, Nankai University, Tianjin 300071, China
} \email{qizhou@nankai.edu.cn}

\date{\today}

\maketitle
\begin{abstract}
We show that some spectral properties of the almost Mathieu operator with  frequency well approximated by rationals can be as poor as at all possible in the class of all one-dimensional discrete Schr\"odinger operators. For the case of critical coupling, we show that the Hausdorff measure of the spectrum may vanish (for appropriately chosen frequencies) whenever the gauge function tends to zero faster than logarithmically. For arbitrary coupling, we show that modulus of continuity of the integrated density of states can be arbitrary close to logarithmic; we also prove a similar result for the Lyapunov exponent as a function of the spectral parameter. Finally, we show that (for any coupling) there exist frequencies for which the spectrum is not  homogeneous in the sense of Carleson, and, moreover, fails the Parreau--Widom condition.  The frequencies for which these properties hold are explicitly described in terms of the growth of the denominators of the convergents.
\end{abstract}

\section{Introduction}

This paper concerns the almost Mathieu operator
\begin{equation}\label{eq:def-op}
(\op\, \psi)(n) = \psi(n + 1) + \psi(n - 1) +2\lambda\cos(2\pi\alpha n +\theta)\psi(n),
\end{equation}
 where $\alpha, \lambda,\theta\in\R$ and $\psi\in\ell^2(\Z\to\C)$. The almost Mathieu operator was first introduced by Peierls \cite{Pe}
as a Hamiltonian describing the motion of an electron on a two-dimensional lattice  in the presence  of a homogeneous magnetic field \cite{Ha,R}. This model has been extensively  studied not only because of  its importance in  physics \cite{AOS,OA}, but also as a fascinating mathematical object. Indeed, by varying the parameters $\lambda$, $\alpha$, and $\theta$, one sees surprising spectral richness and it thus serves as a primary example for many interesting spectral phenomena, for example, the Cantor  structure of the spectrum \cite{AJ05}, sharp phase transitions between three spectral types \cite{Aab,AYZ}, and also the universal hierarchical structure of quasiperiodic eigenfunctions \cite{JLiu}.

It is known that the spectral properties of the almost Mathieu operator depend sensitively on the arithmetic properties of $\alpha$. We illustrate this phenomenon with one archetypal example. It was found by Gordon~\cite{gordon} that Schr\"odinger operators which are well approximated by periodic ones can only have continuous spectrum. Using this result, Avron and Simon \cite{avron-simon:gordon} showed that for well approximated $\alpha$ the spectrum of $\op$, is purely singular continuous for $ \lambda > 1$. 
On the other hand,  if $\alpha$  satisfies a  Diophantine condition  of the form
\begin{equation}\label{eq:dgl}
\inf_{p\in\Z}|n\alpha-p| \geq \frac{\gamma}{|n|^{\tau}}~,\,\, \,n\in\Z\setminus\{ 0\}~,
\end{equation}
then  $\op$ has Anderson localization  for $\lambda > 1$ and a.e. $\theta$  \cite{j-metal}.

This indicated the possibility of a phase transition between singular continuous spectrum and pure point spectrum.  And indeed, in \cite{j-almost-everything} Jitomirskaya introduced the arithmetic parameter
\begin{equation}\label{eq:beta}
\beta(\alpha) = \limsup_{q\to\infty} \frac1q \log \min_p \left| \alpha - \frac{p}{q}\right|  = \limsup_{n\to\infty} \frac{\log q_{n + 1}}{q_n}~,
\end{equation}
where $\left\{ \frac{p_n}{q_n}\right\}$ is the sequence of convergents of $\alpha$, and conjectured that  the spectrum is purely singular continuous if $1 \leq \lambda < e^{\beta(\alpha)}$, while it is pure point  for (an arithmetically defined set of) a.e. $\theta$ if $ \lambda> e^{\beta(\alpha)}$. This conjecture has been recently  proved in \cite{AYZ, JLiu, j-crit-mathieu}.

The goal of the current paper is to show that some spectral properties of $\op$ with very well approximated $\alpha$ can be as poor as at all possible in the class of all Schr\"odinger operators acting on $\ell^2(\Z)$ by
\begin{equation}\label{eq:general-op}
(H\psi)(n) = \psi(n + 1) + \psi(n - 1) + V(n)\psi(n)~.
\end{equation} 

We focus on the following  characteristics: the Hausdorff measure of the spectrum, regularity of the integrated density of states and of the Lyapunov exponent, homogeneity and the Parreau-Widom property of the spectrum. Now let us state the results precisely.

\subsection{Hausdorff measure of the spectrum} Let $\omega: [0, 1] \rightarrow [0, \infty)$ be a gauge function, i.e. a non-decreasing function with $\omega(0) = 0$. The main examples are
\begin{equation}\label{defphi}
\omega_t(s) = \begin{cases}
0 & s = 0\\  \frac{1}{\log^t \frac{1}{e s}} & \text{otherwise}\end{cases}  \,\, (\text{for }t\in(0, \infty)) \quad \text{and} \quad  \widetilde\omega_t(s) = s^t  \,\, (\text{for }t\in(0, 1])~.
\end{equation}

The $\omega$-Hausdorff measure $\measphi(K)$ of a set $K\subset\R$ is by definition the quantity
\begin{equation*}
\measphi(K) = \lim_{\epsilon\to 0} \mathcal{H}^{\omega}_\epsilon (K)~,
\end{equation*}
where 
\begin{equation}\label{eq:meas-omega}
\mathcal{H}^{\omega}_\epsilon (K) = \inf\left\{ \sum_{j = 1}^\infty \omega(b_j - a_j)\quad \big|\quad K\subset\bigcup_{j=1}^\infty (a_j, b_j), \,\,  b_j - a_j\leq \epsilon\right\}~.
\end{equation}
It is known that there exist Schr\"odinger operators with spectrum of zero Hausdorff dimension (i.e.\ zero $\widetilde \omega_t$-Hausdorff measure for any $t > 0$), see the discussion below, the works of Avila \cite{Av-lp} and Damanik--Fillman--Lukic \cite{DFL}, and also the recent survey of Damanik and Fillman \cite{DamFil} and references therein. On the other hand, it is known (see Remark~\ref{rmk:frostman}) that the spectrum of any one dimensional discrete Schr\"odinger operator $H$ of the form \eqref{eq:general-op} is of positive $\omega_1$-Hausdorff measure. Our first result bridges the gap by showing that the latter can not be improved, even for the almost Mathieu operator.
\begin{theorem}\label{thm:hausdorff} For any $\alpha \in \mathbb R \setminus \Q$, let $S(\alpha, \lambda)$ be the spectrum of $\op$.  The following holds for $\lambda = 1$:
\begin{enumerate}
\item\label{hausdorff} If $\alpha \in \mathbb R \setminus \mathbb Q$ is such that $\beta(\alpha) > 0$, then $\mathcal{H}^{\omega_t}(S(\alpha, 1)) = 0$ for any $t > 2$, with  $\omega_t$ as in  \eqref{defphi}.
\item\label{thm:hausdorff-opt} If $\omega$ is any gauge function decaying faster than $\omega_1$, i.e.\ such that
\begin{equation}\label{eq:hausdorff-opt-cond}
\lim_{s\to+0}\omega(s)\log\frac{1}{s} = 0~,
\end{equation}
then  there exists a $G_\delta$-dense set of $\alpha\in\R\setminus\Q$ for which $\measphi(S(\alpha, 1)) = 0$.
\end{enumerate}
\end{theorem}

\begin{rmk}\label{set-of-alpha} The set is explicit but not described in terms of $\beta(\alpha)$.

\end{rmk}

To put this result in context, recall that \cite{AK, last3} for any irrational $\alpha$ the spectrum $S(\alpha, 1)$ is a Cantor set of zero Lebesgue measure. A conjecture  from the 1980's asserts that its Hausdorff dimension equals $\frac12$:
\begin{equation*}\mathrm{dim_H}\, S(\alpha, 1) \overset{\text{def}}{=} \sup \left\{t > 0 \, \mid \, \mathcal{H}^{\widetilde \omega_t}(S(\alpha, 1)) > 0\right\} = \frac{1}{2}~. \end{equation*}
However, Wilkinson and Austin~\cite{wilkinson_austin} argued that the fractal dimension of the spectrum should be sensitive to the arithmetic properties of $\alpha$. The problem of determining the Hausdorff dimension of the spectrum of the critical almost Mathieu operator was also recently advertised by  Simon in his lecture \cite{simon}. 

We are aware of the following rigorous results.  Last~\cite{last3} showed that $\dimension S(\alpha, 1)\leq \frac{1}{2}$ if $\limsup q_{n+1} q_n^{-3} > 0$. In \cite{LS} two of the authors proved that there exists a $G_\delta$-dense set of  $\alpha$ such that $\dimension S(\alpha, 1)= 0$.
Recently, Helffer, Liu, Qu and Zhou  \cite{hz}, building on the works of Helffer and Sj\"ostrand \cite{HS1}-\cite{HS2}, showed that there is also a positive Hausdorff-dimensional set of $\alpha$ with $\beta(\alpha) = 0$ and $\dimension S(\alpha, 1)  > 0$.  Finally, Jitomirskaya and Krasovsky \cite{jk}  showed that $\dimension S(\alpha, 1)\leq \frac{1}{2}$ for all $\alpha \in\R\setminus\Q$.   We also mention that Jitomirskaya and Zhang \cite{JZ} proved that  if $\beta(\alpha) > 0$, then the box-counting (Minkowski) dimension of the spectrum is equal to $1$. The dramatic difference between the Hausdorff and the Minkowski dimensions is a signature of the multifractal structure of the spectrum.

\subsection{Regularity of the integrated density of states and of the Lyapunov exponent} Recall the following basic setup (see \cite{figotin-pastur}). Let $(\Omega, \mathcal{F}, \mathbb P)$ be a probability space, let $T: \Omega\to\Omega$ be an invertible ergodic transformation, and let $v: \Omega\to\R$ be a measurable function. For each $\varpi\in\Omega$ construct a Schr\"odinger operator of the form \eqref{eq:general-op} by taking
\begin{equation}\label{eq:dyn-poten}
V_\varpi(n) = v(T^n\varpi)~.
\end{equation}
The almost Mathieu operator is obtained by taking 
\[\begin{split}
&\begin{cases}
\Omega = \R/ 2\pi\Z~, &\alpha\in\R\setminus\Q \\
\Omega = \{\theta'\in\R/ 2\pi\Z:\, \theta' - \theta \in\frac{2\pi}{q}\Z\}~, & \alpha = \frac p q \in\Q
\end{cases}\\
&\,\,T(\theta) = \theta + 2\pi\alpha~, \quad \text{and} \quad v(\theta) = 2\lambda\cos\theta~. \end{split}\] 
Operators obtained by replacing this $v$ by another continuous $v:  \R/ 2\pi\Z \to\R$ are called one-frequency operators.

For an operator $H_\varpi$ defined by \eqref{eq:dyn-poten}, the integrated density of states is defined by
\begin{equation}\label{eq:def-ids}
\mathcal{N}\, (E) = \lim_{N\to\infty}\frac{1}{N}\mathbb E_\varpi \{ \text{number of eigenvalues of}\, H^{[0, \dots, N- 1]}_\varpi\, \text{in}\, (-\infty, E]\}~,
\end{equation}
where $H^{[0, \dots, N- 1]}_\varpi$ denotes the restriction of the operator $H_\varpi$ to the interval $[0, \dots, N- 1]\subset\Z$. 

For each $E\in\R$, we associate with $H_\varpi$ the Schr\"odinger cocycle
\begin{equation}\label{eq:cocycle}
A_E: \Omega\times\R^2 \rightarrow \Omega\times\R^2, \\
A_E(\varpi, x) = \left( T\varpi, \left(\begin{array}{cc}
E - v(\varpi) & -1 \\
  1      &   0  \\
\end{array}\right)x\right),~
\end{equation}
and define the Lyapunov exponent
\begin{equation}\label{eq:LE-cocycle}
\gamma(E) =  \lim_{N\to\infty}\frac{1}{N}\mathbb E_\varpi \log\|A^{(N)}_E(\varpi, \cdot) \|~,
\end{equation}
where $A^{(N)}_E(\varpi, \cdot)$ is the $2\times 2$ matrix appearing in the second component of the $N$-th iterate of \eqref{eq:cocycle}. In both \eqref{eq:def-ids} and \eqref{eq:LE-cocycle} the expectation $\mathbb E_\varpi$ can be replaced with an almost sure limit, and, if $T$ is uniquely ergodic,  \eqref{eq:def-ids} can even be replaced with a pointwise limit. The integrated density of states and the Lyapunov exponent are related by the  formula
\begin{equation}\label{eq:thouless}
\gamma(E) = \int \log |E - E'| \mathrm{d}\, \mathcal{N} (E')
\end{equation}
 going back to the work of Thouless, Herbert and Jones (see \cite{CS2}).

We are interested in the uniform continuity of $\gamma(E)$ and $\mathcal N(E)$. If $\omega$ is a gauge function, let
\begin{equation}\label{eq:def-uc}
\uc [\omega] =  \left \{ f:\R\rightarrow\C\, |\, \sup_{0 < |E-E'|\leq 1} \frac{|f(E) - f(E')|}{\omega(|E - E'|)} < \infty\right\}
\end{equation}
be the class of uniformly continuous functions with modulus of continuity majorated by $\omega$.
Craig and Simon~\cite{craig-simon} showed using \eqref{eq:thouless} that for any ergodic operator $H_\varpi$
\begin{equation}\label{eq:craig-simon}
\mathcal N \in\uc[\omega_1]~,
\end{equation}
where $\omega_1$ is as defined in \eqref{defphi} with $t = 1$. It is known from the work of Craig \cite{Craig} and Gan--Kr\"uger \cite{kg} that (\ref{eq:craig-simon}) is optimal in the class of operators (\ref{eq:general-op}) with almost periodic potentials. The operators constructed in \cite{Craig} have pure point spectrum of arbitrary Lebesgue measure; those of \cite{kg} have spectrum of zero Lebesgue measure, which is consequently also singular. We also mention the recent work  of Duarte-Klein-Santos  \cite{DKS}, who constructed a class of random Schr\"odinger operators such that $\mathcal N  \notin \bigcup_{t > 2} \uc\,[\omega_t]$. In all these cases, the Lyapunov exponent $\gamma$ is also non-regular, due to a lemma of Goldstein--Schlag that we reproduce in a generalized form as Proposition~\ref{prop-gs}.

In the context of the  almost Mathieu operator, the following was known. In the critical case ($\lambda = 1$), it follows from \cite[Section 8]{Bourg2} and \cite{LS} that there exist irrational $\alpha$ for which the integrated density of states $\mathcal N_{\alpha, 1}\notin \bigcup_{t > 0} \uc\,[\widetilde{\omega}_t]$. For arbitrary $\lambda\neq 0$, it was pointed out by Avila and Jitomirskaya \cite{AJ08}, $\ids \notin \bigcup_{t > 0} \uc\,[\widetilde{\omega}_t]$ for Baire-generic $\alpha$.  These results stand in contrast with what is known for Diophantine frequencies:  it was shown in \cite{Aab}, following earlier works of Goldstein--Schlag, Avila--Jitomirskaya and Bourgain \cite{GS,goldsteinschlag2,AJ08,bourgain} (in which stronger conditions on $\alpha$ were imposed), that if $\beta(\alpha) = 0$, then $\ids \in \uc[\widetilde \omega_{\frac12}]$ for any $\lambda \neq 1$. The cited works, as well as  \cite{amor,CCYZ}, are applicable to more general quasiperiodic operators. We also mention $\uc[\widetilde \omega_{\frac12}]$ continuity is sharp for general quasi-periodic operators \cite{P}, the density of optimal (pointwise) $t-$H\"older continuity with $\frac{1}{2}<t<1$ was recently proved in the subcritical regime \cite{KXZ}.   

 Theorem  \ref{thm:hausdorff} (combined with Frostman's lemma stated as Lemma~\ref{lem:frostman} below) implies that \eqref{eq:craig-simon} is optimal for the critical case $\lambda=1$. The next result shows that \eqref{eq:craig-simon} is  optimal for arbitrary $\lambda\neq 0$.

\begin{theorem}\label{th:ids} Consider the almost Mathieu operator $\op$ with $\alpha\in\R\setminus\Q$.
\begin{enumerate}
\item\label{eq:thm-ids1} If $e^{-\frac{2\beta(\alpha)}{3}} < \lambda < e^{\frac{2\beta(\alpha)}{3}}$, then $(a)\,\, \ids \notin \bigcup_{t > 3} \uc\,[\omega_t]$ and $(b)\,\,\gamma_{\alpha, \lambda}\notin \bigcup_{t > 4} \uc\,[\omega_t]$~. 
\item\label{eq:thm-ids2} If $e^{-\frac{\beta(\alpha)}{2}} < \lambda < e^{\frac{\beta(\alpha)}{2}}$, then $(a)\,\, \ids \notin \bigcup_{t > 2} \uc\,[\omega_t]$ and $(b)\,\, \gamma_{\alpha, \lambda} \notin \bigcup_{t > 3} \uc\,[\omega_t]$~.
\item\label{thm:ids-opt} If $\omega$ is a gauge function satisfying \eqref{eq:hausdorff-opt-cond}, then for any $\lambda\neq 0$ there exists a $G_\delta$-dense set of $\alpha\in\R\setminus\Q$ such that $(a)\,\, \ids\notin\uc\,[\omega]$, and $(b)\,\, \gamma_{\alpha, \lambda} \notin \bigcup_{t > 2} \uc\,[\omega_t]$~. 
\end{enumerate}
\end{theorem}
This provides the first explicit examples of quasiperiodic Schr\"odinger operators with low regularity of $\ids$ and  $\gamma_{\alpha, \lambda}$. An interesting feature of our construction  for $\lambda < 1$ is that the spectrum is purely absolutely continuous, and  for all $\lambda \neq 1$  $\ids$ is absolutely continuous  \cite{AD}.

\subsection{Homogeneity and the Parreau-Widom property of the spectrum} Our last two results are motivated by the inverse spectral problem for quasiperiodic operators: given a compact set $K\subset\R$, what are the spectral properties of the  Jacobi operators whose spectrum coincides with $K$. In \cite{SY95, SY97}  Sodin and Yuditskii studied the inverse spectral problem in the class of  Jacobi operators with almost periodic potentials. In particular, they showed that if the spectrum of a Jacobi operator is homogeneous (see Definition~\ref{def:homog}) and the diagonal elements of the corresponding resolvent operator are purely imaginary on the spectrum, then the operator is almost periodic. 

To the best of our knowledge, all the works on the inverse spectral problem, starting from \cite{SY95, SY97}, require either homogeneity of the spectrum, meaning that it can not be too meager near any point, or the weaker Parreau-Widom property (Definition~\ref{def:homog}), which means that the space of analytic functions on its complement is sufficiently rich \cite{Wid1,Wid2}. Therefore it is natural to ask whether the homogeneity or at least the Parreau-Widom property hold for simple quasiperiodic operators, namely to consider the above questions in the converse direction. Somewhat surprisingly, the answer is not always positive. 

To state the results precisely, recall
\begin{defin}\label{def:homog}
A set $K\subset\R$ is called homogeneous $\openrm$in the sense of Carleson$\closerm$ if there exist $\epsilon_0 > 0$ and $0 < \tau \leq 1$ such that for any $E\in K$ and for any $0< \epsilon \leq\epsilon_0$
\begin{equation}\label{eq:homog}
\left| (E - \epsilon, E + \epsilon)\cap K\right|\geq\tau\epsilon~.
\end{equation}
A set $K\subset\R$ is said to satisfy the Parreau-Widom condition if
\begin{equation}\label{eq:pw-cond}
\sum_{(a, b)}\max_{E\in(a, b)} g(E) < \infty,
\end{equation}
where $g(z)$ is the Green function of the Dirichlet Laplacian in $\C\setminus K$, and the sum is over the bounded connected components $(a,b)$ of $\R\setminus K$.
\end{defin}
\begin{rmk}
Every homogeneous set satisfies the Parreau--Widom property, while the converse implication is not true -- see  Jones and Marshall \cite[pp.~297––298]{JM}.
\end{rmk}

For continual  quasiperiodic Schr\"{o}dinger operators of the form $H=-\frac{d^2}{dx^2}+V$,  homogeneity of the spectrum was shown by Damanik-Goldstein-Lukic \cite{DGL}, provided that $V$ is small enough and that $\alpha$ satisfies the Diophantine condition.
For discrete quasiperiodic (one-frequency) Schr\"{o}dinger operators Damanik-Goldstein-Schlag-Voda  have obtained \cite[Theorem H]{DGSV} that in the regime of positive Lyapunov exponent, for strong  Diophantine  $\alpha$, each non-empty intersection of the spectrum with an open interval is homogeneous. Leguil-You-Zhao-Zhou \cite{LYZZ} further proved that if $\alpha$  is strong Diophantine,  then for a (measure-theoretically) typical real-analytic potential, the spectrum is homogeneous.  Later, Liu-Shi \cite{LiuS} proved the spectrum is also homogeneous for weak Liouvillean frequency. Very recently, K.~Tao proved \cite{KTao}  that the same is true for a class of Gevrey potentials when the frequency satisfies Diophantine condition.

Recently,  Simon \cite{simon1} conjectured that the spectrum $ \specU$ of the operator $\op$ is  homogeneous for any  $\lambda\neq \pm1$. We disprove this conjecture in the following strong sense.

\begin{theorem}\label{thm:homog} Assume $\alpha\in\R\setminus\Q$. Then we have the following:
\begin{enumerate}
\item \label{thm-homo} If $e^{-{\frac{2\beta (\alpha)}{3}}} < \lambda < e^{\frac{2\beta (\alpha)}{3}}$, then $S(\alpha, \lambda)$ is not homogeneous.
\item \label{thm-pw}  If $e^{-\frac{\beta(\alpha)}{3}}<\lambda < e^{\frac{\beta(\alpha)}{3}}$, then $S(\alpha, \lambda)$ does not satisfy the Parreau-Widom condition.
\end{enumerate}
\end{theorem}
To the best of our knowledge, these are the first examples of ergodic Schr\"odinger operators with a spectrum of positive Lebesgue measure which is not homogeneous and does not even satisfy the Parreau--Widom condition. These examples show that a  solution of the inverse spectral problem which would include the almost Mathieu operator with Liouville frequencies may require essentially new methods.

\subsection{Structure of the paper, and an avalanche lemma} The proofs of the three theorems are based on the approximation of $\op$ by periodic operators $H_{\frac p q, \lambda, \theta}$. The main arguments are developed for $\lambda \leq 1$ and then extended to $\lambda > 1$ by Aubry--Andr\'e duality. 

For fixed $\theta$, the spectrum of $H_{\frac p q, \lambda, \theta}$ consists of $q$ bands. Using Chambers' formula (Proposition~\ref{prop1}), we show that for $\lambda \leq 1$, the lengths of these bands are not exponentially small in $q$, whereas as we vary $\theta$, the edges of a band vary by an amount of order $(\lambda + o(1))^q$. The small\footnote{note that the intervals grow as $\lambda \to 1 - 0$, and for $\lambda =1$ they cover most of the spectrum} intervals near the band edges which lie in the spectrum for some $\theta$ but not for other ones will be the main focus of our attention.

When we pass from $\frac p q$ to $\frac {\widetilde p} {\widetilde q}$, the spectrum moves by a quantity which is bounded by $C |\frac p q - \frac {\widetilde p} {\widetilde q}|^{\frac12}$ (see Proposition~\ref{continuity_of_spectra} due to Avron, van Mouche and Simon). Our assumptions on $\alpha$ ensure that, for an appropriately chosen sequence of approximants, this quantity is much smaller than the lengths of the intervals described above. Thus the intervals have to contain some spectrum of $H_{\frac {\widetilde p}{\widetilde q}, \lambda, \theta}$, and, eventually, of $\op$. In Proposition~\ref{th:ky-fan}, possibly of independent interest, we prove a counterpart of  Proposition~\ref{continuity_of_spectra} for the integrated density of states. It is used to show that the ($\theta$-averaged) integrated density of states that $H_{\frac {\widetilde p}{\widetilde q}, \lambda, \theta}$ and $\op$ assign to these intervals is also not too small. 

On the other hand, in Corollaries~\ref{cor:LS}  and \ref{cor:LS1}  we show that the spectrum of $H_{\frac {\widetilde p}{\widetilde q}, \lambda, \theta}$ in these intervals is very meager: it is small when measured by any measure having a decent amount of regularity. This basic fact, combined with the lower bound on the amount of spectrum in the intervals as described above, eventually implies all the results of the paper. One ingredient  used in the deductions and elsewhere in the paper is the Frostman lemma (Lemma~\ref{lem:frostman}).

Corollaries~\ref{cor:LS}  and \ref{cor:LS1}   are deduced from Propositions~\ref{prop:LS} and \ref{prop:LS1}, which are the main technical results of the current paper. These propositions provide a lower bound on the Lyapunov exponent $\gamma_{\frac {\widetilde p}{\widetilde q}, \lambda, \theta}(E + i\epsilon)$ when $E$ lies in the intervals described above, and $\epsilon$ can be made very small. These estimates improve and generalize \cite[Theorem 2]{LS}, where only the case $\lambda = 1$ was considered, and much stronger approximability was assumed. These improvements are possible due to combination of several new ingredients. One of these is a new  avalanche-type lemma, which we state here as it may be of independent interest:

\begin{namedtheorem}[Proposition~\ref{lem:avalanche}]
For every $0 < c < 1$ the following holds for sufficiently small $0 < b < b_0(c)$. Let $A_j \in SL_2(\C)$ and let $0 < \delta_j < b$, $1 \leq j \leq n$, be such that
\begin{enumerate}\label{cond-avalanche}
\item\label{avalanche-cond1} $\delta_{j + 1} \leq \delta_j + b\delta^{3/2}_j, \,\,\, 1 \leq j\leq n - 1$,
\item\label{avalanche-cond2} $\|A_{j + 1} - A_j \|\leq b\delta_j, \,\,\, 1 \leq j\leq n - 1$,
\item\label{avalanche-cond3} $|\mathrm{Tr}\, A_j| \geq 2+ (1 - b)\delta_j, \,\,\, 1 \leq j \leq n$,
\end{enumerate}
Then, for any vector $u_0 \in \C^2$ such that $\|A_1 u_0 \| \geq \exp((1 - c) \sqrt{\delta_1})\| u_0\|$,
\begin{equation}\label{eq:avalanche}
\|A_n\cdots A_1 u_0\| \geq \exp\left((1 - c)\sum_{j = 1}^n \sqrt{\delta_j}\right)\| u_0\|.
\end{equation}
\end{namedtheorem}

This lemma falls into the category of avalanche principles, which deduce the global growth of a matrix product from conditions on pairs of adjacent matrices. The original avalanche principle was introduced by Goldstein--Schlag \cite{GS}; a version incorporating   subsequent refinements can be found in the monograph of Duarte--Klein \cite{DK}. The current lemma,  improving on  \cite[Theorem 3]{LS},   is useful when the hyberbolicity of the terms in the matrix product is weak.

\subsection{More general one-frequency operators} It is natural to expect that some of the results of the current paper could be extended to more general one-frequency operators. For example, in the supercritical case the role of the coupling constant should be played by $e^\gamma$, where $\gamma$ is the Lyapunov exponent of the operator. However, our current techniques heavily rely on the fine properties of the almost Mathieu operator, including self-duality, Chambers' formula and estimates on the gaps.

\subsection{Dependency chart}The structure of the paper is illustrated by the  chart in Figure~\ref{fig:dep}.

\begin{figure}
  \centering
   \includegraphics[width=13cm]{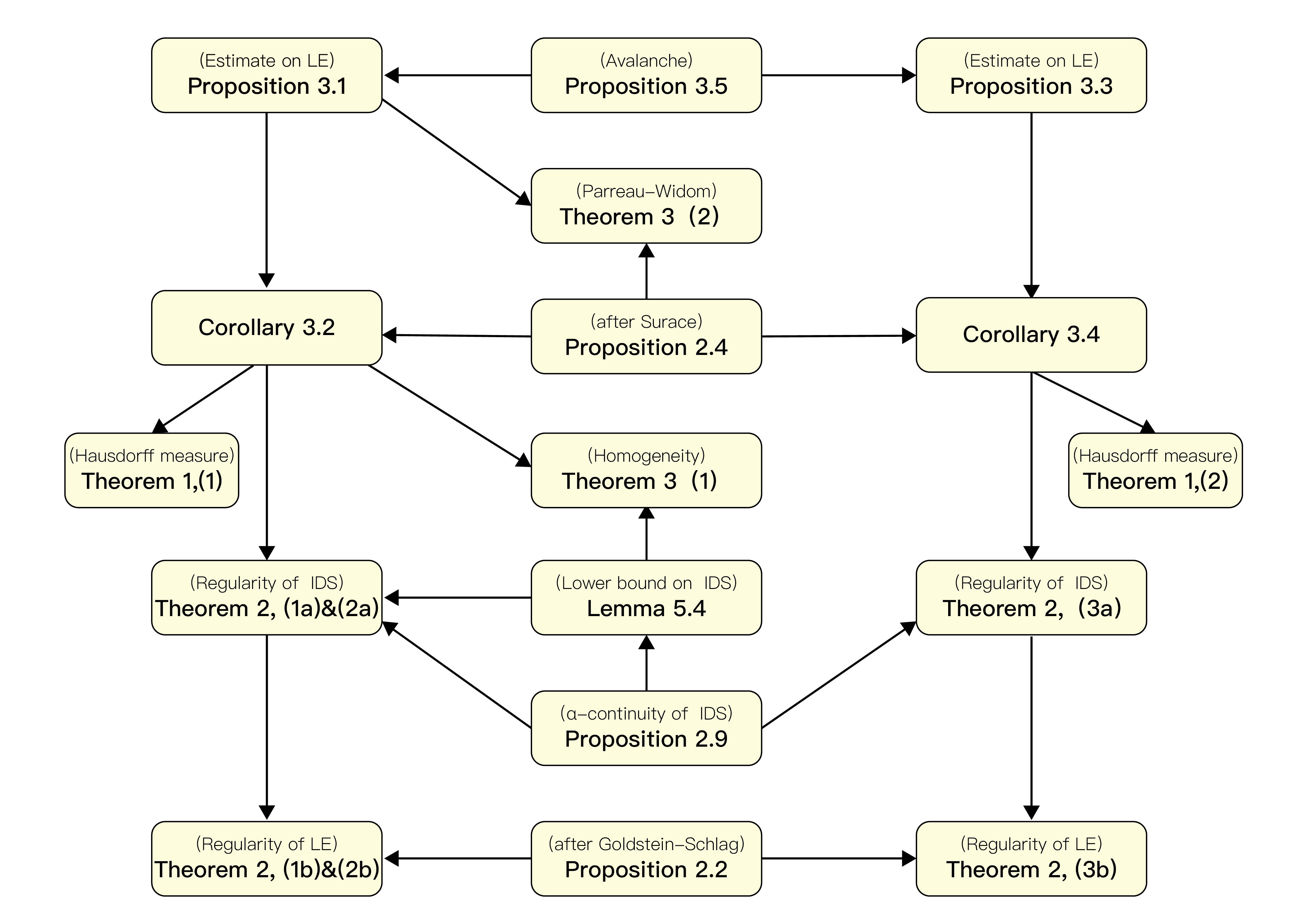} \\
\caption{Dependency Chart}\label{fig:dep}
\end{figure}

\section{Three auxiliary propositions}
This section contains several auxiliary results which we will  use in the sequel.  In the first part, we prove two statements which we need in greater generality than we could find in the literature (their common feature is that estimates of singular integrals appear in the proof). The first one, Proposition~\ref{prop-gs}, relates the regularity of the integrated density of states to that of the Lyapunov exponent. The argument closely follows  Goldstein--Schlag \cite[Lemma 10.3]{GS}, who proved a similar but less general statement (e.g.\  their formulation applies to $\widetilde \omega_t$ but not to $\omega_t$).  The second one, Proposition~\ref{prop-surace}, shows that a decently regular measure $\mu$ assigns little mass to the set of energies $E$ for which the Lyapunov exponent evaluated at $E$ differs significantly from its value at $E+i\epsilon$. It is an extension of  \cite[Lemma 2]{surace} due to Surace, who considered the case when $\mu$ is the Lebesgue measure. 

In the second part of the section, we recall a result of Avron, van Mouche and Simon  \cite{simon_vanmouche_avron} pertaining to the continuity of the spectrum of a one-frequency operator as a function of the frequency $\alpha$ (Proposition~\ref{continuity_of_spectra}), and prove its counterpart for the integrated density of states (Proposition~\ref{th:ky-fan}).

\subsection{Two estimates on singular integrals}
To state the results, we introduce some notation.

Let $\omega: (0, 1] \rightarrow (0, \infty)$ (it will be convenient not to insist that $\omega$ is non-decreasing or continuous) and define $\uc\, [\omega]$ as in \eqref{eq:def-uc}. For any $j \geq 0$ define
\begin{enumerate}
\item\label{w1} $(W_1\omega)(2^{-j}) = \sum_{0 \leq k \leq j } 2^{-2(j - k)}\omega(2^{-k})$,
\item\label{w2} $(W_2\omega)(2^{-j}) =  \sum_{0 \leq k \leq j } 2^{-(j - k)}\omega(2^{-k})  + \sum_{k > j}\omega(2^{-k})$,
\item\label{w3} $(W_3\omega)(2^{-j}) =  \sum_{0 \leq k \leq j } 2^{-2(j - k)}\omega(2^{-k})  + j\sum_{l=1}^\infty l\omega(2^{-jl})$,
\end{enumerate}
and extend $W_i \omega$, for $i = 1, 2, 3$, to $(0, 1]$ by
\[
(W_i  \omega)(\epsilon) = (W_i \omega)(2^{-j}),\,\, \quad 2^{-j - 1} < \epsilon \leq 2^{-j}.
\]

The series \eqref{w2}, \eqref{w3} may diverge, in which case some of the following statements become empty.

\begin{rmk}\label{log-t}  One can check by direct computation that for $\omega_t$ and $\widetilde{\omega}_t$ as in \eqref{defphi}
\begin{enumerate}
\item if $t > 2$, then 
\[
W_1\,\omega_t \leq C_t\,\omega_t, \quad W_2\,\omega_t \leq C_t\,\omega_{t - 1}, \quad W_3\,\omega_t \leq C_t\,\omega_{t - 1}~.
\]
\item for any $0 < t \leq 1$,
\[
W_i \,\widetilde\omega_t \leq C_t\, \widetilde\omega_t,\quad i = 1, 2, 3.
\]
\end{enumerate}
\end{rmk}

Let $\rho$ be a compactly supported probability measure on $\R$ such that
\begin{equation}\label{eq-rho}
\inf_{E\in\R}\int \log_-|E - E'|\mathrm{d}\rho(E') > -\infty.
\end{equation}
Set
\begin{enumerate}
\item $\mathcal N(E) = \mathcal N_\rho(E) = \rho(-\infty, E]$ for $E\in\R$,
\item $\gamma(z) = \gamma_\rho(z) = \int\log|z - E|\mathrm{d}\rho(E)$ for $z\in\C$.
\end{enumerate}
By the Thouless formula, if $\mathcal N(E)$ is the integrated density of states of an ergodic Schr\"odinger operator, then $\gamma$ is the Lyapunov exponent. The regularity properties of $\mathcal N$ can be inferred  from those of $\gamma$ (and vice versa) with the help of

\begin{proposition}\label{prop-gs} Assume $\omega$ is non-decreasing, and $\rho$ is a compactly supported probability measure on $\R$ satisfying (\ref{eq-rho}).
\begin{enumerate}
\item\label{gs-1} If $\mathcal N_\rho\in \uc\,[\omega]$, then the restriction of $\gamma_\rho$ to $\R$ satisfies  $\gamma_\rho|_\R \in \uc\,[W_2\omega]$. 
\item \label{gs-2}If $\gamma_\rho|_\R \in \uc\,[\omega]$, then $\mathcal N_\rho\in \uc\,[W_2\omega]$.
\end{enumerate}
\end{proposition}

Proposition~\ref{prop-gs} and Remark~\ref{log-t} imply:
\begin{corollary}\label{cor:log-t}
For any $t>2$,  we have
\begin{enumerate}
\item\label{gs-1} if  $\mathcal N_\rho \in \uc[\omega_t]$ , then $\gamma_\rho|_\R  \in \uc[\omega_{t-1}]$; 
\item \label{gs-2} if  $\gamma_\rho|_\R \in \uc[\omega_t]$, then $\mathcal N_\rho\in \uc[\omega_{t-1}]$.
\end{enumerate}
\end{corollary}
We do not know whether this corollary is valid for $t \in (1, 2]$. For $t = 1$, one can construct $\rho$ such that $\mathcal N_\rho \in \uc[\omega_1]$ but $\gamma_\rho|_\R$ is discontinuous  \cite{Bochi2,wangyou1}.

\bigskip The next proposition shows the set of $E$ for which $\gamma_\rho(E+i\epsilon)$ is far from $\gamma_\rho(E)$ is very meager.

\begin{proposition}\label{prop-surace}  Assume $\omega$ is non-decreasing with $\lim_{\epsilon\to 0^+}\omega(\epsilon) = 0$.  Suppose  $\rho$ is a compactly supported probability measure on $\R$ satisfying (\ref{eq-rho}) and $\mu$ is a compactly supported probability measure such that $\mathcal N_\mu\in \uc\,[\omega]$. For any $\epsilon\in(0, 1]$ and for any $\xi > 0$
\begin{equation*}
\mu\{E\in\R | \quad |\gamma_\rho (E + i\epsilon) - \gamma_\rho (E)| \geq \xi \} \leq \frac{C_\mu (W_3\omega)(\epsilon)}{\xi},
\end{equation*}
where $C_\mu > 0$ depends only on $\mu$.
\end{proposition}

\subsubsection{Proof of Proposition~\ref{prop-gs}}
The proof of Proposition~\ref{prop-gs} relies on the Littlewood-Paley decomposition (see e.g.\ \cite[Lemma 8.3]{schlag}). We fix a Schwartz function $\phi$ the Fourier transform $\hat\phi$ of which  is compactly supported in $\R\setminus\{ 0\}$, and such that
\begin{equation}\label{eq-phi}
\sum_{j = -\infty}^\infty \hat\phi(2^{-j} x) \equiv 1,
\end{equation}
and for any $x$ at most two terms in (\ref{eq-phi}) are not equal to zero. Let
\[
\phi_j(x) = 2^j\phi(2^j x),
\]
so that for any integrable function $u$
\[
u = \sum_j u*\phi_j~,
\]
where $*$ denotes convolution.

\begin{lemma}\label{lem-gs1} Let $u\in C_0(\R)$ be a continuous compactly supported function. 
\begin{enumerate}
\item\label{lem1-1} If $\omega$ is non-decreasing and $u\in \uc\,[\omega]$, then $\sup\limits_{j \geq 0}\frac{\|u*\phi_j\|_\infty}{(W_1\omega)(2^{-j})} < \infty$.
\item\label{lem1-2} If $\sup\limits_{j \geq 0}\frac{\|u*\phi_j\|_\infty}{\omega(2^{-j})} < \infty$, then $u\in \uc\, [W_2\omega]$.
\end{enumerate}
\end{lemma}
\begin{proof} Using an appropriate smooth partition of unity, we can assume without loss of generality that $\mathrm{diam}\, \mathrm{supp}\, u \leq 1$. Let us prove (\ref{lem1-1}). Observe that $\phi$ is compactly supported, hence for sufficiently large $j$ we have $\operatorname{supp} \phi_j \subset \left[-\frac12, \frac12\right]$. Then
\begin{equation*}
\begin{split}
|(u*\phi_j)(x)| &= \left| \int (u (x - y)  - u(x))\phi_j(y)\mathrm{d}y\right|  \\
& \leq \sum_{k\leq j} \int_{2^{-k-1} < |y| \leq 2^{-k}}  |u (x - y)  - u(x)||\phi_j(y)|\mathrm{d}y 
+ \int_{|y| \leq 2^{{-j - 1}}}|u (x - y)  - u(x)||\phi_j(y)|\mathrm{d}y  \\
& \leq \sum_{k \leq j} \omega(2^{-k}) \int_{2^{-k-1} < |y| \leq 2^{-k}} |\phi_j(y)|\mathrm{d}y + \omega(2^{{-j - 1}}) \int_{|y| \leq 2^{{-j - 1}}}|\phi_j(y)|\mathrm{d}y.
\end{split}
\end{equation*}
Since $\phi$ is a Schwartz function, in particular $|\phi(y)|\leq \frac{C}{|y|^3}$ for some constant $C>0$. Therefore, for any $j\geq 0$ we have $|\phi_j(y)|\leq \frac{C}{2^{2j}|y|^3}$. For $k\leq j$ we obtain
\[
 \int_{2^{-k-1} < |y| \leq 2^{-k}} |\phi_j(y)|\mathrm{d}y \leq 2^{-k}\frac{C}{2^{2j}2^{-3(k+1)}} = 8C 2^{-2(j - k)},
\]
where $2^{-k}$ is the length of the interval $2^{-k-1} < |y| \leq 2^{-k}$. Also,
\[
 \int_{|y| \leq 2^{{-j - 1}}}|\phi_j(y)|\mathrm{d}y = \int_{|y| \leq \frac{1}{2}}|\phi(y)|\mathrm{d}y \leq C.
\]
Now we prove (\ref{lem1-2}). Let us write
\[
u = \sum_{k = -\infty}^\infty u*\phi_k = \left( \sum_{k < 0} + \sum_{k \geq 0}\right)  u*\phi_k \equiv u^- + u^+, 
\quad \text{so that} \quad \widehat{u^-} = {\left(u* \sum_{k < 0}\phi_k  \right)}^\wedge = \hat u\sum_{k < 0}\hat\phi_k.
\]
Note that $\sum_{k < 0}\hat\phi_k$ is compactly supported. By Bernstein's inequality (see, e.g \cite[Proposition 5.2]{wolff}) we obtain
\[
\|(u^-)'\|_\infty = \left\| \left(\sum_{k \geq 0} u*\phi_k\right)'\right\|_\infty \leq C \left\|u* \sum_{k \geq 0}\phi_k \right\|_\infty < \infty.
\]
On the other hand, for $2^{-j-1} \leq |x - y| \leq 2^{-j}$ we have
\[
|u^+(x) - u^+(y)| \leq \sum_{k\geq 0}|(u*\phi_k)(x) - (u*\phi_k)(y)|.
\]
For $k\geq j$ we obtain
\[
|(u*\phi_k)(x) - (u*\phi_k)(y)| \leq 2\|u*\phi_k\|_\infty \leq C_u\omega(2^{-k}).
\]
For $0\leq k\leq j$ we get
\[
\begin{split}
|(u*\phi_k)(x) - (u*\phi_k)(y)| 
&\leq | x - y|  \|(u* \phi_k )'\|_\infty \leq C |x - y|  \|u*\phi_k \|_\infty 2^k \\
&\leq C 2^{-j+k}  \|u*\phi_k \|_\infty \leq \widetilde C_u2^{-(j-k)}\omega(2^{-k}).\qedhere
\end{split}
\]
\end{proof}

\begin{lemma}\label{lem-gs2} Assume that $\omega$ is non-decreasing. Then $W_2W_1\omega \leq 4W_2\omega$.
\end{lemma}
\begin{proof} It is sufficient to estimate $(W_2W_1\omega)(2^{-j})$. This is done as follows:
\[
\begin{split}
(W_2W_1\omega)(2^{-j}) &  = \sum_{k \leq j} 2^{-(j - k)}(W_1\omega)(2^{-k}) + \sum_{k > j} (W_1\omega)(2^{-k}) 
\\& = \sum_{m \leq j} \omega(2^{-m}) 2^{-2(j - m)} \left( \sum_{m\leq k\leq j} 2^{j - k} +  \sum_{k > j} 2^{-2(k - j)}\right)
 + \sum_{m > j} \omega(2^{-m})\sum_{k\geq m} 2^{-2(k - m)}
\\& \leq \sum_{m \leq j }\omega(2^{-m}) 2^{-2(j - m)}(2^{j - m + 1} + 1) + 2\sum_{m > j} \omega(2^{-m})
\\& \leq 4 \left( \sum_{m \leq j }\omega(2^{-m}) 2^{-(j - m)}  + \sum_{m > j} \omega(2^{-m})\right).\qedhere
\end{split}
\]
\end{proof}
\begin{proof}[Proof of Proposition~\ref{prop-gs}] Suppose that $\mathrm{supp}\,\rho\subset [-A , A]$. For $E\notin [-A -1 , A + 1]$
\[
|\gamma'_\rho(E)| = \left| \int \frac{1}{E - E'}\mathrm{d}\,\rho(E')\right | \leq 1.
\]
Thus it suffices to estimate $|\gamma_\rho (E_1) - \gamma_\rho(E_2)|$ for $E_1, E_2 \in [-A - 1, A + 1]$. Choose a smooth compactly supported function $\chi\in C^\infty_0$ so that $\chi \geq 0$ and the restriction of $\chi$ to the interval $[-A -2, A + 2]$:  $\chi|_{[-A - 2, A + 2]} \equiv 1$. Denote by
\[
u_1 = \mathcal N_\rho \chi,\quad u_2 = \mathcal N_\rho (1 - \chi),
\]
and let us rewrite
\[
\begin{split}
\gamma_\rho(E) & = \int \log |E - E'| \mathrm{d}\,\rho(E') 
\\&=  \int \log |E - E'| \mathrm{du_1}\,(E') + \int \log |E - E'| \mathrm{du_2}\,(E') \equiv v_1(E) + v_2 (E).
\end{split}
\]
Then the derivative $v'_2$ is bounded in $[-A - 1, A + 1]$, whereas $u_1\in C_0(\R)$ and
\[
v'_1(E) = \int \frac{1}{E - E'} \mathrm{du_1}(E') = \widetilde u_1(E)
\]
is (up to a constant) the Hilbert transform of $u_1$. By (\ref{lem1-1}) of Lemma~\ref{lem-gs1}, we have for any $j\geq 0$
\[
\|u_1 *\phi_j \|_\infty \leq C(W_1\omega)(2^{-j}).
\]
Let $\Phi$ be another Schwartz function satisfying (\ref{eq-phi}) and $\widehat\Phi|_{\mathrm{supp}\,\widehat \phi} \equiv 1$. Define $\Phi_j(x ) = 2^j\Phi(2^j x)$ satisfying $\Phi_j * \phi_j = \phi_j$ for all $j$. Then we have
\begin{equation}\label{eq:tilde-u1}
\|\widetilde u_1 * \phi_j \|_\infty = \|\widetilde u_1 * \Phi_j * \phi_j \|_\infty = \|\widetilde \Phi_j * u _ 1 * \phi_j \|_\infty \leq \|\widetilde \Phi_j \|_1 \| u*\phi_j\|_\infty.
\end{equation}
Observe that
\[
\mathrm{supp\,}\widehat{\widetilde\Phi_j} \subset \mathrm{supp\,}\widehat\Phi_j \subset [-C 2^j, C 2^j]\setminus[-c 2^j, c 2^j]~.
\]
 Then
\[
\| {\widetilde\Phi_j} \|_1 \leq C 2^{-\frac j 2}\|{\widetilde\Phi_j} \|_2 = C'  2^{- \frac j 2}\|\Phi_j \|_2 \leq C'' \|\Phi_j \|_1 \leq C''',
\]
where the first inequality follows from \cite[Inequality after Exercise 6.7]{schlag}, the equality holds since the Hilbert transform is an isometry in $L_2$ (up to a constant), and the second inequality follows from \cite[Lemma 6.12]{schlag}. The last inequality holds since by the definition of $\Phi_j(x)$ we get
\[
\int |\Phi_j| \mathrm{d}x = \int |\Phi| \mathrm{d}x \leq \operatorname{Const}.
\]
Therefore, by (\ref{lem1-2}) of Lemma~\ref{lem-gs1} and by Lemma~\ref{lem-gs2} we conclude from (\ref{eq:tilde-u1}) that 
\[ \widetilde u_1 \in \uc[W_2W_1\omega]\subset \uc[W_2\omega]~.\]
The second statement (\ref{gs-2}) is proved by the same argument.
\end{proof}

\subsubsection{Proof of Proposition~\ref{prop-surace}}
We need the following lemma.

\begin{lemma}\label{lem-3} Assume that $\omega$ is non-decreasing and let $\mu$ be a compactly supported probability measure such that $\mathcal N_\mu \in \uc\,[\omega]$. Then, for any $0 < \epsilon \leq 1$
\begin{equation}\label{eq:lem3}
\sup_{E'\in\R}\int_{\R}\log\left( 1 + \frac{\epsilon^2}{(E - E')^2}\right) \mathrm{d}\mu(E)  \leq C_\mu(W_3\omega)(\epsilon),
\end{equation}
where $C_\mu > 0$ is a constant that depends only on $\mu$. 
\end{lemma}
\begin{proof} Let $\mu = c_1\mu_1 + c_2\mu_2 + \cdots + c_k\mu_k$, $1\leq k\leq \infty$, where $\mathrm{diam\,supp\,}\mu_k \leq 1$ and $c_k \geq 0$ for every $k$, with $\sum_k c_k \equiv 1$. Then,
\[
\int\log\left( 1 + \frac{\epsilon^2}{(E - E')^2}\right) \mathrm{d}\mu(E)  = \sum_{k\geq 1} c_k \int \log\left( 1 + \frac{\epsilon^2}{(E - E')^2}\right) \mathrm{d}\mu_k(E). 
\]
Therefore, it suffices to prove (\ref{eq:lem3}) for measures with $\mathrm{diam\,supp\,}\mu \leq 1$. We can also without loss of generality let $E' = 0$ (since the assumption is invariant under shifts) and that $\epsilon = 2^{-j}$. Then we have
\[
\begin{split}
\int\log\left( 1 + \frac{\epsilon^2}{E^2}\right)\mathrm{d}\mu(E)  & \leq 2\int_{|E|\leq 2^{-j}}\log\frac{2^{-j}}{|E|}\mathrm{d}\mu(E) + \int_{2^{-j}\leq |E|\leq 1} \frac{2^{-2j}}{E^2}\mathrm{d}\mu(E) \\
&\vspace{-1cm}\leq C\sum_{l = 1}^\infty \int_{2^{-j(l + 1)}\leq |E| \leq 2^{-jl}} jl \mathrm{d}\mu(E) 
 + C\sum_{k = 0}^{j - 1} \int_{2^{-k-1}\leq |E| \leq 2^{-k}} 2^{-2(j - k)} \mathrm{d}\mu(E)
\\&\vspace{-1cm} \leq C_\mu\left( \sum_{l = 1}^\infty  jl\,\, \omega(2^{-jl}) + \sum_{k = 0}^{j - 1} 2^{-2(j - k)}\omega(2^{-k})\right) \leq C_\mu (W_3\omega)(2^{-j}).\qedhere
\end{split}
\]
\end{proof}

\begin{proof}[Proof of Proposition~\ref{prop-surace}] 

Denote  $\Lambda=\left\{ E\in\R\,\, |\,\, |\gamma_\rho(E + i\epsilon) - \gamma_\rho(E)| \geq \xi \right\} $, then 
by Chebyshev's inequality and Fubini's theorem,  we obtain
\[
\begin{split}
\mu(\Lambda)
&\leq \frac{1}{\xi}\int_{\Lambda}  |\gamma_\rho(E + i\epsilon) - \gamma_\rho(E)|  \mathrm{d}\mu(E) 
\\&\leq \frac{1}{2\xi} \int_{\Lambda} \mathrm{d}\mu(E) \int_{\R} \log\left( 1 + \frac{\epsilon^2}{(E - E')^2}\right)  \mathrm{d}\rho(E')
\\& = \frac{1}{2\xi} \int_{\R} \mathrm{d}\rho(E') \int_{\Lambda}\log\left( 1 + \frac{\epsilon^2}{(E - E')^2}\right)  \mathrm{d}\mu(E) 
\\&\leq \frac{1}{2\xi}\sup_{E'} \int_{\R}  \log\left( 1 + \frac{\epsilon^2}{(E - E')^2}\right)\mathrm{d}\mu(E). 
\end{split}
\]
Lemma~\ref{lem-3} now implies that
\[
\mu\left\{ E\in\R\,\, |\,\, |\gamma_\rho(E + i\epsilon) - \gamma_\rho(E)| \geq \xi \right\} \leq \frac{1}{2\xi} C_\mu (W_3\omega)(\epsilon).
\qedhere\]
\end{proof}

\subsection{Continuity estimates}\label{s:cont} Let $H_{\alpha,\theta}$ be a one-frequency operator:
\begin{equation}\label{eq:qp1}
(H_{\alpha, \theta}\psi)(n) = \psi(n+1) + \psi(n - 1) + \phi(2\pi\alpha n + \theta)\psi(n)~,
\end{equation}
where $\phi: \mathbb R / \mathbb Z \to \mathbb R$ satisfies a uniform Lipschitz condition
\[ \| \phi\|_{\mathrm{Lip}} = \max\frac{|\phi (x) - \phi(y)|}{|x - y|} < \infty~.\]

The following result due to Avron--van  Mouche--Simon \cite{simon_vanmouche_avron} shows that the set $S(\alpha) = \bigcup_\theta \sigma(H_{\alpha,\theta})$, where $\sigma(H_{\alpha,\theta})$ is the spectrum of $H_{\alpha,\theta}$, 
depends  continuously on $\alpha$, with a quantitative estimate. 
\begin{proposition}[Avron--van Mouche--Simon \cite{simon_vanmouche_avron}]\label{continuity_of_spectra}
Let $H_{\alpha,\theta}$ be as in (\ref{eq:qp1}). If $|\alpha - \alpha^\prime|$ is sufficiently small, then 
\[ \dist(S(\alpha^\prime), S(\alpha)) \leq 6( \| \phi\|_{\mathrm{Lip}} |\alpha - \alpha^\prime|)^{\frac{1}{2}}~, \]
where the Hausdorff distance between two sets $A, B \subset \R$ is defined as
\[ \dist(A, B) = \max(\sup_{E \in A} \dis(E, B), \sup_{E \in B} \dis(E, A)) = 
\max( \sup_{E \in A} \inf_{E' \in B} |E - E'|,  \sup_{E \in B} \inf_{E' \in A} |E - E'|)~.\]
\end{proposition}

Our next proposition, possibly of independent interest, is a counterpart of this fact for the integrated density of states; it quantifies a result of Avron and Simon \cite[Theorem 3.3]{simon1}. We mention similar-looking results for random Schr\"odinger operators recently obtained in \cite{hislop-marx, shamis}.

Let 
\[ \overline{\mathcal N}_{\alpha}(E) = \int d\theta\, \mathcal N_{\alpha,\theta}(E) \]
be the $\theta$-averaged integrated density of states, (where $\mathcal N_{\alpha,\theta}$ is the integrated density of states corresponding to $H_{\alpha,\theta}$), and let $\overline\rho_{\alpha}$ be the corresponding measure,
\begin{equation}\label{eq:defbarrho}\overline\rho_{\alpha}(E_1, E_2] = \overline{\mathcal N}_{\alpha}(E_2)-\overline{\mathcal N}_{\alpha}(E_1)~.\end{equation}


\begin{proposition}\label{th:ky-fan} Let $H_{\alpha,\theta}$ be as in (\ref{eq:qp1}). For any $[r_-, r_+] \subset \R$,\ $\alpha, \alpha'\in\R$,\ $L \geq 1$, if
\begin{equation}\label{eq:ky-fan-as}\kappa \geq 2\pi \|\phi\|_{\mathrm{Lip}}|\alpha' - \alpha|L,\end{equation} then we have
\begin{equation}\label{eq:ky-fan}
\overline\rho_{\alpha'}[r_-, r_+] \geq \overline\rho_\alpha[r_- +\kappa, r_+ - \kappa] - \frac{4}{L}~.
\end{equation}
\end{proposition}

\begin{rmk} It would be interesting to know whether this result could be significantly improved in the regime of positive Lyapunov exponent (without any further assumptions on $\alpha$).
\end{rmk}

\begin{proof}[Proof of Proposition~\ref{th:ky-fan}] Let $H^L_{\alpha, \theta}$ be the restriction of the operator $H_{\alpha, \theta}$ to a finite box $\{0, \dots, L - 1\}$. For any $\alpha', \alpha \in\R$ we have
\begin{equation*}
\|H^L_{\alpha', \theta} - H^L_{\alpha, \theta}\|_{\mathrm{op}} \leq 2\pi \|\phi\|_{\mathrm{Lip}}|\alpha' - \alpha|L \leq \kappa~.
\end{equation*}
First, let us show that for any interval $I\in\R$
\begin{equation}\label{eq:est1}
\left| \overline\rho_\alpha(I) - \int d\theta \frac{1}{L}\mathrm{Tr\,}\mathbf{1}_I(H^L_{\alpha, \theta})\right| \leq \frac{2}{L}~.
\end{equation}

\noindent By definition we have
\[
 \overline\rho_\alpha(I)  = \lim_{k\to\infty}\frac{1}{kL}\int  d\theta \frac{1}{L}\mathrm{Tr\,}\mathbf{1}_I(H^{kL}_{\alpha, \theta})~.
\]
The matrix
\[
H^{kL}_{\alpha, \theta} - \bigoplus_{j = 0}^{k - 1} H^{L}_{\alpha, \theta + jL\alpha}
\]
is a sum of $k - 1$ matrices of rank $2$, thus
\[
\mathrm{rank}\left[H^{kL}_{\alpha, \theta} - \bigoplus_{j = 0}^{k - 1} H^{L}_{\alpha, \theta + jL\alpha} \right] \leq 2(k - 1) \leq 2k~.
\]
Therefore, using the interlacing property, we obtain
\[
\left|\frac{1}{kL} \mathrm{Tr\,}\mathbf{1}_I(H^{kL}_{\alpha, \theta}) - \frac{1}{k}\sum_{j = 0}^{k - 1}\frac{1}{L} \mathrm{Tr\,}\mathbf{1}_I(H^{L}_{\alpha, \theta+ j L \alpha})\right| \leq \frac{2}{L}~.
\]
Integrating over $\theta$ and taking $k\to\infty$ gives (\ref{eq:est1}).

 For an $L\times L$ matrix $A$ and an interval $I\in\R$ denote  
\[
n(A; I) = \mathrm{Tr\,}\mathbf{1}_I (A)~.
\]
Then, for any $L\times L$ matrix $\widetilde A$ with $\| \widetilde A - A\|_{\mathrm{op}} \leq \kappa$ for some $\kappa > 0$
\[
n(\widetilde A; [r_-, r_+]) \geq  n( A; [r_- + \kappa, r_+ - \kappa])~,
\]
since the eigenvalues of $A$ can be shifted by (at most) distance $\kappa$. Thus   if \eqref{eq:ky-fan-as} holds, then  we obtain

\[
\int d\theta \frac{1}{L} \left[ n(H^{L}_{\alpha', \theta} ; [r_-, r_+]) - n(H^{L}_{\alpha, \theta} ; [r_- + \kappa, r_+ - \kappa] )\right]\geq 0~.
\]
Therefore, using (\ref{eq:est1}) we get
\[
\begin{split}
&\overline\rho_{\alpha'}[r_-, r_+]  - \overline\rho_\alpha[r_- +\kappa, r_+ - \kappa]  = \overline\rho_{\alpha'}[r_-, r_+]  - \int d\theta \frac{1}{L} n(H^{L}_{\alpha', \theta} ; [r_-, r_+]) \\& + \int d\theta \frac{1}{L} \left[ n(H^{L}_{\alpha', \theta} ; [r_-, r_+]) - n(H^{L}_{\alpha, \theta} ; [r_- + \kappa, r_+ - \kappa] )\right] \\& + \int d\theta \frac{1}{L} n(H^{L}_{\alpha, \theta} ; [r_- + \kappa, r_+ - \kappa]) -  \overline\rho_\alpha[r_- +\kappa, r_+ - \kappa] \geq -\frac{2}{L} + 0 -\frac{2}{L}~,
\end{split}
\]
namely \eqref{eq:ky-fan} holds.
\end{proof}

\section{The main estimate}

The proofs of  Theorems~\ref{thm:hausdorff}-\ref{thm:homog} rely on the following propositions, which are improvements and generalizations of \cite[Theorem 2]{LS}.  Recall that the spectrum of the periodic almost Mathieu operator $H_{\frac{p}{q}, \lambda, \theta}$  depends on $\theta$ and denote
\begin{equation*}
\specUq = \bigcup_{\theta}\sigma\left(\frac{p}{q}, \lambda, \theta\right) \text{,} \quad
\specIq = \bigcap_{\theta}\sigma\left(\frac{p}{q}, \lambda, \theta\right).
\end{equation*} 
These sets are also well defined in the case that $\alpha\in\mathbb{R\setminus{Q}}$,\ in which we have
\[
S(\alpha, \lambda) = S_{-}(\alpha, \lambda) = \sigma(\alpha, \lambda, \theta)
\]
 (consistently with the notation in the introduction). Propositions~\ref{prop:LS}, \ref{prop:LS1} and their Corollaries~\ref{cor:LS} and \ref{cor:LS1} show that for $\frac{\widetilde p}{\widetilde q}$ sufficiently close to $\frac p q$, $H_{\frac{\widetilde p}{\widetilde q},\lambda,\theta}$ has very little spectrum outside $\specIq$. The difference between the two propositions is that the latter one has stronger assumptions ($|\frac{p}{q} -  \frac{\widetilde p}{\widetilde q}| = \frac{1}{q\widetilde q}$), and  consequently  a stronger conclusion.

To state the propositions precisely, we need a polynomial  $\Discr$ (called the discriminant) which is defined in (\ref{chamber_formula}) below, and satisfies:
\[
\sigma\left(\frac{p}{q}, \lambda, \theta\right) = \left\{ E\in\R\, \big| \left| \Discr - 2\lambda^q\cos(\theta\, q)\right| \leq 2\right\}~.
\]
As we recall below, $\specIq = \varnothing$ for $\lambda > 1$, and 
\begin{equation}\label{eq:discr-on-Sminus}
\specIq = \left\{ E: |\Discr| \leq
2 - 2\lambda^q\right\}
\end{equation}
if $\lambda \leq 1$. This leads us to define, for $\lambda \leq 1$, a family of sets $J_\delta \subset \mathbb R \setminus \specIq$:
\begin{equation}\label{eq:jdelta}
J_\delta = \left\{ E\in\R \, \big| \,  |\Discr| >  2 - 2\lambda^q + \delta \right\}~, \quad \delta > 0~,
\end{equation}
which consist of energies that lie outside the spectrum $\sigma\left(\frac p q, \lambda, \theta\right)$  for
a sizeable fraction of $\theta$.

\begin{proposition}\label{prop:LS}
Let $0 < \lambda \leq 1$ and $\alpha \in \R\setminus\Q$. For any $r > 0$ there exist $q_0 = q_0(r, \alpha)\in\N$ and $c_0 = c_0(r, \alpha) > 0$ such that the following holds.
Assume that $\frac{p}{q}$, $\frac{\widetilde p}{\widetilde q}\in\Q$, $\delta\in (0,1)$ are such that
\begin{enumerate}
\item\label{eq:cond1} $\left|\frac{p}{q} - \alpha\right |, \left|\frac{\widetilde p}{\widetilde q} - \alpha\right | < \delta e^{-rq}$~,
\item\label{eq:cond2} $\widetilde q > q \geq q_0$~,
\item \label{eq:cond3} $\exp(-e^{\frac{rq}{2}} )< \delta  \leq c_0(r, \alpha)\, q^2 \lambda^q$~.
\end{enumerate}
Then, for any $E\in J_\delta$ and for any $\epsilon \geq \exp\left(-\frac{e^{rq}}{15000 q^2 \lambda^{\frac q 2}}\right)$, the Lyapunov exponent,  that corresponds to $H_{\frac{\widetilde{p}}{\widetilde{q}}, \lambda, \theta}$ obeys for any $\theta\in[0, 2\pi)$
\begin{equation}\label{eq:propLS}
\gamma_{\frac{\widetilde{p}}{\widetilde{q}}, \lambda, \theta} (E + i\epsilon) \geq \frac{\delta}{9600 q^2 \lambda^{\frac q 2}}~.
\end{equation}
\end{proposition}

Proposition~\ref{prop:LS}, Proposition~\ref{prop-gs}, and the Thouless formula \eqref{eq:thouless} will imply (see Section~\ref{sec:pfLS}):

\begin{corollary}\label{cor:LS} Let $\omega: (0, 1]\rightarrow\R_+$ be non-decreasing function such that $\lim_{\epsilon\rightarrow 0^+}\omega(\epsilon) = 0$, and let $W_3\,\omega$ be as in Proposition~\ref{prop-gs}. Let $\mu$ be a probability measure on $\R$ such that 
\[
\mu(a, b) \leq C\,\omega(b - a), \quad b-a \leq 1~.
\]
Then, in the setting of Proposition~\ref{prop:LS} we have
\begin{equation*}
\mu\left(\sqnp\bigcap J_\delta\right) \leq \frac{C_\mu\, q^2\, \lambda^{\frac{q}{2}}}{\delta}(W_3\,\omega)\left[  \exp\left(-\frac{e^{rq}}{15000 q^2 \lambda^{\frac q 2}}\right)\right]~.
\end{equation*}
\end{corollary}

If $\frac{\widetilde p}{\widetilde q}$ and $\frac{p}{q}$ are successive convergents of $\alpha$, one can significantly relax the restriction on $\epsilon$.

\begin{proposition}\label{prop:LS1}
Let $0 < \lambda \leq 1$, $\alpha \in \R\setminus\Q$ and let $\left\{\frac{p_n}{q_n}\right\}$ be the sequence of convergents of $\alpha$. For any $r > 0$ and $\delta\in(0,1 )$ there exists $n_0 = n_0(r, \alpha)$ such that the following holds. If, for some $n\geq n_0$,
\begin{equation}\label{eq:cond-ls1}
 q_{n+1}\delta > e^{r q_n},
\end{equation}
then for any $E\in J_\delta$, $\epsilon \geq \exp\left(-\frac{c\, q_{n + 1}\, \delta}{q_n}\right)$, and $\theta\in[0, 2\pi)$
\begin{equation*}
\gamma_{\frac{p_{n + 1}}{q_{n + 1}}, \lambda, \theta} (E + i\epsilon) \geq \frac{c\, \delta}{q_n^2}~.
\end{equation*}
\end{proposition}

Similarly to  Corollary \ref{cor:LS},  the combination of Proposition~\ref{prop:LS1}, Proposition~\ref{prop-gs}, and the Thouless formula \eqref{eq:thouless} implies the following: 

\begin{corollary}\label{cor:LS1}  In the setting of Proposition~\ref{prop:LS1} we have
\begin{equation*}
\left|S\left( \frac{p_{n + 1}}{q_{n + 1}}, \lambda\right)\bigcap J_\delta\right| \leq \frac{C\, q_n^2}{\delta} \exp\left(-\frac{q_{n + 1}\delta}{C_1\, q_n}\right)~,
\end{equation*}
where $C_1 > 0$ is a constant and $|\cdot|$ denotes the Lebesgue measure.
\end{corollary}

\subsection{An avalanche estimate}\label{sub:tech} One of the important new ingredients in the proofs of Proposition~\ref{prop:LS} and Proposition~\ref{prop:LS1}   is the following avalanche-type proposition (already stated in the introduction).

\begin{proposition}\label{lem:avalanche} For every $0 < c < 1$ the following holds for sufficiently small $0 < b < b_0(c)$. Let $A_j \in SL_2(\C)$ and let $0 < \delta_j < b$, $1 \leq j \leq n$, be such that
\begin{enumerate}\label{cond-avalanche}
\item\label{avalanche-cond1} $\delta_{j + 1} \leq \delta_j + b\delta^{\frac32}_j, \,\,\, 1 \leq j\leq n - 1$,
\item\label{avalanche-cond2} $\|A_{j + 1} - A_j \|\leq b\delta_j, \,\,\, 1 \leq j\leq n - 1$,
\item\label{avalanche-cond3} $|\mathrm{Tr}\, A_j| \geq 2+ (1 - b)\delta_j, \,\,\, 1 \leq j \leq n$,
\end{enumerate}
Then, for any vector $u_0 \in \C^2$ such that $\|A_1 u_0 \| \geq \exp((1 - c) \sqrt{\delta_1})\| u_0\|$,
\begin{equation*}
\|A_n\cdots A_1 u_0\| \geq \exp\left((1 - c)\sum_{j = 1}^n \sqrt{\delta_j}\right)\| u_0\|.
\end{equation*}
\end{proposition}
\begin{rmk}
One can check that $b_0(1/2) \geq 1/10$.
\end{rmk}

 \begin{proof} Let $u_0\in\C^2$ be a unit vector such that
\begin{equation}\label{eq:initial-cond}
\|A_1 u_0\| \geq \exp((1 - c)\sqrt\delta_1).
\end{equation}
Denote $B_j = A_j\cdots A_1,\,\, 1 \leq j\leq n$. It is sufficient to show that for every $1 \leq j\leq n$
\begin{equation}\label{eq:conclusion}
\|B_j u_0\| \geq \exp((1 - c)\sqrt\delta_j) \|B_{j - 1} u_0\|.
\end{equation}
By (\ref{eq:initial-cond}) we obtain (\ref{eq:conclusion}) for $j = 1$. Assume that (\ref{eq:conclusion}) holds for some $1 \leq j \leq n - 1$. Then, we get
\begin{equation}\label{eq:est0}
\begin{split}
\exp(-(1 - c)\sqrt\delta_j)\|B_j u_0\| \geq \|B_{j - 1} u_0\| &= \| A^{-1}_j A_j A_{j - 1}\cdots A_1 u_0\| \\&= \|A^{-1}_jB_j u_0\|.
\end{split}
\end{equation}
Since $\det A_j = 1$ for every $1 \leq j\leq n$, we obtain
\begin{equation*}
A_j + A^{-1}_j = (\mathrm{Tr}\, A_j)\mathrm{Id}_{2\times 2}.
\end{equation*}
Therefore, we get
\[
\begin{split}
\| B_{j + 1} u_0\| &= \|A_{j + 1} B_j u_0 \| = \| [(A_{j + 1} - A_j) + (A_j + A^{-1}_j) - A^{-1}_j]B_j u_0 \| \\ & \quad\geq \|(A_j + A^{-1}_j)B_j u_0 \| - \|A_{j + 1} - A_j \|\|B_j u_0\| - \| A^{-1}_jB_j u_0\|\\ & \quad\geq  (2 + (1 - b)\delta_j - b\delta_j - \exp(-(1 - c)\sqrt\delta_j)) \|B_j u_0\|,
\end{split}
\]
where the last inequality follows from conditions (\ref{avalanche-cond2}), (\ref{avalanche-cond3}) and from (\ref{eq:est0}). Therefore, to show (\ref{eq:conclusion}) for $j + 1$, it is sufficient to show that
\[
2 + (1 - 2b)\delta_j - \exp(-(1 - c)\sqrt\delta_j) \geq \exp((1 - c)\sqrt{\delta_{j + 1}}).
\]
This is equivalent to showing that
\[
(1 - 2b)\delta_j  \geq (e^{(1 - c)\sqrt\delta_j} + e^{-(1 - c)\sqrt\delta_j} - 2) + (e^{(1 - c)\sqrt{\delta_{j + 1}}} - e^{(1 - c)\sqrt\delta_j}).
\]
Since $0 < 1 - c < 1$, we obtain for sufficiently small $b> 0$
\begin{equation*}
\begin{split}
e^{(1 - c)\sqrt\delta_j} + e^{-(1 - c)\sqrt\delta_j} - 2 &\leq\,  2\left(1 + \frac{(1 - c)^2}{2}\delta_j + \frac{(1 - c)^4}{4!}\delta^2_j\cosh(\widetilde c) - 1\right) \\ &< ((1 - c)^2 + b)\delta_j,
\end{split}
\end{equation*}
where the last inequality holds since $\widetilde c > 0$ is sufficiently small, therefore, $\cosh(\widetilde c) < 2$ and since by our assumption $0 < \delta_j < b$. By assumption (\ref{avalanche-cond1}) we get
\[
\sqrt{\delta_{j + 1}} \leq \sqrt{\delta_j(1 + b\sqrt{\delta_j})} \leq \sqrt{\delta_j} \left(1 + \frac{b}{2}\sqrt{\delta_j}\right),
\]
namely,
\begin{equation}\label{eq:est3}
\sqrt{\delta_{j + 1}} - \sqrt{\delta_{j}} \leq \frac{b}{2}\delta_j.
\end{equation}
Using (\ref{eq:est3}) we obtain for $\delta_j \leq \xi \leq \delta_{j + 1}$
\[
e^{(1 - c)\sqrt{\delta_{j + 1}}} - e^{(1 - c)\sqrt\delta_j} \leq e^{(1 - c)\sqrt\xi}\left(\sqrt{\delta_{j + 1}} - \sqrt{\delta_{j}}\right)\leq e^{(1 - c)\sqrt\xi} \frac{b}{2}\delta_j \leq b\delta_j,
\]
where the last inequality holds since for sufficiently small $\xi > 0$ we have $ e^{(1 - c)\sqrt\xi} \leq 2$. Therefore, since $(1 - 2b)\delta_j \geq ((1 - c)^2 + b)\delta_j$, namely, $(1 - c)^2 + 4b < 1$ for sufficiently small $b$, we conclude the proof.
\end{proof}

\subsection{Preliminaries}
In this section, we introduce the notation and collect several (mostly) known facts which we use in the proof of Propositions~\ref{prop:LS} and \ref{prop:LS1}.
 
\subsubsection{Geometric resolvent expansions, and the Combes--Thomas estimate} Here we discuss several claims that hold true in the following general setting. Let $H$ be one-dimensional discrete Schr\"{o}dinger operator of the form \eqref{eq:general-op}. Define  the restriction $H^{[a, b]}$ of $H$ to an interval $[a, b]\subset \Z$ as follows.
\begin{equation*}
\left( H^{[a, b]} \psi\right) (n) = \phi(n)~,
\end{equation*}
where
\begin{equation*}
\phi(n) =\left\{ \begin{array}{ll}
\psi(n + 1) + \psi (n - 1) +V(n) & \textrm{\ if\ } a < n < b,\\
\psi(a + 1) + V(a)  & \textrm{\ if\ }  n = a > -\infty,\\
\psi(b - 1) + V(b)  & \textrm{\ if\ }  n = b < +\infty~.
\end{array}\right.
\end{equation*}
Define the corresponding Green functions via
\begin{equation}\label{eq:restriction-green}
G^{[a, b]} (n, m; z) = \langle \delta_m, (H^{[a, b]} - z)^{-1}\delta_n\rangle,
\end{equation}
where $\delta_j$ is the vector in $\ell^2(\Z)$, with entries $\delta_j(n)$ that are $1$ if $n = j$ and $0$ otherwise. We write $G^{[a, b]} (n, m; z) \equiv G^{[a, b]} (n, m)$. Then we have (see, e.g.  \cite[Theorem 5.20]{combes-thomas}).
\begin{proposition} \label{gre} Let $a, b, c \in\Z,\,\, a \leq c \leq b \leq \infty$, $\;z\not\in\big(\sigma(H^{[a, b]})\bigcup\sigma(H^{[a, c]})\big)$. Then, we have the following.
\begin{enumerate}
\item If  $a \leq m \leq c < n \leq b$, then 
 \begin{eqnarray}\label{gre1}
G^{[a, b]}(m,n)
 &=& G^{[a, c]}(m,c)G^{[a, b]}( c + 1,n)~, \\
\label{gre1'}
G^{[a, b]}(n,m)
 &=& G^{[a, b]}(n,c)G^{[c + 1, b]}( c + 1, m)~.
\end{eqnarray}
\item If $a \leq m,\, n \leq c$, then
\begin{eqnarray}\label{gre2} 
G^{[a, b]}(m, n)- G^{[a, c]}(m, n) = G^{[a, b]}(m, c + 1)G^{[a, c]}(c, n)~.
\end{eqnarray}
\end{enumerate}
\end{proposition}

\begin{claim}\label{cl:green-decomp}
Let 
$$-\infty \leq a\leq i < c_1< c_2< \cdots c_M <j \leq b\leq \infty~.$$
Then
\begin{eqnarray*}
G^{[a,b]} (i,j) =&& G^{[a,c_1]}(i,c_1) G^{[a,b]}(c_1+1,c_2)
\\&& G^{[c_2+1,c_3]}(c_2+1,c_3) G^{[c_2+1,b]}(c_3+1,c_4) \cdots \\
&& G^{[c_{M-1}+1,c_M]}(c_{M-1}+1,c_M) G^{[c_{M-1}+1,b]}(c_{M}+1,j),   
\end{eqnarray*}
where $G^{[\cdot, \cdot]} (\cdot, \cdot)$ is the corresponding Green's function defined in (\ref{eq:restriction-green}).
\end{claim}
\begin{proof}
By \eqref{gre1} of Proposition~\ref{gre}, we have 
\begin{equation}\label{ga1}
G^{[a,b]} (i,j)=G^{[a,c]} (i,c)G^{[a,b]} (c+1,j)~.
\end{equation}
Now deduce from \eqref{gre1'} of Proposition~\ref{gre}
\begin{equation}\label{ga2}
G^{[a,b]} (i,j)=G^{[a,b]} (i,c)G^{[c+1,b]} (c+1,j).
\end{equation}
Using \eqref{ga1}, we obtain 
\begin{equation}\label{ga3}
G^{[a,b]} (i,j)=G^{[a,c_1]} (i,c_1)G^{[a,b]} (c_1+1,j),
\end{equation}
and using \eqref{ga2}, we obtain 
\begin{equation}\label{ga4}
G^{[a,b]} (c_1+1,j)=G^{[a,b]} (c_1+1,c_2)G^{[c_2+1,b]} (c_2+1,j).
\end{equation}
An additional use of \eqref{ga1} gives 
\begin{equation}\label{ga5}
G^{[c_2+1,b]} (c_2+1,j)=G^{[c_2+1,c_3]} (c_2+1,c_3)G^{[c_2+1,b]} (c_3+1,j).
\end{equation}
Plugging \eqref{ga5} into \eqref{ga4} and the resulting expression into \eqref{ga3}, we obtain
\begin{eqnarray*}
G^{[a,b]} (i,j) =&& G^{[a,c_1]}(i,c_1) G^{[a,b]}(c_1+1,c_2)
\\&& G^{[c_2+1,c_3]}(c_2+1,c_3) G^{[c_2+1,b]}(c_3+1,j)~.
\end{eqnarray*}
Continuing this way (now expand $G^{[c_2+1,b]}(c_3+1,j)$) and so on, we obtain the desired claim. 
\end{proof}

Finally, we recall the Combes--Thomas estimate (cf.\ \cite[Theorem 11.2]{combes-thomas}) :
\begin{proposition}[Combes--Thomas] There exists a numerical constant $c > 0$ such that for  any $z \in \mathbb C \setminus \sigma(H)$ and any $m, n \in \mathbb Z$, the Green function $G$ of $H$ satisfies
\begin{equation}\label{eq:ct1}
|G(n,m; z)| \leq \frac{1}{\kappa} \exp(- c \kappa |m - n|)~, \quad
\text{where} \quad \kappa = \min(1, \operatorname{dist}(z, \sigma(H)))~.
\end{equation}
\end{proposition}

\begin{rmk}For the  case of ergodic Schr\"odinger operators, (\ref{eq:ct1}) and the general relation 
\[ \gamma(z) = - \lim_{n \to \infty} \mathbb E_\varpi \frac{1}{n} \log |G(1, n; z)|~, \quad z \notin \sigma(H)\] 
(or (\ref{eq:LE-green}) below) imply:
\begin{equation}\label{eq:ct}
\gamma(z) \geq c \min(1, \operatorname{dist}(z, \sigma(H)))~.
\end{equation}
\end{rmk}

\subsubsection{The periodic almost Mathieu operator} Consider the almost Mathieu operator $\op$,\ where
$\alpha = \frac{p}{q}$,\ $p, q\in\N$,\ and assume that $p < q$\ are relatively prime. For any $j,\, n \geq 1$ let
\begin{multline}\label{eq:transfer-matrix}
\begin{split}
&\Phi_n(E, \alpha, \theta) = T_n(E) \cdots T_2(E) T_1(E)~, \qquad
T_j(E) = T_j(E, \alpha, \theta) = 
\left(\begin{array}{cc}
E - 2\lambda\cos(2\pi\alpha j + \theta) & -1 \\
  1      &   0  \\
\end{array}\right)~
\end{split}
\end{multline}
be the $n$-step and the one-step transfer matrices respectively.  Denote by
\begin{equation}\label{eq:Dpq}
D_{\frac{p}{q}, \lambda, \theta}(E) = \mathrm{Tr}\, \left(\Phi_q\left(E, \frac{p}{q}, \theta\right)\right),
\end{equation}
the trace of a one-period transfer matrix. It is a polynomial of degree $q$ that has $q$ real simple zeros. The $\theta$\ dependence of $D_{\frac{p}{q}, \lambda, \theta}(E)$ is described by the following formula, due to Chambers \cite{chembers}.

\begin{proposition}[Chambers]\label{prop1}
If $p$,\ $q$\ are relatively prime, then:
\begin{equation}\label{chamber_formula}
D_{\frac{p}{q}, \lambda, \theta}(E) = \Discr -
2\lambda^q\cos{\theta{q}},
\end{equation}
where $\Discr\equiv D_{\frac{p}{q}, \lambda,\frac{\pi}{2q}}(E)$.
\end{proposition}
\noindent From this formula 
\begin{equation}\label{eq:discr-on-S}
\specUq= \left\{ E : |\Discr|
\leq 2 + 2\lambda^q \right\}.
\end{equation}
\noindent Moreover, one can see from (\ref{chamber_formula}) that if
$\lambda > 1$\ then $\specIq = \varnothing$,\ and
if $\lambda \leq 1$\ then
\begin{equation}\label{eq:discr-on-Sminus'}
\specIq = \left\{ E: |\Discr| \leq
2 - 2\lambda^q\right\}
\end{equation}
as we claimed in (\ref{eq:discr-on-Sminus}). 
From the result of Choi, Elliott, and Yui~\cite{CEY} the set $\specUq$ is the union of $q$ closed intervals (bands), that may intersect only at the edges. In particular,  $\Delta_{\frac{p}{q},\lambda}(E) \geq 2 + 2\lambda^q$\ at all its maxima points, and $\Delta_{\frac{p}{q}, \lambda}(E) \leq -2 - 2\lambda^q$\ at all
its minima points.

Recall that the Lyapunov exponent is defined by
\begin{equation}\label{eq:LE}
\gamma_{\alpha, \lambda, \theta} (E) = \lim_{n\to\infty}\frac{1}{n} \ln\|\Phi_n(E, \alpha, \theta) \|~.
\end{equation}
By unique ergodicity, if $\alpha\in\R\setminus\Q$, this expression does not depend on $\theta$. In the periodic case, $\alpha = \frac{p}{q}$, it is determined by the one-period ($q$-step) transfer matrix as follows.
\begin{equation}\label{eq:lyap-exp-per}
\gamma_{\frac{p}{q}, \lambda, \theta}(E) = \lim_{n\rightarrow\infty}\frac{1}{n}\ln\left\|\Phi_n\left(E, \frac{p}{q}, \theta\right)\right\| = \frac{1}{q}\ln\mathrm{Spr}\,\left(\Phi_q\left(E, \frac{p}{q}, \theta\right)\right),
\end{equation}
where $\mathrm{Spr}\,(\cdot)$ denotes the spectral radius.

If $\alpha\in\R\setminus\Q$, $E \in S(\alpha, \lambda)$, the Lyapunov exponent that corresponds to the almost Mathieu operator $\op$ is given by \cite[Corollary 2]{BJ}
\begin{equation}\label{eq:bj-formula}
\gamma (E)=\max\{0,\ln |\lambda|\}~.
\end{equation}

\subsubsection{Norms of transfer matrix products} To estimate the norm of transfer matrices we need the following lemma,  similar to one used in the work \cite{AYZ}. For completeness, we recall the formulation and sketch the proof. Recall that $\Phi_m(\cdot, \cdot, \theta)$ is the $m$-step transfer matrix defined by (\ref{eq:transfer-matrix}) corresponding to the operator $\op$.

\begin{lemma}\label{prop:norms} Let $\alpha\in\R\setminus\Q,\,\, 0<\lambda\leq 1, \nu >0$. Then, there exist $\eta = \eta(\nu, \alpha) > 0$ and $q_0 = q_0(\nu, \alpha) > 0$ such that  for every $|\alpha' - \alpha|, |\frac{p}{q} - \alpha| < \eta$,  $|\epsilon| < \eta$, $|\widetilde\theta|< \eta$, and for every $E\in \R$ with $\dis\, (E, S(\alpha, \lambda)) < \eta$, we have for any $\theta\in\mathbb{T}$ and for any $q\geq q_0$,
\begin{eqnarray}
\label{eq:cor-norms-1}  \left\| \Phi_{q} \left( E +i\epsilon, \frac{p}{q}, \theta\right)- \Phi_{q}\left( E, \frac{p}{q}, \theta\right) \right\| &\leq& |\epsilon|e^{\nu q}~, \\
\label{eq:cor-norms-2}  \left\| \Phi_{q} \left( E +i\epsilon, \alpha', \theta\right)- \Phi_{q} \left( E+i\epsilon, \frac{p}{q}, \theta\right) \right\| &\leq& \left|\alpha'-\frac{p}{q}\right| e^{\nu q}~, \\
\label{eq:cor-norms-3}  
\left\| \Phi_{q} \left( E +i\epsilon, \alpha', \theta\right)- \Phi_{q} \left( E+i\epsilon, \alpha', \theta+\widetilde\theta\right) \right\| &\leq& |\widetilde\theta| e^{\nu q}.~ 
\end{eqnarray}
\end{lemma}
 The main ingredient in the proof of Lemma~\ref{prop:norms} is the following lemma.

\begin{lemma}\label{usc} Let $\alpha\in\R\setminus\Q,\,0<\lambda\leq 1, \nu >0$. Then, there exist $\eta = \eta(\nu, \alpha) > 0$ and $m_0 = m_0(\nu, \alpha) > 0$ such that for any $ |\alpha' - \alpha| < \eta$,  $|\epsilon| < \eta$ and for every $E'$ with $\dis\,(E', S(\alpha, \lambda))< \eta$, we have for any $|m|\geq m_0$
\begin{equation}\label{eq:cor-norms}
\sup_{\theta\in\mathbb{T}}\| \Phi_m ( E' +i\epsilon,  \alpha', \theta)\| \leq e^{\nu m}.
\end{equation}
\end{lemma}

\begin{proof} Since $0 < \lambda \leq 1$, by (\ref{eq:bj-formula}),  we have 
$\gamma_{\alpha, \lambda}(E)=0$ for any $E \in S(\alpha,\lambda)$. 
By continuity of the Lyapunov exponent \cite{BJ,JKS},  for any 
$\nu>0$ there exists $\epsilon_0=\epsilon_0(\nu)>0$, such that if $|\epsilon|\leq \epsilon_0$, then
$\gamma_{\alpha, \lambda}(E+i\epsilon) \leq \frac{\nu}{3}$ for any $E \in S(\alpha,\lambda)$.
The rotation of the circle by $\alpha\in\R\setminus\Q$ is uniquely ergodic and the sequence $\{\ln\|\Phi_m(\theta, \alpha, E) \| \}_m$ is subadditive, therefore by a result of Furman \cite{Fu} we obtain the following. For any $E\in\specU$ there exists $m_0(E, \nu, 
\alpha)$ such that for  any $m \geq m_0(E, \nu, \alpha)$
\[
\sup_{\theta\in\mathbb{T}} \frac{1}{m}\ln\|\Phi_m(E, \alpha, \theta) \|  < \frac{2\nu}{3}~.
\]
This implies that there exists $\eta = \eta(E, \nu, \alpha)>0$ such that if $|\alpha'-\alpha|<\eta(E, \nu, \alpha)$, $|E'-E|< \eta(E, \nu, \alpha)$,  then  for any $m_0(E, \nu, \alpha) \leq m \leq 2m_0(E, \nu, \alpha) +1$
\begin{equation}\label{eq:est-norm}
\sup_{\theta\in\mathbb{T}} \frac{1}{m}\ln\|\Phi_m(E' + i\epsilon, \alpha', \theta) \|  < \nu~.
\end{equation}
By subadditivity, $(\ref{eq:est-norm})$ holds for every $m>m_0(E, \nu, \alpha).$ By compactness of  $\specU$,  there exist $\eta= \eta(\nu, \alpha)>0$, $m_0 = m_0(\nu, \alpha) > 0$, such that if $|\alpha'-\alpha|< \eta(\nu, \alpha)$, $|\epsilon| < \eta$,  and $E'$ is such that $\mathrm{dist_{H}}\,(E', S(\alpha, \lambda))<  \eta(\nu, \alpha)$, then \eqref{eq:cor-norms} holds true for any $|m|\geq m_0$. 
\end{proof}

\begin{proof}[Proof of Lemma~\ref{prop:norms}] Let us prove (\ref{eq:cor-norms-1}), the proofs of other inequalities are similar. Using the continuity of the spectra (Proposition \ref{continuity_of_spectra}), we conclude that if $E\in S(\frac{p}{q}, \lambda)$, then 
\begin{equation}\label{contingap}
\dis\, (E , \specU)< \left|\alpha-\frac{p}{q}\right|^{\frac{1}{2}}.
\end{equation}
In the notations of Lemma~\ref{usc}, let $\alpha' = \frac{p}{q}$. Having \eqref{contingap} Lemma~\ref{usc} implies that there exists $q_0 = q_0(\nu, \alpha) > 2 m_0(\nu,\alpha)$, such that for any $q > q_0(\nu,\alpha)$ and $m\geq m_0(\nu,\alpha)$
\begin{eqnarray}\label{eq:norms-1}
\sup_{\theta} \left\| \Phi_m \left( E +i\epsilon,  \frac{p}{q}, \theta\right)\right\| &\leq& e^{\frac{2\nu m}3},\\
\label{eq:norms-2}\sup_{\theta} \left\| \Phi_m \left( E,  \frac{p}{q}, \theta\right)\right\| &\leq& e^{\frac{2\nu m}3}~.
\end{eqnarray} 
Denote by $\Phi_{q} \left( E +i\epsilon, \frac{p}{q}, \theta\right) = T^\epsilon_q\cdots T^\epsilon_1$ and $\Phi_{q} \left( E, \frac{p}{q}, \theta\right) = T_q\cdots T_1$, where $T^\epsilon_j = T_j(E + i\epsilon), T_j = T_j(E)$\, are one-step transfer matrices defined by (\ref{eq:transfer-matrix}). Then, 
\[
\begin{split}
&\Phi_{q} \left( E +i\epsilon, \frac{p}{q}, \theta\right)- \Phi_{q} \left( E, \frac{p}{q}, \theta\right) =  \sum_{i=1}^{q}T^\epsilon_{q}\cdots
T^\epsilon_{i+1}\Big(T^\epsilon_{i}- {T}_{i}\Big) {T}_{i-1}\cdots
{T}_1\\ & = \Big(\sum_{i=1}^{m_0}+ \sum_{i=m_0+1}^{q-m_0}
+\sum_{i=q-m_0+1}^{q}\Big)T^\epsilon_{q}\cdots
T^\epsilon_{i+1}\Big(T^\epsilon_{i}- {T}_{i}\Big) {T}_{i-1}\cdots
{T}_1\\
&= I_1 + I_2 + I_3~.
\end{split}
\]
By definition (\ref{eq:transfer-matrix}) of $T_j$ we have $\|T^\epsilon_{i} - {T}_{i}\| = |\epsilon|$, thus for any $1\leq j \leq q$,  combining \eqref{eq:norms-1} and \eqref{eq:norms-2},  we obtain
\begin{eqnarray*}
\| I_1\|&\leq& |\epsilon|\sum_{i=1}^{m_0}
(4\lambda+3)^{i-1}e^{(q-i) \frac{2\nu}3}, \\
\| I_2\|&\leq&  |\epsilon| \sum_{i=m_0+1}^{q-m_0}e^{(q-1) \frac{2\nu}3},\\
\| I_3\|&\leq &  |\epsilon| \sum_{i=q-m_0+1}^{q}
(4\lambda+3)^{q-i}e^{(i-1) \frac{2\nu}3}~,
\end{eqnarray*}
thus,  (\ref{eq:cor-norms-1}) holds for sufficiently large $q$.
\end{proof}

\subsection{Overview, and preliminary reductions} The general strategy of the proof of Proposition~\ref{prop:LS} is  as follows. Consider the matrix product
\[ \Phi_n\left(E+i\epsilon, \frac{\widetilde p}{\widetilde q}, \theta\right) =  T_n\left(E+i\epsilon, \frac{\widetilde p}{\widetilde q}, \theta\right) T_{n-1}\left(E+i\epsilon, \frac{\widetilde p}{\widetilde q}, \theta\right) \cdots T_1\left(E+i\epsilon, \frac{\widetilde p}{\widetilde q}, \theta\right) \]
corresponding to $\frac{\widetilde p}{\widetilde q}$. Inside this product, we identify long stretches of indices $j$ for which $E \notin \sigma(\frac p q, \lambda, \theta + 2\pi \frac{\widetilde p}{\widetilde q} j)$. On each such interval, we approximate the transfer matrices corresponding to $\frac{\widetilde p}{\widetilde q}$ by those corresponding to $\frac p q$  and use the avalanche Proposition~\ref{lem:avalanche} to ensure that the growth of the  matrix product does not deteriorate too much. On the remaining intervals, we use a rough bound relying on the imaginary part $\epsilon$ of the spectral parameter. These pieces are glued together using an intermediate operator, which we define in (\ref{eq:inter-op}) below. For the convenience of the reader, we record the main steps:
\[ \text{Proposition~\ref{prop:LS}} \Longleftarrow  (\ref{eq:green2prod}) \Longleftarrow (\ref{eq:green4})
\Longleftarrow \text{Proposition}~\ref{lem:avalanche} + \begin{cases}(\ref{eq:ver-cond2}) \\ (\ref{eq:ver-cond3}) \\  (\ref{eq:u})
\end{cases}\]

To implement this strategy, we need some preliminary reductions. First, we observe that if $\epsilon \geq\frac{C_1\delta}{q^2\lambda^\frac{q}{2}}$ or $E\in J_\delta$ and $\dis \left(E, \sqnp\right) \geq \frac{C_1\delta}{q^2\lambda^{\frac q 2}}$, then  for any $\theta\in\mathbb T$
\[
\gamma_{\frac{\widetilde p}{\widetilde q}, \lambda, \theta}(E + i\epsilon)\geq \gamma_{\frac{\widetilde p}{\widetilde q}, \lambda, \theta}(E) \geq \frac{c\delta}{q^2 \lambda^{\frac q 2}}~.
\]
Indeed, the first inequality holds since by the Thouless formula \eqref{eq:thouless} the Lyapunov exponent is an increasing function of $\epsilon$, and the second one follows from the Combes-Thomas estimate (\ref{eq:ct}).
Thus it suffices to prove (\ref{eq:propLS}) for $E$ and $\epsilon$ satisfiying
\begin{equation}\label{eq:dist-from-spec}
\max\left(\epsilon, \dis\left(E, \sqnp\right)\right) <  \frac{C_1\delta}{q^2\lambda^{\frac q 2}} \leq C_1 c_0(r, \alpha) \lambda^{\frac q 2} \leq C_1 c_0(r, \alpha) ~,
\end{equation}
 where $c_0$ is from assumption (\ref{eq:cond3}), which we used for the second inequality.

Further, we decompose
\[
\begin{split}
J_\delta &= J_\delta^+\bigcup J_\delta^- = \left\{\Discr > 2 - 2\lambda^q + \delta \right\} \bigcup  \left\{\Discr < - 2 + 2\lambda^q - \delta \right\}~.
\end{split}
\]
To  prove (\ref{eq:propLS}), we can also assume that \begin{equation}\label{eq:eps=}
\epsilon = \exp\left( -\frac{e^{rq}}{15000\, q^2\, \lambda^\frac{q}{2}}\right)~,
\end{equation}
since by the Thouless formula the Lyapunov exponent $\gamma_{\frac{\widetilde p}{\widetilde q}, \lambda, \theta} (E + i\epsilon)$ is an increasing function of $\epsilon$.

By $\widetilde q$-periodicity of $H_{\frac{\widetilde p}{\widetilde q}, \lambda, \theta}$ we obtain
\begin{equation}\label{eq:LE-green}
\gamma_{\frac{\widetilde p}{\widetilde q}, \lambda, \theta} (E + i\epsilon) = -\frac{1}{\widetilde q} \ln \left| G^{[1, \infty)}(1, \widetilde q, E + i\epsilon)\right| = -\frac{1}{q\,\widetilde q} \ln \left| G^{[1, \infty)}(1, q\, \widetilde q, E + i\epsilon)\right|~,
\end{equation}
where $G^{[1, \infty)}(\cdot, \cdot, E + i\epsilon)$ is the restricted Green function corresponding to the restricted operator $H^{[1, \infty)}_{\frac{\widetilde p}{\widetilde q}, \lambda, \theta}$. Thus our main goal is to obtain an upper bound  on $ \left| G^{[1, \infty)}(1, q\, \widetilde q, E + i\epsilon)\right| \equiv \left| G^{[1, \infty)}(1, q\, \widetilde q)\right|$. 

Recall the following rough bound. For any interval $I\subset\Z$
\begin{equation}\label{eq:rough-bound}
\|G^I\| = \|(H^I - E- i\epsilon)^{-1}\| \leq \frac{1}{\epsilon}~.
\end{equation}
Define
\begin{equation}\label{eq:theta-inter}
\theta_k = \theta + 2\pi \left(\frac{\widetilde p}{\widetilde q} - \frac{p}{q} \right)k, \quad \theta\in[0, 2\pi)~.
\end{equation}
Define two following sets:
\begin{equation*}
\begin{split}
 \mathcal J_-  & = \left\{j \in\{1,  \dots, q\widetilde q \}\, \big|\,  |(q\theta_j)  \text{mod} 2\pi| \leq \frac{\sqrt\delta}{10 \lambda^{\frac q 2}},
 \text{and}\,\,   |(q\theta_{j - 1} ) \text{mod} 2\pi| > \frac{\sqrt\delta}{10  \lambda^{\frac q 2}} \right\} \\&= \{j_1^-, \dots,  j^-_{M^-}\}~,
 \end{split}
\end{equation*}
and
\begin{equation*}
\begin{split}
 \mathcal J_+ &= \left\{j \in\{1, \dots, q\widetilde q \}\, \big|\,  |(q\theta_j + \pi)  \text{mod} 2\pi| \leq \frac{\sqrt\delta}{10 \lambda^{\frac q 2}},
 \text{and}\,\,   |(q\theta_{j - 1}  + \pi) \text{mod} 2\pi| > \frac{\sqrt\delta}{10  \lambda^{\frac q 2}} \right\} \\&= \{j_1^+, \dots,  j^+_{M^+}\}~.
 \end{split}
\end{equation*}
Denote $\mathcal J = \mathcal J_+ \cup \mathcal J_-$.


Let $2\widehat\eta =  \left|\frac{p}{q} - \frac{\widetilde p}{\widetilde q} \right | < 2 \delta e^{-rq}$, where the inequality follows from assumption  (\ref{eq:cond1}). We claim that 
\begin{equation}\label{eq:sizeM}
\# \mathcal J_\pm = M^\pm \geq [ 2q^2\widetilde q\, \widehat\eta]~.
\end{equation}
Indeed, $q\theta_k$ form an arithmetic progression with step $4\pi q \widehat\eta$, which is smaller than the length $\frac{2\sqrt\delta}{10\lambda^\frac{q}{2}}$ of the arcs
\[
\left( -\frac{2\sqrt\delta}{10\lambda^\frac{q}{2}}, +\frac{2\sqrt\delta}{10\lambda^\frac{q}{2}}\right), \quad \left(\pi -\frac{2\sqrt\delta}{10\lambda^\frac{q}{2}}, \pi +\frac{2\sqrt\delta}{10\lambda^\frac{q}{2}}\right)
\]
due to assumption  (\ref{eq:cond1}). Thus the progression crosses each of the two arcs every time it winds around the circle. The number of windings is equal to
\[
\left[4\pi q \widehat\eta\cdot q\widetilde{q}\cdot\frac{1}{2\pi}  \right] = [2 q^2 \widetilde{q}  \widehat\eta],
\]
thus $\# \mathcal J_\pm = M^\pm \geq [ 2q^2\widetilde q\, \widehat\eta]$ as claimed. We note that $\widehat\eta \geq \frac{1}{q\widetilde{q}}$, hence $\# \mathcal J_\pm = M^\pm \geq q \gg 1$.

Applying Claim~\ref{cl:green-decomp} to $G^{[1, \infty)}(1, q\, \widetilde q)$ with $j^+_k\in \mathcal J_+$ if $E\in J^+_\delta$, $1\leq k \leq M^+$, and with  $j^-_k\in \mathcal J_-$ if $E\in J^-_\delta$, $1\leq k \leq M^-$, we obtain
\begin{equation}\label{eq:green1prod}
\begin{split}
 &\left| G^{[1, \infty)}(1, q\, \widetilde q)\right| \\&\quad= \prod_{k = 1}^{M^\pm-1} \Big[ |G^{[j_k^\pm,  j_k^\pm + l_0 q - 1]}(j_k^\pm,  j_k^\pm + l_0 q - 1)| \times |G^{[j_k^\pm, \infty)}(j_k^\pm + l_0 q,  j_{k + 1}^\pm - 1)|\Big]\\
&\qquad\qquad  |G^{[j_M^\pm,  j_M^\pm + l_0 q - 1]}(j_M^\pm,  j_M^\pm + l_0 q - 1)| \times |G^{[j_M^\pm, \infty)}(j_M^\pm + l_0 q,  q\widetilde q)| \\& \quad\leq \prod_{k = 1}^{M^\pm} \left[ |G^{[j_k^\pm,  j_k^\pm + l_0 q - 1]}(j_k^\pm,  j_k^\pm + l_0 q - 1)|\frac{1}{\epsilon} \right]~,
\end{split}
\end{equation}
where the last inequality follows from the rough bound (\ref{eq:rough-bound}). The main step of the proof is to show that
\begin{equation}\label{eq:green2prod}
|G^{[j_k^\pm,  j_k^\pm + l_0 q - 1]}(j_k^\pm,  j_k^\pm + l_0 q - 1)|\leq \epsilon\exp\left( -\frac{l_0\sqrt\delta}{60}\right)~,
\end{equation}
where we set
\begin{equation}\label{eq:l0}
l_0 = \left[\frac{\sqrt\delta}{160 q^3 \lambda^{\frac q 2} \widehat\eta}\right],\quad \widehat\eta\, < \delta e^{-{rq}}~.
\end{equation}
Note that, the assumption (\ref{eq:cond1}) and the upper bound on $\delta$ in the assumption (\ref{eq:cond3}) guarantee that $1 \leq l_0 \leq \widetilde{q}$. Having (\ref{eq:green2prod}) at hand, combining (\ref{eq:LE-green})  and (\ref{eq:green1prod}), we obtain for $\epsilon \geq \exp\left(-\frac{l_0\sqrt\delta}{90} \right)$
\begin{equation*}
\gamma_{\frac{\widetilde p}{\widetilde q}, \lambda, \theta} (E + i\epsilon) \geq \frac{l_0\sqrt\delta\, M^\pm}{60\, q\, \widetilde q} \geq \frac{l_0\sqrt\delta\, q^2\, \widetilde q\, \widehat \eta}{60 q\, \widetilde q} = \frac{l_0\sqrt\delta}{60}\widehat\eta\, q \geq \frac{\delta}{9600\, q^2\, \lambda^\frac{q}{2}}~,
\end{equation*}
where the second inequality follows from (\ref{eq:sizeM}). This concludes the proof of the proposition.

\subsection{Construction of an intermediate operator} The proof of (\ref{eq:green2prod}) relies on a construction of an intermediate operator. For the rest of the argument we fix $j = j_k \in\mathcal J$. Consider the following operator (depending on $j$)
\begin{equation}\label{eq:inter-op}
(\widetilde{H}\psi)(n) = \left\{ \begin{array}{ll}
\psi(n + 1) + \psi(n - 1) + \widetilde{V}(n)\psi(n), & \textrm{\ if\ } n > j \\
\psi(j + 1) + \widetilde V(j)\psi(j), & \textrm{\ if\ } n = j
\end{array}\right.
\end{equation}
that acts on $\ell^2([j, \infty)\subset\Z)$, where 
\begin{equation}\label{periodic_approximation}
\widetilde{V}(n) = \left\{ \begin{array}{ll}
2\lambda\cos\left(2\pi\frac{p}{q}n + \theta_n\right) & \textrm{\ if\ } j \leq n \leq j + l_0q - 1\\
V^{\mathrm{per}}(n) & \textrm{\ if\ }  n >j + l_0q - 1,
\end{array}\right.
\end{equation}
where $V^{\mathrm{per}}$ is a periodic potential of period $q$ with the period determined by
\[
V^{\mathrm{per}}(m) = 2\lambda\cos\left(2\pi\frac{p}{q}m + \theta_m\right),\quad \text{if}\quad j + (l_0 - 1)q \leq m < j + l_0q - 1~,
\]
namely, we repeat the last $q$ values periodically. Denote by $\widetilde G^{[a, b]}(\cdot, \cdot, E + i\epsilon) \equiv \widetilde G^{[a, b]}(\cdot, \cdot)$ the Green function corresponding to the restriction of the operator $\widetilde H$ to an interval $[a, b] \subset [j, \infty) \subset \Z$. Note that the restricted operator $H_{\frac{\widetilde p}{\widetilde q}, \lambda, \theta}^{[j, j + l_0q - 1]}$ coincides with the intermediate problem restricted to that interval $\widetilde{H}^{[j, j + l_0q - 1]}$, thus 
\begin{equation*}
G^{[j, j + l_0q - 1]}(j, j + l_0q - 1) = \widetilde{G}^{[j, j + l_0q - 1]}(j, j + l_0q - 1)~.
\end{equation*}
Therefore (\ref{eq:green2prod}) is equivalent to
\begin{equation}\label{eq:green2tilde}
|\widetilde G^{[j,  j + l_0 q - 1]}(j,  j + l_0 q - 1)|\leq \epsilon\exp\left( -\frac{l_0\sqrt\delta}{60}\right)~,
\end{equation}
and now we turn to the proof of this inequality. 

First, using \eqref{gre2} of Proposition \ref{gre}, we have 
\[
\begin{split} 
&(\widetilde{G}^{[j, \infty)}- \widetilde{G}^{[j, j + l_0q - 1]} )(j, j + l_0q - 1) = \widetilde{G}^{[j, \infty)}(j, j + l_0q)\widetilde{G} ^{[j, j + l_0q - 1]} (j + l_0q - 1, j + l_0q - 1)~,
\end{split}
\]
thus
\begin{equation*}
\begin{split}
&| \widetilde G^{[j,  j + l_0 q - 1]}(j,  j + l_0 q - 1)| \\
&\leq
| \widetilde G^{[j,  \infty)}(j,  j + l_0 q - 1)| + | \widetilde G^{[j,  \infty)}(j,  j + l_0 q)| |\widetilde{G} ^{[j, j + l_0q - 1]} (j + l_0q - 1, j + l_0q - 1)| \\&
\leq |\widetilde G^{[j,  \infty)}(j,  j + l_0 q - 1)| + | \widetilde G^{[j,  \infty)}(j,  j + l_0 q)| \frac{1}{\epsilon} \\&
\leq \frac{2}{\epsilon}\left[ |\widetilde G^{[j,  \infty)}(j,  j + l_0 q - 1)| + | \widetilde G^{[j,  \infty)}(j,  j + l_0 q)| \right]~,
\end{split}
\end{equation*}
where the second inequality follows from the rough bound (\ref{eq:rough-bound}), and the last one holds since $0 < \epsilon < 1$. Observe that $\epsilon \geq \exp\left(-\frac{l_0\sqrt\delta}{90} \right)$, and note that $l_0\sqrt\delta \geq \left[ \frac{e^{rq}}{160\, q^3  \lambda^\frac{q}{2}}\right] \geq 240$, therefore (\ref{eq:green2tilde}) is implied by 
\begin{equation}\label{eq:green4}
|\widetilde G^{[j,  \infty)}(j,  j + l_0 q - 1)|, | \widetilde G^{[j,  \infty)}(j,  j + l_0 q)| \leq  \frac{2}{\epsilon} \exp\left(-\frac{l_0\sqrt\delta}{20} \right)~.
\end{equation}

\subsubsection{Proof of (\ref{eq:green4})} We may assume without loss of generality that $j = 1$. Indeed $G^{[j, j + l_0q - 1]}(j, j + l_0q - 1)$ corresponding to our $\theta$ is equal to $G^{[1,  l_0q]}(1, l_0q)$ corresponding to $\theta' = \theta + 2\pi\, j\frac{\widetilde p}{\widetilde q}$, and if $j$ lies in $\mathcal J$ corresponds to $\theta$, then $1$ lies in $\mathcal J$ corresponds to $\theta'$.

 For $1 \leq m, k \leq q$, set
\[
Q_{sq + m}^{-1}(E + i\epsilon)= \left(\begin{array}{cc}
0     &      1 \\
-1    &   E + i\epsilon - 2\lambda\cos\left(2\pi\frac{p}{q}(sq + m) + \theta_{sq + m}\right)  \\
\end{array}\right),
\]
\[
\widehat{Q}_{sq + k}^{-1}(E + i\epsilon)= \left(\begin{array}{cc}
0     &      1 \\
-1    &   E + i\epsilon - 2\lambda\cos\left(2\pi\frac{p}{q}(sq + k) + \theta_{sq + 1}\right)  \\
\end{array}\right)~.
\]
The matrix ${Q}_{sq + m}^{-1}$ is the one-step transfer matrix corresponding to the intermediate problem, and $\widehat{Q}_{sq + k}^{-1}$ is the one-step transfer matrix that corresponds to the periodic almost Mathieu operator $H_{\frac{p}{q},\lambda, \theta_{sq + 1}}$ of period $q$. Let
\begin{equation}\label{eq:matrices}
\begin{split}
T_{s}^{-1}(E + i\epsilon) &= (Q_{sq + 1}^{-1}\ldots{Q_{(s
+1)q}^{-1}})(E + i\epsilon), \textrm{\ } 0 \leq s\leq l_0 - 1,
\\
\widehat{T}_{s}^{-1}(E + i\epsilon) &= (\widehat{Q}_{sq + 1}^{-1}\ldots{\widehat{Q}_{(s
+1)q}^{-1}})(E + i\epsilon), \textrm{\ }  0 \leq s\leq l_0 - 1,\\
\qquad \Phi^{-1}_{l_0q}\left(E + i\epsilon, \frac{p}{q}, \theta\right) &=
(T_{0}^{-1}\ldots{T_{l_0 - 1}^{-1}})(E + i\epsilon)~.
\end{split}
\end{equation}
For any $n \geq l_0 q$ the intermediate problem $\widetilde H$ is periodic of period $q$, thus we obtain
\begin{equation}\label{matrix_eq_for_tildeG}
\begin{split}
\left(\begin{array}{cc}
\widetilde{G}^{[1, \infty)}(1, 1)\\
1\\
\end{array}\right) &= \Phi^{-1}_{l_0q}\left(E + i\epsilon, \frac{p}{q}, \theta\right)\left(\begin{array}{cc}
\widetilde{G}^{[1, \infty)}(1, l_{0}q + 1)\\
\widetilde{G}^{[1, \infty)}(1, l_{0}q)\\
\end{array}\right) \\& \equiv \Phi^{-1}_{l_0q}\left(E + i\epsilon, \frac{p}{q}, \theta\right)u~.
\end{split}
\end{equation}
If the assumptions of Proposition~\ref{lem:avalanche} hold, namely if for
\[
\delta_s = \widehat\delta = \frac{\delta}{100}, \,\, 0 \leq s \leq l_0 - 1,\,\, b = \frac{1}{10}, \,\, c = \frac{1}{2},
\]
\begin{align}\label{eq:avalanche-cond}
&\delta_s \leq \delta_s + b\delta_s^\frac{3}{2}~,\\
&\label{eq:ver-cond2}\| T_s^{-1} (E + i\epsilon) - T_{s + 1}^{-1} (E + i\epsilon)\| \leq \frac{\widehat\delta}{10}~,\\
&\label{eq:ver-cond3}|\mathrm{Tr}\, T^{-1}_s(E + i\epsilon)| \geq 2 + (1 - b)\widehat\delta,
\end{align}
and for the vector $u$ defined by (\ref{matrix_eq_for_tildeG})
\begin{equation}\label{eq:u}
\| T^{-1}_{l_0 - 1} u\| \geq \exp \left( \frac{\sqrt\delta}{2}\right)\|u\|,
\end{equation}
then, an application of Proposition~\ref{lem:avalanche} gives
\[
\begin{split}
\frac{1}{\|u\|}\left\|\Phi^{-1}_{l_0q}\left(E + i\epsilon, \frac{p}{q}, \theta\right)u\right\| &\geq \exp \left(\frac{1}{2} \sum_{s = 0}^{l_0 - 1} \sqrt{\widehat\delta}\right) =  \exp \left(\frac{l_0 \sqrt{\widehat\delta}}{2}\right) = \exp \left(\frac{l_0 \sqrt{\delta}}{20}\right)~. 
\end{split}
\]
Using the rough bound (\ref{eq:rough-bound}), we obtain
\begin{equation*}
\begin{split}
&|\widetilde{G}^{[1, \infty)}(1, l_{0}q + 1)|,
|\widetilde{G}^{[1, \infty)}(1, l_{0}q)| \leq 
\|u\|\leq \exp\left(-\frac{l_0\sqrt\delta}{20} \right)\left\|\Phi^{-1}_{l_0q}\left(E + i\epsilon, \frac{p}{q}, \theta\right)u\right\|  \\&\quad= \exp\left(-\frac{l_0\sqrt\delta}{20} \right) \sqrt{|\widetilde{G}^{[1, \infty)}(1, 1)|^2 + 1} \leq \exp\left(-\frac{l_0\sqrt\delta}{20} \right) \frac{2}{\epsilon}~.
\end{split}
\end{equation*}
This concludes the proof of (\ref{eq:green4}).

It is clear that (\ref{eq:avalanche-cond}) holds for any $b > 0$ and $\delta_s = \frac{\delta}{100} \equiv \widehat\delta$, thus it is left to verify (\ref{eq:ver-cond2}), (\ref{eq:ver-cond3}),  and (\ref{eq:u}).

\subsection{Verification of (\ref{eq:ver-cond2}), (\ref{eq:ver-cond3}),  and (\ref{eq:u})} In the notation of Proposition~\ref{lem:avalanche} let $A_j = T_s^{-1},\,\, 0 \leq s \leq l_0 - 1$, and $A_n\cdots A_1\, u_0 = \Phi^{-1}_{l_0q} u$. 

\subsubsection{Verification of (\ref{eq:ver-cond2})} We apply Lemma~\ref{prop:norms} with $\nu = \frac{r}{2},\, \alpha' = \frac{\widetilde p}{\widetilde q},\,\, \epsilon = \exp\left( -\frac{e^{rq}}{15000\, q^2\, \lambda^\frac{q}{2}}\right)$, and $\widetilde \theta = \theta_{(s + 1)q + j} -  \theta_{s q + j}$ for $1 \leq j \leq q$. Then, by assumption (\ref{eq:cond1}) for sufficiently large $q$

\[
\left|\alpha - \frac{\widetilde p}{\widetilde q}\right|, \left| \alpha - \frac{p}{q}\right| < \frac{1}{100}\eta\left(\frac{r}{2}, \alpha\right)^2,
\]
and
\begin{equation}\label{eq:diffth}
 | \theta_{(s + 1)q + j} -  \theta_{s q + j}| = 2\pi\, q \left| \frac{p}{q} - \frac{\widetilde p}{\widetilde q}\right| = 2\pi\, q \left| \frac{\widetilde p}{\widetilde q} - \frac{p}{q}\right| <  \eta\left(\frac{r}{2}, \alpha\right)~.
\end{equation}
Further, our reduction (\ref{eq:dist-from-spec}) yields that
\[
\max\left(\epsilon, \,\dis\, \left(E, S\left(\frac{\widetilde p}{\widetilde q}, \lambda\right)\right)\right)  <  \frac12 \eta\left(\frac{r}{2}, \alpha\right)
\]
provided that we make sure that $c_0(r, \alpha) \leq \frac{1}{2 C_1}\eta\left(\frac{r}{2}, \alpha\right)$. By the continuity of the spectrum (Proposition~\ref{continuity_of_spectra}) for sufficiently large $q \geq q_0$
\[
\begin{split}
\dis\, (E, S(\alpha, \lambda)) &\leq \dis\, \left(E, S\left(\frac{\widetilde p}{\widetilde q}, \lambda\right)\right) + \dist\, \left(S\left(\frac{\widetilde p}{\widetilde q}, \lambda\right), S(\alpha, \lambda)\right) \\& 
\leq \frac12 \eta\left(\frac{r}{2}, \alpha\right) + \frac12 \eta\left(\frac{r}{2}, \alpha\right) = \eta\left(\frac{r}{2}, \alpha\right)~.
\end{split}
\]
Thus Lemma~\ref{prop:norms} is applicable. Also observe that the lower bound  (\ref{eq:cond3}) on $\delta$ and  the assumption (\ref{eq:eps=}) imply  that $\epsilon \leq \delta e^{-rq}$. 
Therefore  using (\ref{eq:cor-norms-1}) 
\begin{equation}\label{eq:step1}
\|\widehat T^{-1}_s (E)  - \widehat T^{-1}_s (E + i\epsilon)\| \leq \delta e^{-rq} e^\frac{rq}{2} < \frac{\delta}{1000} = \frac{\widehat\delta}{10}~.
\end{equation}
The inequality  (\ref{eq:cor-norms-2}) yields that
\begin{equation}\label{eq:step2}
\|\widehat T^{-1}_s (E + i\epsilon)  - T^{-1}_s (E + i\epsilon)\| \leq \delta e^{-rq} e^\frac{rq}{2} < \frac{\delta}{1000} = \frac{\widehat\delta}{10}~,
\end{equation}
and, lastly, (\ref{eq:diffth}) and  (\ref{eq:cor-norms-3}) imply that
\begin{equation}\label{eq:step3}
\|T^{-1}_s (E + i\epsilon)  - T^{-1}_{s + 1} (E + i\epsilon)\| \leq 4 \pi q \delta e^{-rq} e^\frac{rq}{2} < \frac{\delta}{1000} = \frac{\widehat\delta}{10}~.
\end{equation}
Thus (\ref{eq:ver-cond2}) holds true.

\subsubsection{Verification of  (\ref{eq:ver-cond3})} We need to show that for any $0 \leq s \leq l_0 - 1,\,\, E\in J_\delta$, and $\epsilon =\exp\left( -\frac{e^{rq}}{15000\, q^2\, \lambda^\frac{q}{2}}\right)$
\[
|\mathrm{Tr}\, T^{-1}_s(E + i\epsilon)| \geq 2 + (1 - b)\widehat\delta~.
\]
For any $0 \leq s \leq l_0 - 1$ we have
\begin{equation}\label{eq:trace1}
|\mathrm{Tr}\, T^{-1}_s(E + i\epsilon)| \geq |\mathrm{Tr}\, \widehat T^{-1}_s(E)| - 2\|\widehat T^{-1}_s (E)  - T^{-1}_s (E + i\epsilon)\|~. 
\end{equation}
Combining (\ref{eq:step1}) and (\ref{eq:step2}) we obtain
\begin{equation}\label{eq:trace2}
\begin{split}
2\|\widehat T^{-1}_s (E)  - T^{-1}_s (E + i\epsilon)\|  &\leq  
2\|\widehat T^{-1}_s (E)  - \widehat T^{-1}_s (E + i\epsilon)\| + 2\|\widehat T^{-1}_s (E + i\epsilon)  - T^{-1}_s (E + i\epsilon)\| \\& 
\leq 4\delta e^{-rq} e^\frac{rq}{2} \leq \frac{\delta}{5}~.
\end{split}
\end{equation}
Therefore, we need to show that
\begin{equation}\label{eq:trace3}
|\mathrm{Tr}\, \widehat T^{-1}_s(E)| \geq 2 + \frac{3}{4}\delta~.
\end{equation}
By definition (\ref{eq:matrices}) the matrices $\widehat T^{-1}_s$ are the one-period ($q$-step) transfer matrices that correspond  to the periodic almost Mathieu operator $H_{\frac{p}{q}, \lambda, \theta_{sq + 1}}$. Let us show that our choice (\ref{eq:l0}) of $l_0$ guarantees that for all $k\in[j, j + l_0q - 1]$ the energy $E$ is sufficiently far from the spectrum of the operator $H_{\frac{p}{q}, \lambda, \theta_k}$ with $\theta_k$ defined by (\ref{eq:theta-inter}).

\begin{claim}\label{cl:trace} If $E\in J_\delta^+$ and $j\in\mathcal J_+$ or $E\in J_\delta^-$ and $j\in\mathcal J_-$, then for all $k\in[j, j + l_0q - 1]$
\begin{equation*}
\left| D_{\frac{p}{q}, \lambda, \theta_k} (E)\right| \geq 2 + \frac{3}{4}\delta~.
\end{equation*}
\end{claim}
\begin{proof}[Proof of Claim~\ref{cl:trace}] Assume that $E\in J_\delta^+$ and $j\in\mathcal J_+$. Then $|(q\theta_j + \pi)\mathrm{mod}\, 2\pi| < \frac{\sqrt\delta}{10\lambda^\frac{q}{2}}$, therefore by our choice (\ref{eq:l0}) of $l_0 = \left[ \frac{\sqrt\delta}{160\, q^3\, \lambda^\frac{q}{2}\widehat\eta}\right],\,\ \widehat\eta\equiv \delta e^{-rq}$ we get for any $k\in[j, j + l_0q - 1]$
\begin{equation}\label{eq:trace4}
\begin{split}
|(q\theta_k + \pi)\mathrm{mod}\, 2\pi| & \leq |(q\theta_j + \pi)\mathrm{mod}\, 2\pi| + |q(\theta_k - \theta_j)|\\&
\leq \frac{\sqrt\delta}{10\lambda^\frac{q}{2}} + 4\pi\,q\,\widehat\eta\, |k - j| \leq \frac{\sqrt\delta}{10\lambda^\frac{q}{2}} + 4\pi\,q\,\widehat\eta\, l_0 q \\&
\leq \frac{\sqrt\delta}{10\lambda^\frac{q}{2}} + \frac{4\pi}{160}\frac{\sqrt\delta}{\lambda^\frac{q}{2}}  \leq \frac{\sqrt\delta}{5\lambda^\frac{q}{2}}~. 
\end{split}
\end{equation}
Since $E\in J_\delta^+$ we have $\Discr > 2 - 2\lambda^q + \delta$, thus by Chambers' formula (\ref{chamber_formula}) for any $k\in[j, j + l_0q - 1]$
\[
\begin{split}
D_{\frac{p}{q}, \lambda, \theta_k} (E) &= \Discr - 2\lambda^q\cos(q\theta_k)\\&
\geq  2 - 2\lambda^q + \delta -  2\lambda^q\left( -1 + \frac{\delta}{12\pi^2 \lambda^q}\right)\\& = 2 - 2\lambda^q + \delta +  2\lambda^q - \frac{\delta}{6\pi^2}
\geq 2 + \frac{3}{4}\delta~,
\end{split}
\]
where the first inequality holds since for sufficiently large $q$, using (\ref{eq:trace4}) we obtain
\[
\cos(q\theta_k) \leq -1 + \frac{2}{\pi^2} (q\theta_k +\pi)^2 \leq -1 + \frac{2}{\pi^2}\frac{\delta}{25 \lambda^q} \leq -1 + \frac{\delta}{12\pi^2 \lambda^q}~.
\]
Now let $E\in J_\delta^-$ and $j\in\mathcal J_-$. In the same way we obtain that for any $k\in[j, j + l_0q - 1]$
\begin{equation}\label{eq:trace5}
|q\theta_k\, \mathrm{mod}\, 2\pi| \leq \frac{\sqrt\delta}{5\lambda^\frac{q}{2}}~.
\end{equation}
Since $E\in J_\delta^-$ we have $\Discr < - 2 + 2\lambda^q - \delta$, thus by Chambers' formula (\ref{chamber_formula}) for any $k\in[j, j + l_0q - 1]$
\[
D_{\frac{p}{q}, \lambda, \theta_k} (E) 
\leq  -2 + 2\lambda^q - \delta -  2\lambda^q\left( 1 - \frac{\delta}{50 \lambda^q}\right) < -2 -\frac{3}{4}\delta,
\]
where the first inequality holds since (\ref{eq:trace5}) implies for sufficiently large $q$
\[
\cos(q\theta_k) \geq 1 - \frac{(q\theta_k)^2}{2} \geq 1 - \frac{\delta}{50\lambda^q}~.
\]
\end{proof}
Thus combining (\ref{eq:trace1}), (\ref{eq:trace2}), and (\ref{eq:trace3}) we conclude
\begin{equation}\label{eq:trace-final}
|\mathrm{Tr}\, T^{-1}_s (E+ i\epsilon)| > 2 + \frac{3}{4}\delta - \frac{\delta}{5} \geq 2 + \frac{9}{10}\widehat\delta = 2 + (1 - b)\widehat\delta,
\end{equation}
where the equality follows from our choice $\widehat\delta = \frac{\delta}{100},\,\, b = \frac{1}{10}$.

\subsubsection{Verification of (\ref{eq:u})} We need to verify that for the vector $u$ defined by (\ref{matrix_eq_for_tildeG})
\[
\|T^{-1}_{l_0 - 1} u\| \geq \exp\left(\frac{\sqrt\delta}{2} \right)\|u\|~.
\]
By (\ref{eq:trace-final}) we have for $s = l_0 - 1$ 
\[
\left|\mathrm{Tr}\, T^{-1}_{l_0 - 1} (E + i\epsilon)\right| > 2 + \frac{9}{10}\widehat\delta~.
\]
Denote by $e^{\pm(\gamma_s + i\zeta_s)q}$ the eigenvalues of $T^{-1}_{s} (E + i\epsilon)$. Then,
\[
|\mathrm{Tr}\, T^{-1}_{l_0 - 1} (E + i\epsilon)| = |e^{(\gamma_{l_0 - 1} + i\zeta_{l_0 - 1})q} + e^{-(\gamma_{l_0 - 1} + i\zeta_{l_0 - 1})q}|~,
\]
thus we have
\begin{equation}\label{eq:est-gamma}
\gamma_{l_0 - 1}q > \cosh^{-1} \left( 1 + \frac{9}{20}\widehat\delta\right) \geq \frac{1}{2}\sqrt{\widehat\delta}~.
\end{equation}
By Proposition \ref{gre}, we have
\begin{equation}\label{eq:green1}
\begin{split}
&\widetilde{G}^{[1, \infty)}(1, l_{0}q + 1) = \widetilde{G}^{[1, \infty)}(1, (l_{0} -1)q + 1)\widetilde{G}^{[(l_{0} -1)q + 2, \infty)}((l_{0} -1)q + 2, l_{0}q + 1),\\&
\widetilde{G}^{[1, \infty)}(1, l_{0}q) = \widetilde{G}^{[1, \infty)}(1, (l_{0} -1)q)\widetilde{G}^{[(l_{0} -1)q + 1, \infty)}((l_{0} -1)q + 1, l_{0}q)~.
\end{split}
\end{equation}
By definition (\ref{periodic_approximation}) of the intermediate operator $\widetilde H$, using the periodicity and the relation (\ref{eq:LE-green}) between the Lyapunov exponent and Green's function, we get
\begin{equation}\label{eq:green2}
\begin{split}
|\widetilde{G}^{[(l_{0} -1)q + 2, \infty)}((l_{0} -1)q + 2, l_{0}q + 1)| &= |\widetilde{G}^{[(l_{0} -1)q + 1, \infty)}((l_{0} -1)q + 1, l_{0}q)| \\& =e^{-\gamma_{l_0 - 1}q}~.
\end{split}
\end{equation}
By definition of the intermediate operator $\widetilde H$
\[
\left(\begin{array}{cc}
\widetilde{G}^{[1, \infty)}(1, (l_{0} - 1)q + 1)\\
\widetilde{G}^{[1, \infty)}(1, (l_{0} - 1)q)\\
\end{array}\right) = T^{-1}_{l_0 - 1}(E + i\epsilon) u~,
\]
thus, combining the definition (\ref{matrix_eq_for_tildeG}) of $u$, and identities (\ref{eq:green1}) and (\ref{eq:green2}), we obtain
\[
\|u\| = \| T^{-1}_{l_0 - 1} u\|  e^{-\gamma_{l_0 - 1}q}~.
\]
Using the lower bound on $\gamma_{l_0 - 1}q$ given by (\ref{eq:est-gamma}), we conclude
\[
\| T^{-1}_{l_0 - 1} u\| = e^{\gamma_{l_0 - 1}q}\|u\| \geq \exp\left(\frac{1}{2}\sqrt{\widehat\delta}\right)\|u\|~.
\]
This completes the verification of (\ref{eq:u}), thus we have completed the proof of Proposition~\ref{prop:LS}.\qed

\subsection{Proof of Corollary~\ref{cor:LS}}\label{sec:pfLS} Fix an arbitrary $\theta\in[0, 2\pi)$, e.g. $\theta = 0$. First, observe that for any $E\in\sqnp$, $0 < \lambda \leq 1$, by Proposition~\ref{prop1} we have $|\Delta_{\frac{\widetilde p}{\widetilde q}, \lambda} (E)| \leq 2 + 2\lambda^{\widetilde q}$, thus for $E$ such that $\gamma_{\frac{\widetilde p}{\widetilde q}, \lambda, 0}(E) \neq 0$, we obtain
\[
6 \geq 2 + 4\lambda^{\widetilde q} \geq |D_{\frac{\widetilde p}{\widetilde q}, \lambda, 0}(E)| = e^{\gamma_{\frac{\widetilde p}{\widetilde q}, \lambda, 0}(E) \widetilde q} + e^{-\gamma_{\frac{\widetilde p}{\widetilde q}, \lambda, 0}(E) \widetilde q} \geq e^{\gamma_{\frac{\widetilde p}{\widetilde q}, \lambda, 0}(E) \widetilde q}~,
\] 
namely
\begin{equation*}
\gamma_{\frac{\widetilde p}{\widetilde q}, \lambda, 0}(E)\leq \frac{\ln(2 + 4\lambda^{\widetilde q})}{\widetilde q} \leq \frac{\ln 6}{\widetilde q}~.
\end{equation*}
This inequality obviously holds also for $\gamma_{\frac{\widetilde p}{\widetilde q}, \lambda, 0}(E) = 0$. Let $\epsilon = \exp\left( -\frac{e^{rq}}{15000\, q^2\, \lambda^\frac{q}{2}}\right)$. For any measure $\mu$ for which
\[
\mu(a, b) \leq C\omega(b - a),\,\, b-a\leq 1~,
\]
we obtain using Proposition~\ref{prop-surace}
\begin{equation}\label{eq:lyap-diff}
\begin{split}
\mu&\left\{ E\in\R\, \big|\, \left|\gamma_{\frac{\widetilde p}{\widetilde q}, \lambda, 0}(E + i\epsilon) - \gamma_{\frac{\widetilde p}{\widetilde q}, \lambda, 0}(E)\right| \geq \frac{\delta}{19200 q^2 \lambda^{\frac q 2}}\right\}\\& \leq  \frac{C_\mu\,19200\, q^2 \lambda^{\frac q 2}}{\delta} (W_3\, \omega)\left[  \exp\left( -\frac{e^{rq}}{15000\, q^2\, \lambda^\frac{q}{2}}\right)\right] ~.
\end{split}
\end{equation}
By the assumption (\ref{eq:cond1}) of Proposition~\ref{prop:LS} we have
\[
\frac{1}{q\, \widetilde q} \leq \left| \frac{p}{q} - \frac{\widetilde p}{\widetilde q}\right| < 2\delta e^{-rq}~,
\]
namely $\widetilde q > \frac{e^{rq}}{2\delta\, q}$, thus for sufficiently large $q$
\[
\frac{\ln 6}{\widetilde q} \leq \frac{\delta}{19200\, q^2}\leq \frac{\delta}{19200\, q^2\, \lambda^\frac{q}{2}}~.
\]
Therefore, using (\ref{eq:lyap-diff}) and the estimate (\ref{eq:propLS}) given by Proposition~\ref{prop:LS}, we obtain
\[
\begin{split}
&\mu\left(\left\{E\in\R\big|\,\, \gamma_{\frac{\widetilde p}{\widetilde q}, \lambda, 0}(E) \leq \frac{\ln 6}{\widetilde q}  \right\} \bigcap J_\delta \right) \\& \leq  \mu\left( \left\{E\big|\,\, \gamma_{\frac{\widetilde p}{\widetilde q}, \lambda, 0}(E) \leq \frac{\delta}{19200 q^2 \lambda^{\frac q 2}}  \right\} \bigcap \left\{E\big|\,\, \gamma_{\frac{\widetilde p}{\widetilde q}, \lambda, 0}(E + i\epsilon) \geq \frac{\delta}{9600 q^2 \lambda^{\frac q 2}}  \right\} \right)\\
&\leq \mu\left\{ E\in\R\, \big|\, \left|\gamma_{\frac{\widetilde p}{\widetilde q}, \lambda, 0}(E + i\epsilon) - \gamma_{\frac{\widetilde p}{\widetilde q}, \lambda, 0}(E)\right| \geq \frac{\delta}{19200 q^2 \lambda^{\frac q 2}}\right\}\\
&\leq  \frac{\widehat C_\mu\,q^2 \lambda^{\frac q 2}}{\delta} (W_3\, \omega)\left[  \exp\left( -\frac{e^{rq}}{15000\, q^2\, \lambda^\frac{q}{2}}\right)\right]~. 
\end{split}
\]
This concludes the proof of Corollary~\ref{cor:LS}.\qed

\subsection{Proof of Proposition~\ref{prop:LS1}} Invoking the Combes--Thomas estimate as in the proof of Proposition~\ref{prop:LS} (\ref{eq:dist-from-spec}) we may assume that
\begin{equation}\label{eq:dist-from-spec1}
\max\left(\epsilon , \dis\, \left( E, S\left(\frac{p_{n + 1}}{q_{n + 1}}, \lambda\right)\right) \right) \leq \frac{c\delta}{q_n^2}~.
\end{equation} 
Furthermore, we only need to consider the case $E\in J^-_{\delta}$, namely  $\Discr < - 2 + 2\lambda^q - \delta $, since the case $E\in J^+_{\delta}$ can be dealt with similarly.   
By our assumption $\left\{ \frac{p_n}{q_n}\right\}$ is the sequence of convergents of $\alpha$, thus
\begin{equation*}
|p_nq_{n + 1} - p_{n + 1}q_n| = 1,\,\, \left|\frac{p_n}{q_n}  - \frac{p_{n + 1}}{q_{n + 1}}\right| = \frac{1}{q_n\, q_{n + 1}}~.
\end{equation*}
Unlike in the proof of Proposition~\ref{prop:LS} here we make use of only one intermediate problem, which is defined as follows. Let
\begin{equation*}
\widetilde{V}(m) = \left\{ \begin{array}{ll}
2\lambda\cos\left(2\pi\frac{p_n}{q_n}m + \theta_m\right) & \textrm{\ if\ } 1 \leq m <  l_0q \\
2\lambda\cos\left(2\pi\frac{p_n}{q_n}m + \theta_{l_0 q}\right)& \textrm{\ if\ }  m \geq  l_0q,
\end{array}\right.
\end{equation*}
where $1\leq l_0 \leq q_{n + 1}$\ is an integer chosen by 
\begin{equation}\label{eq:l0-2}
l_0 = \left[\frac{q_{n + 1} \sqrt\delta}{100\, q_n}\right]~, 
\end{equation}
and 
\begin{equation*}
\theta_m = \theta + 2\pi \left(\frac{p_{n+1}}{q_{n + 1}} - \frac{p_n}{q_n} \right)m, \quad \theta\in[0, 2\pi)~.
\end{equation*}
Note that $l_0  \geq 1$ by (\ref{eq:cond-ls1}). Then we define the operator $\widetilde H$ as in (\ref{eq:inter-op}) acting on $\ell^2([1, \infty)\subset\Z)$. By definition
\[
H^{[1, l_0 q_n]}_{\frac{p_{n + 1}}{q_{n + 1}}, \lambda, \theta} = \widetilde H^{[1, l_0 q_n]},
\]
thus we have an equality between the corresponding Green functions
\begin{equation*}
G^{[1, l_0q_n]}(1, l_0q_n) = \widetilde{G}^{[1, l_0q_n]}(1, l_0q_n)~.
\end{equation*}
Let us show that the estimates (\ref{eq:green1prod}) and (\ref{eq:green4}) remain true with the current definitions if we set $j = 1$. We restate these estimates as follows.

\begin{enumerate}
\item\label{eq:item1} \text{For any}\, $\epsilon \geq \exp\left( -\frac{l_0\sqrt\delta}{90}\right)$  
\begin{equation}\label{eq:prop-ls1-2}
|\widetilde G^{[1, \infty)}(1, l_0q_n)| \leq \frac{2}{\epsilon}  \exp\left( -\frac{l_0\sqrt\delta}{20}\right)~.
\end{equation}
\item\label{eq:item2} 
\begin{equation}\label{eq:prop-ls1-3}
|G^{[1, \infty)}(1, q_n\, q_{n + 1})| \leq \frac{5}{\epsilon^3} |\widetilde G^{[1, \infty)}(1, l_0q_n)|  \leq  \exp\left( -\frac{l_0\sqrt\delta}{60}\right)~.
\end{equation}
\end{enumerate}
With these estimates in hand, the proof is concluded in the same way as in Proposition~\ref{prop:LS}.

\subsubsection{Proof of (\ref{eq:prop-ls1-2})} As we can see from the proof of  (\ref{eq:green4}), the key ingients are Lemma~\ref{prop:norms} and Claim~\ref{cl:trace}. Let us first verify the conditions of Lemma~\ref{prop:norms}.  Let $\nu = \frac{r}{2}$, then for any $n\geq n_0$
\[
 \frac{c}{q_n^2} < \frac{1}{2}\eta\left(\frac{r}{2}, \alpha\right)~,
\]
where  $\eta\left(\frac{r}{2}, \alpha\right)$  is given by Lemma~\ref{prop:norms}. In the notation of Lemma~\ref{prop:norms}, let $\alpha' =  \frac{p_{n + 1}}{q_{n + 1}},\,\, \widetilde\theta = \theta_{(s+1)q_n + j} - \theta_{sq_n + j},\,\, 1 \leq j \leq q_n$.  Now, assume that \eqref{eq:cond-ls1} holds,  then for any $n\geq n_0$, we have 
\[
\left|\alpha - \frac{p_{n}}{q_{n}} \right|,\,\, \left|\alpha - \frac{p_{n+1}}{q_{n+1}} \right| \leq  \frac{1}{q_n\, q_{n + 1}} < 2\delta e^{-rq_n} < \frac{1}{100} \eta\left(\frac{r}{2}, \alpha\right)^2,
\]
\[
|\widetilde\theta| = 2\pi\, q_n \left|\frac{p_n}{q_n}  - \frac{p_{n + 1}}{q_{n + 1}}\right| = \frac{2\pi}{q_{n + 1}} <  \frac{1}{100} \eta\left(\frac{r}{2}, \alpha\right)^2.
\]
On the other  hand,  by (\ref{eq:dist-from-spec1}), we have
\[
\max\left(\epsilon , \dis\, \left( E, S\left(\frac{p_{n + 1}}{q_{n + 1}}, \lambda\right)\right) \right) \leq \frac{c\delta}{q_n^2}< \frac{1}{2}\eta\left(\frac{r}{2}, \alpha\right)~.
\]
By the continuity of the spectrum (Proposition~\ref{continuity_of_spectra}) we get
\[
\dis\, (E, S(\alpha, \lambda)) \leq  \dis\, \left( E, S\left(\frac{p_{n + 1}}{q_{n + 1}}, \lambda\right)\right) + \left|\alpha- \frac{p_{n + 1}}{q_{n + 1}}\right|^{\frac{1}{2}} \leq \eta\left(\frac{r}{2}, \alpha\right)~.
\]
Thus Lemma~\ref{prop:norms} is applicable and we apply it in the same way as in the proof of Proposition~\ref{prop:LS}.

\medskip Next, we need the following version of Claim~\ref{cl:trace}.
\begin{claim}\label{cl:trace1} For any $E\in J_\delta^-$ and any $1\leq k \leq l_0 q_n$, we have
\begin{equation*}
\left|D_{\frac{p_n}{q_n}, \lambda, \theta_k}(E) \right| > 2 + \frac{3}{4}\delta~.
\end{equation*}
\end{claim}
\begin{proof}[Proof of Claim~\ref{cl:trace1}] As in the proof of Proposition~\ref{prop:LS} we may assume that $\theta = 0$. Then, for $l_0 = \left[ \frac{q_{n + 1}\sqrt\delta}{100 q_n}\right]$ we get
\[
|\theta_k| \leq |\theta_{l_0 q_n}| = 2\pi l_0 q_n \left|\frac{p_n}{q_n}  - \frac{p_{n + 1}}{q_{n + 1}}\right| = \frac{2\pi l_0}{q_{n + 1}} = \frac{2\pi\sqrt\delta}{100 q_n}~,
\]
thus
\[
\cos(q_n \theta_{l_0 q_n}) \geq 1 - \frac{(q_n \theta_{l_0 q_n})^2}{2} \geq 1 - \frac{2\pi^2\, \delta}{10000}~.
\]
Therefore, for any $E\in J_\delta^-$ and any $1\leq k\leq l_0 q_n$
\[
\begin{split}
D_{\frac{p_n}{q_n}, \lambda, \theta_k}(E) &= \Delta_{\frac{p_n}{q_n}, \lambda} (E) - 2\lambda^{q_n} \cos(\theta_k q_n) \\&< -2 + 2\lambda^{q_n} -\delta - 2\lambda^{q_n} + \frac{4\lambda^{q_n}\pi^2 \delta}{10000} < - 2 - \frac{3}{4}\delta~,
\end{split}
\]
where the last inequality holds since $0 < \lambda \leq 1$.  
\end{proof}
Having this estimate at hand the proof of (\ref{eq:prop-ls1-2}) is concluded in the same way as in Proposition~\ref{prop:LS}. We omit the details. 

\subsubsection{Proof of (\ref{eq:prop-ls1-3})} We closely follow \cite[Proof of Theorem 2]{LS}. By (\ref{gre1'}) of Proposition~\ref{gre} for any $k < m < l$
\begin{equation}\label{eq:prop-ls1-3-1}
G^{[k, \infty)} (k, l) = G^{[k, \infty)} (k, m)G^{[m + 1, \infty)} (m+1, l)~,
\end{equation}
therefore
\begin{equation}\label{eq:est-green-ls1}
\begin{split}
|G^{[1, \infty)} (1, q_n q_{n + 1}) | & \leq |G^{[1, \infty)} (1, l_0 q_n) G^{[ l_0 q_n + 1, \infty)} ( l_0 q_n + 1, q_n q_{n + 1})| \\&
\leq |G^{[ l_0 q_n + 1, \infty)} ( l_0 q_n + 1, q_n q_{n + 1})|\times\\&
\left(|\widetilde G^{[1, \infty)} (1, l_0 q_n)|  + |G^{[1, \infty)} (1, l_0 q_n)  -\widetilde G^{[1, \infty)} (1, l_0 q_n) | \right)~.
\end{split}
\end{equation}
To estimate the difference, first we use the second resolvent identity and obtain 
\[
\begin{split}
 & \left| G^{[1, \infty)}(1, l_0 q_n) - \widetilde{G}^{[1, \infty)}(1, l_0q_n)\right| \\&
 \leq \sum_{k = l_0 q_n +1}^\infty \left| \widetilde G^{[1, \infty)}(1, k)(\widetilde V(k) - V(k))G^{[1, \infty)}(k, l_0 q_n)\right| \\&
 \leq 4\lambda  \sum_{k = l_0 q_n +1}^\infty \left| \widetilde G^{[1, \infty)}(1, k)G^{[1, \infty)}(k, l_0 q_n)\right|. 
\end{split}
\]
Using (\ref{eq:prop-ls1-3-1}) we obtain for $k \geq l_0 q_n + 1$
\[
 \widetilde G^{[1, \infty)}(1, k) =  \widetilde G^{[1, \infty)}(1, l_0 q_n) \widetilde G^{[l_0 q_n + 1, \infty)}(l_0 q_n + 1, k),
\]
thus
\[
\begin{split}
 &\left|G^{[1, \infty)}(1, l_0 q_n) - \widetilde{G}^{[1, \infty)}(1, l_0q_n)\right|  \\&\leq 4\lambda \left|  \widetilde G^{[1, \infty)}(1, l_0 q_n) \right| \sum_{k = l_0 q_n +1}^\infty \left| \widetilde G^{[l_0 q_n + 1, \infty)}(l_0 q_n + 1, k)G^{[1, \infty)}(k, l_0 q_n)\right| \\&
\leq 4\lambda \left|\widetilde G^{[1, \infty)}(1, l_0 q_n) \right| \left(\sum_{k=1}^\infty   \left| \widetilde G^{[1, \infty)}(1, k)\right|^2\right)^{\frac12}\left(\sum_{k=1}^\infty  \left| G^{[1, \infty)}(k, l_0 q_n)\right|^2\right)^{\frac12},
\end{split}
\]
where in the last step we used the Cauchy-Schwarz inequality. Recall that for any $k_0\in[1, \infty)$
\begin{equation*}
\sum_{k=1}^\infty  \left| G^{[1, \infty)}(k_0, k)\right|^2 = \frac{1}{\epsilon}\mathrm{Im\,\,} G^{[1, \infty)}(k_0, k_0) \leq  \frac{1}{\epsilon^2}~,
\end{equation*}
thus we conclude that 
\[
\sum_{k=1}^\infty   \left| \widetilde G^{[1, \infty)}(1, k)\right|^2, \sum_{k=1}^\infty  \left| G^{[1, \infty)}(k, l_0 q_n)\right|^2 \leq  \frac{1}{\epsilon^2}~.
\]
Therefore,
\[
|G^{[1, \infty)} (1, l_0 q_n)  -\widetilde G^{[1, \infty)} (1, l_0 q_n) |  \leq \frac{4\lambda}{\epsilon^2}|\widetilde G^{[1, \infty)} (1, l_0 q_n) |~. 
\]
Combining \eqref{eq:est-green-ls1}, and the last inequality with the rough bound (\ref{eq:rough-bound}) we obtain
\[
\begin{split}
\left|{G}^{[1, \infty)}(1, q_n q_{n + 1}) \right| &\leq \left|\widetilde G^{[1, \infty)}(1, l_0 q_n) \right| \left(1 +   \frac{4\lambda}{\epsilon^2} \right)\left|{G}^{[l_0 q_n + 1, \infty)}(l_0 q_n + 1, q_n q_{n + 1})\right| \\& \leq \frac{5}{\epsilon^3} \left|\widetilde G^{[1, \infty)}(1, l_0 q_n) \right|~,
\end{split}
\]
where the last inequality follows from the assumption that $\epsilon < 1$ and $\lambda \leq 1$. This concludes the proof of  (\ref{eq:prop-ls1-3}) and of Proposition~\ref{prop:LS1}.
We omit the proof of  Corollary~\ref{cor:LS1}, which is parallel to that of Corollary~\ref{cor:LS} (see Section~\ref{sec:pfLS}).\qed

\section{Proof of Theorem~\ref{thm:hausdorff}}
We repeatedly use Frostman's lemma, which we now state. Recall that the $\omega$-Hausdorff content of a set  $K\subset\R$ is defined via
\[
\cont(K) = \inf \left\{ \sum_{j = 1}^\infty \omega(b_j - a_j)\quad \big|\quad K\subset\bigcup_{j=1}^\infty (a_j, b_j) \right\}~.
\]
This is consistent with setting $\epsilon = \infty$ in (\ref{eq:meas-omega}), and in particular $\cont(K) \leq \measphi(K)$. Moreover, $\cont(K) = 0$ if and only if $\mathcal{H}^\omega$ = 0 \cite[Proposition~1.2.6]{bp}. 
\begin{lemma}[Frostman, see \cite{bp}]\label{lem:frostman} There exists $C > 0$ such that the following holds. Let $\omega$ be a gauge function. Let $K\subset\R$ be a compact set with positive $\omega$-Hausdorff content, $\cont(K) > 0$. Then there exists a non-zero Borel measure $\mu\geq 0$ supported on $K$ satisfying
\begin{equation}\label{eq:frost1}
\mu(a, b) \leq C \omega(b - a)~.
\end{equation}
Vice versa, if there exists a non-zero Borel measure $\mu\geq 0$ supported on $K$ satisfying (\ref{eq:frost1}), then $\cont(K) > 0$.
\end{lemma}
\begin{rmk}\label{rmk:frostman} Lemma~\ref{lem:frostman} and the result of Craig and Simon~\cite{craig-simon} imply that $\mathcal H^{\omega_1} (\sigma(H)) \geq \cont(\sigma(H)) > 0$ for any ergodic Schr\"odinger operator $H$. Using  a result of Bourgain and Klein~\cite{bourgain-klein} in place of \cite{craig-simon}, we see that the same conclusion holds for any discrete Schr\"odinger operator.
\end{rmk}

\subsection{Proof of  Theorem~\ref{thm:hausdorff} \eqref{hausdorff}}  The proof   is by contradiction.  Assume that $ \mathcal{H}^{\omega_t}(S(\alpha, 1)) > 0$. Then, by Frostman's Lemma (Lemma~\ref{lem:frostman}) there exists a non-zero Borel  measure $\mu$ (we can assume without loss of generality that $\mu$ is supported on an interval of length $\frac1e$), such that

\begin{enumerate}
\item\label{meas:cond-supp} $\mathrm{supp}\, \mu \subset S(\alpha, 1)$,
\item\label{meas:cond-log-holder} $\mu[a, b] \leq\frac{C}{\ln^t\frac{1}{|b - a|}}$ for any $a < b $, $t > 2$~.
\end{enumerate}
Denote  $\mu(S(\alpha, 1)) = m > 0$. Assume that  $\beta(\alpha)>0$. Then, for any $ 0 < r < \beta(\alpha)$,  there exists a sequence  $\frac{P_k}{Q_k}\rightarrow \alpha$ (a subsequence of convergents of $\alpha$) such that
\begin{eqnarray}
\label{large-q_n'}\left|\alpha - \frac{P_k}{Q_k}\right| \leq e^{-(\beta(\alpha) - r ) Q_k},
\end{eqnarray}
 In the current proof, we choose $r=\frac{\beta(\alpha)}{100 t}$.

From Proposition~\ref{continuity_of_spectra}, we have
\begin{equation}\label{eq:inclusion}
S(\alpha, 1) \subset S\left(\frac{P_k}{Q_k}, 1\right)
    + \left(- 6 \sqrt{2 \left|\alpha - \frac{P_k}{Q_k}\right|}, \, 6 \sqrt{2 \left|\alpha-\frac{P_k}{Q_k}\right|}\right)~,
\end{equation}
where on the right hand side we have the Minkowski sum
\[ A + B = \left\{ a + b \, \big| \, a \in A, \, b \in B \right\},\ \quad A, B \subset \R ~. \]

By \eqref{large-q_n'} and \eqref{eq:inclusion}, the set $S(\alpha, 1)\setminus S\left(\frac{P_k}{Q_k}, 1\right)$ is contained in the union of (at most) $2Q_k$ intervals of length $\leq C e^{-\frac{ (\beta(\alpha) -r ) Q_{k }}{2}}$, since     $S\left(\frac{P_k}{Q_k}, 1\right)$ has at most $Q_k$ bands.    Therefore, by the assumption  \eqref{meas:cond-log-holder}, for sufficiently large $k$ we obtain
\begin{equation}\label{eq:mu-alpha-minus-pn}
\mu \left( S(\alpha, 1)\setminus S\left(\frac{P_k}{Q_k}, 1\right)\right) \leq \frac{2C_t Q_k}{ (\beta(\alpha) - r)^t Q_k^t}  \leq \frac{\widetilde C_t}{Q^{t - 1}_k}\leq \frac{m}{2}~,
\end{equation}
since $t > 2$.  By our assumptions (1)--(2) above, the measure $\mu$ restricted to the set $S\left(\frac{P_k}{Q_k}, 1\right)$, which we denote by $\mu_k = \mu\restriction S\left(\frac{P_k}{Q_k}, 1\right)$, obeys  the following:
\begin{enumerate}
\item $\mathrm{supp}\, \mu_k \subset S\left(\frac{P_k}{Q_k}, 1\right)$~,
\item $\mu_k [a, b]  \leq\frac{C}{\ln^t\frac{1}{|b - a|}}$~,
\item \label{eq:mu_n-large} $\mu_k \left(S\left(\frac{P_k}{Q_k}, 1\right)\right) \geq \frac{m}{2}$~.
\end{enumerate}
\noindent Let $0 < \nu < (t - 2)r$. From  (\ref{large-q_n'}) and since $r = \frac{\beta(\alpha)}{100 t}$, we obtain for sufficiently large $k$
\[
\left| \alpha - \frac{P_{k - 1}}{Q_{k - 1}}\right|, \qquad \left| \alpha - \frac{P_k}{Q_k}\right| \leq e^{-(\beta(\alpha) - r) Q_{k-1}} \leq e^{-r Q_{k-1}}  e^{-\nu Q_{k-1}}.
\]
Thus  we can  apply  Corollary~\ref{cor:LS} with $\frac{p}{q} = \frac{P_{k - 1}}{Q_{k - 1}}, \,\, \frac{\widetilde p}{\widetilde q} = \frac{P_{k}}{Q_{k}}, \,\, \delta =  e^{-\nu Q_{k-1}}$ and get using Remark~\ref{log-t} (it is here that the assumption $t > 2$ is used)
\begin{equation*}
\mu_k \left(S\left(\frac{P_k}{Q_k}, 1\right)\bigcap J_\delta\right) \leq C_t Q_{k -1}^{2t}e^{-(t - 1)r Q_{k - 1}}e^{\nu Q_{k - 1}}~,
\end{equation*}
where \[ J_\delta = \left\{ E \, \big| \, \left|\Delta_{\frac{P_{k-1}}{Q_{k-1}},1}(E)\right| > \delta \right\}~.\] Since $t > 2$ and $0 < \nu < (t - 2)r$, we obtain
\begin{equation}\label{eq:mu_n-intersec}
\lim_{k\to\infty} \mu_k \left(S\left(\frac{P_k}{Q_k}, 1\right)\bigcap J_\delta\right) = 0~.
\end{equation}
\noindent By \cite[equation (5.3), Lemma 5.1]{LS}, we obtain
\begin{equation*}
|J_\delta^c| \leq  \frac{2e \delta}{Q_{k - 1}}~,
\end{equation*}
therefore  the set $S\left(\frac{P_k}{Q_k}, 1\right)\setminus J_\delta$ is contained in the union of (at most) $Q_{k - 1}$ intervals of total length at most  $$\frac{2 e\delta}{Q_{k - 1}} + C e^{-\frac{  (\beta(\alpha) - r ) Q_{k - 1}}{2}} \leq Ce^{-\frac{\nu}{2}Q_{k - 1}},$$ where the last inequality holds by the choice of $\delta$. Therefore,
\begin{equation}\label{eq:mu_n-setminus}
\mu_k \left( S\left(\frac{P_k}{Q_k}, 1\right)\setminus J_\delta \right) \leq Q_{k - 1}\frac{C_t}{\ln^t(e^{\frac{\nu}2 Q_{k - 1}})} \underset{k\to\infty}{\longrightarrow} 0~,
\end{equation}
since $t > 2$. Combining the third condition (\ref{eq:mu_n-large}) on $\mu_k$  with the estimates (\ref{eq:mu_n-intersec}) and (\ref{eq:mu_n-setminus}), we obtain a contradiction. \qed

\subsection{Proof of  Theorem~\ref{thm:hausdorff} \eqref{thm:hausdorff-opt}} Let $\frac{p_n}{q_n}\rightarrow \alpha$ be the sequence of convergents of $\alpha$. We show that $\measphi(S(\alpha, 1)) = 0$ if the denominators $q_n$ grow sufficiently fast. We construct such $\alpha$ via its convergents by induction: for any fixed $n_0$ the terms $q_1, \dots, q_{n_0}$ can be chosen in an arbitrary way. If $q_1, \dots, q_{n}, \, n > n_0$, are given, we find $x_n > 0$ such that for any $0 < x < x_n$
\begin{equation}\label{eq:haus-opt-cond1}
\omega(x)\log\frac{1}{x} \leq e^{-2q_n}~,
\end{equation}
and we choose
\begin{equation}\label{eq:haus-opt-cond2}
q_{n+1} \geq \max\left( 2C_1 e^{q_n}\log\frac{1}{x_n}, \frac{1}{\omega^{-1}\left( \frac{1}{q_n^2}\right)^2}\right)~,
\end{equation}
where $C_1 > 0$ is the constant from Corollary~\ref{cor:LS1}.  It is clear that the set of $\alpha$ thus constructed is $G_\delta$-dense.

 For each $n$, we choose $\delta = e^{-\frac{q_n}{2}}$. { We can construct a cover of $S(\alpha, 1)$ as follows. Let $X_n$ consist of all the intervals
 \[
\left (E - 6\sqrt2 \sqrt{\left| \alpha - \frac{p_{n + 1}}{q_{n + 1}}\right|,}\, E\right), \quad \left(E',\, E' + 6\sqrt2 \sqrt{\left| \alpha - \frac{p_{n + 1}}{q_{n + 1}}\right|}\right),
 \]
 where $(E, E')$ is a band of $S\left(\frac{p_{n + 1}}{q_{n + 1}}, 1\right)$. Then by the continuity of the spectrum (Proposition~\ref{continuity_of_spectra})}

\begin{equation*}
\begin{split}
S(\alpha, 1) &\subset S\left(\frac{p_{n + 1}}{q_{n + 1}}, 1 \right) + \left( -6\sqrt2 \sqrt{\left| \alpha - \frac{p_{n + 1}}{q_{n + 1}}\right|}, + 6\sqrt2 \sqrt{\left| \alpha - \frac{p_{n + 1}}{q_{n + 1}}\right|}\right)\\& \subset S\left(\frac{p_{n + 1}}{q_{n + 1}}, 1 \right) \bigcup X_n~, \\
& \subset \left(  S\left(\frac{p_{n + 1}}{q_{n + 1}}, 1 \right) \bigcap J_\delta\right) \bigcup J_\delta^c \bigcup X_n.
\end{split}
\end{equation*} 
To estimate   $\measphi(S(\alpha, 1))$, one can proceed as follows. First note that $X_n$ is a union of (at most) $2q_{n + 1}$ intervals of length $\leq 12\sqrt2 \sqrt{\left| \alpha - \frac{p_{n + 1}}{q_{n + 1}}\right|} \leq \frac{1}{\sqrt{q_{n + 2}}}$ (for sufficiently large $n$).  By \eqref{eq:haus-opt-cond2} we conclude $\omega\left( \frac{1}{\sqrt{q_{n + 2}}}\right) \leq \frac{1}{q_{n + 1}^2}$, thus the contribution of these intervals to  \eqref{eq:meas-omega} is
\[
\leq \omega\left( \frac{1}{\sqrt{q_{n + 2}}}\right) \times 2q_{n+ 1} \leq  \frac{1}{q_{n + 1}^2} \times 2q_{n+ 1} \underset{n\to\infty}{\longrightarrow} 0~.
\]
On the other hand, by Corollary~\ref{cor:LS1} we have
\[
\begin{split}
\left|S\left(\frac{p_{n + 1}}{q_{n + 1}}, 1 \right) \bigcap J_\delta\right| &\leq Cq_n^2 e^\frac{q_n}{2}\exp\left(-\frac{q_{n+1}}{C_1 q_n}  e^{-\frac{q_n}{2}}\right) \\&\leq \exp\left(-\frac{q_{n+1}}{2C_1 q_n}  e^{-{q_n}}\right)~.
\end{split}
\]
This set can be covered by (at most) $q_{n + 1} + q_n \leq 2q_{n + 1}$ intervals that contribute to (\ref{eq:meas-omega}) at most $2q_{n + 1}\omega\left( \exp\left(-\frac{q_{n+1}}{2C_1 q_n}  e^{-{q_n}}\right)\right)$. Observe that by \eqref{eq:haus-opt-cond2} we have $q_{n + 1} \geq 2C_1 e^{q_n}\log\frac{1}{x_n}$ for $x_n$ defined by  \eqref{eq:haus-opt-cond1}, namely $\exp\left(-\frac{q_{n+1}}{2C_1 q_n}  e^{-{q_n}}\right) \leq x_n$. Therefore, \eqref{eq:haus-opt-cond1} implies that
\[
\omega\left( \exp\left(-\frac{q_{n+1}}{2C_1 q_n}  e^{-{q_n}}\right)\right) \log\frac{1}{\exp\left(-\frac{q_{n+1}}{2C_1 q_n}  e^{-{q_n}}\right)} \leq e^{-2q_n}~,
\]
thus we obtain
$$2q_{n + 1}\omega\left( \exp\left(-\frac{q_{n+1}}{2C_1 q_n}  e^{-{q_n}}\right)\right)  \leq 4C_1 q_n e^{-q_n} \underset{n\to\infty}{\longrightarrow} 0.$$
Finally, by \cite[equation (5.3), Lemma 5.1]{LS}, we obtain
\begin{equation*}
|J_\delta^c| \leq  \frac{2e \delta}{q_n} = \frac{2e}{q_n}e^{-\frac{q_n}{2}}~.
\end{equation*}
$J_\delta^c$ is a union of $q_n$ intervals, thus its contribution to \eqref{eq:meas-omega} is
$$ \leq
  q_n\omega\left(\frac{2e}{q_n}e^{-\frac{q_n}{2}} \right) \leq 4\omega\left(\frac{2e}{q_n}e^{-\frac{q_n}{2}} \right) \log\frac{1}{\frac{2e}{q_n}e^{-\frac{q_n}{2}}} \underset{n\to\infty}{\longrightarrow} 0,
  $$
   where the convergence follows by the assumption  \eqref{eq:hausdorff-opt-cond}. This concludes the whole proof. \qed

\section{Construction of auxiliary intervals}\label{section:prop-i+/i-} 
Let $\mathcal N_{\alpha, \lambda, \theta} (E)$ be the density of states measure corresponding to the almost Mathieu operator $\op$.  For $\alpha \in \mathbb Q$, $\mathcal N_{\alpha, \lambda, \theta}$ depends on $\theta$, therefore, similarly to Section~\ref{s:cont}, we define
\begin{equation}\label{eq:barN} \overline{\mathcal N}_{\alpha, \lambda}(E) = \frac{1}{2\pi} \int_0^{2\pi} \mathcal N_{\alpha, \lambda, \theta} (E) d\theta~. \end{equation}
For irrational $\alpha$, $\overline{\mathcal N}_{\alpha, \lambda}(E) = {\mathcal N}_{\alpha, \lambda}(E) = {\mathcal N}_{\alpha, \lambda,\theta}(E)$.

Due to Aubry--Andr\'e duality \cite{aubry-andre, aubry}, the following holds true for any $\lambda\neq 0$
\begin{equation}\label{eq:aubry-andre}
\begin{split}\overline {\mathcal N}_{\alpha,\lambda} (E) &= \overline {\mathcal N}_{\alpha,\frac{2}{\lambda}} \left(\frac{E}{\lambda}\right)~,\\
S(\alpha, \lambda) &= \lambda S\left( \alpha, \frac{2}{\lambda}\right)~.
\end{split}\end{equation}
Therefore, it is enough to prove the results for $0< \lambda \leq 1$.  For the time being, we assume that $0 < \lambda < 1$; a separate argument will be given for $\lambda = 1$.

\begin{figure}
  \centering
   \includegraphics[width=9cm]{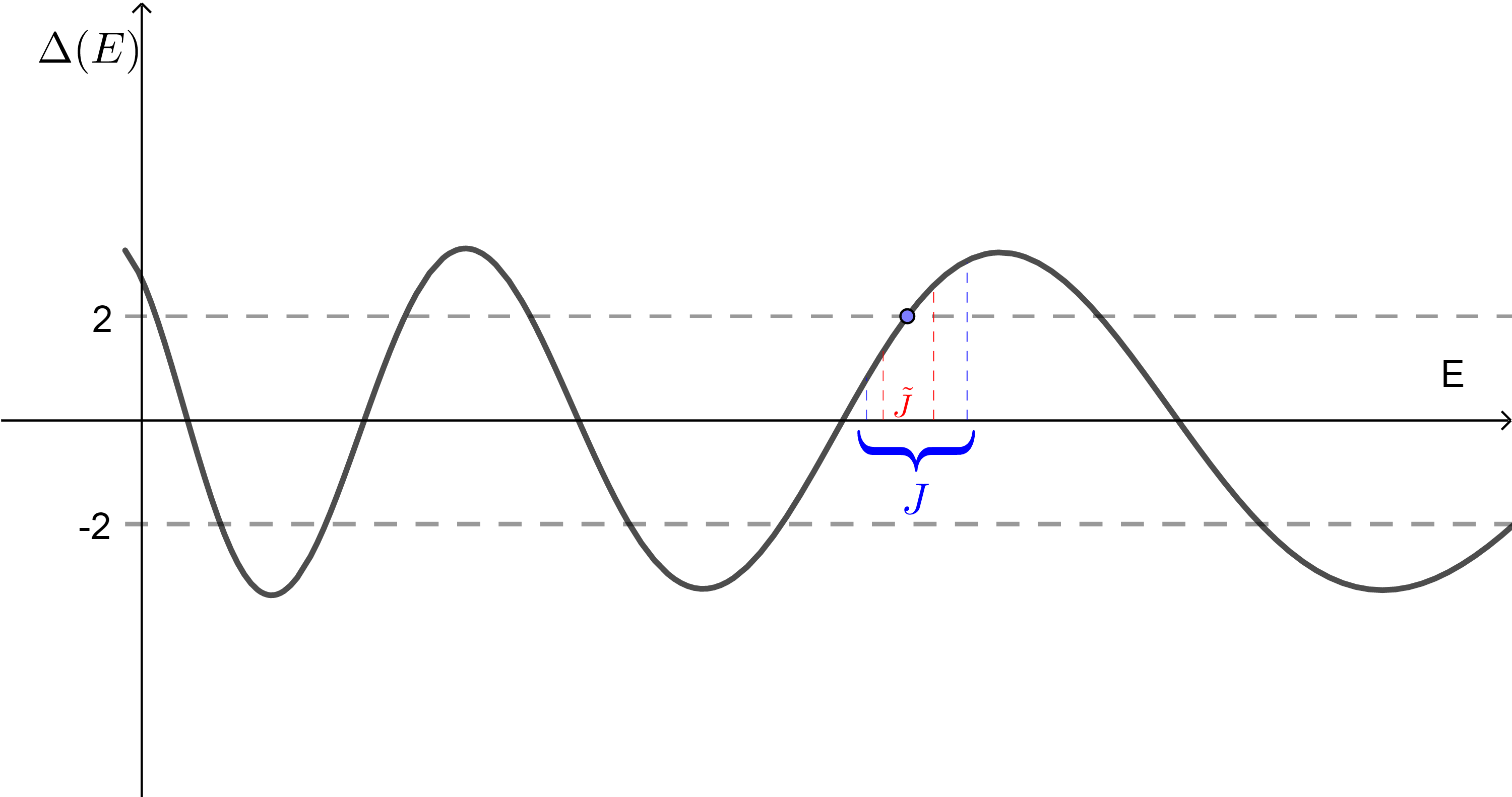} \\
  \caption{Construction of $F_\tau$}\label{f1}
\end{figure}

In this section, we construct  a family of intervals $F_\tau$ ($0< \tau \leq \frac12$) contained in $S\left(\frac p q, \lambda\right)$ but disjoint from $S_-\left(\frac p q, \lambda\right)$, as  illustrated in Figure~\ref{f1}.  Clearly, these intervals will be found close to the edges of one of the bands of $S\left(\frac p q, \lambda\right)$. We prove that these intervals are not too small,  both in the sense of Lebesgue measure (see Proposition~\ref{prop:i+/i-extended}):
\begin{equation}\label{eq:lbd-0} |F_\tau| \geq \frac{(1-\lambda)\tau}{2 q^3} \lambda^q \end{equation}
and spectrally (see Lemma~\ref{cl:ids-per}):
\begin{equation}\label{eq:lbd-1}\overline\rho_{\frac{p}{q}, \lambda} (F_\tau) \geq \frac{\tau\arccos\tau}{\pi} \,  \frac{\lambda^{\frac{q}{2}}}{q}\end{equation}
 where $\overline\rho_{\frac{p}{q}, \lambda}$ is  defined as in (\ref{eq:defbarrho}). Using Proposition~\ref{th:ky-fan}, we shall deduce from (\ref{eq:lbd-1}) that also
\begin{equation}\label{eq:lbd-2} \overline\rho_{\alpha, \lambda} (F_\tau) \geq \frac{\tau\arccos\tau}{2 \pi} \,  \frac{\lambda^{\frac{q}{2}}}{q}~. \end{equation}

While \eqref{eq:lbd-0} will play a r\^ole in the proof of Theorem~\ref{th:ids}-\eqref{thm:ids-opt}, 
Theorem~\ref{th:ids}-\eqref{eq:thm-ids1} and  \eqref{eq:thm-ids2} will follow by comparing (\ref{eq:lbd-2}) with Corollaries~\ref{cor:LS} and \ref{cor:LS1}, which  tell us that regular measures can not assign large mass to $F_\tau \subset J_\delta$. A somewhat similar argument will be used in the proof of Theorem~\ref{thm:homog}.

\medskip
The construction goes as follows. First, for  $-1 \leq \tau \leq 1$ let
\begin{equation*}
S_{\tau}\left(\frac{p}{q}, \lambda\right) = \left\{ E\in\R\, :\,\, |\Discr| \leq 2 + 2\tau\lambda^q\right\} = \bigcup_{j = 1}^q\Itau~.
\end{equation*}
These sets interpolate between $S_{-1}\left(\frac{p}{q}, \lambda\right) = S_-\left(\frac{p}{q}, \lambda\right)$ and  $S_{+1}\left(\frac{p}{q}, \lambda\right) = S\left(\frac{p}{q}, \lambda\right)$. 
For $\xi \in\{ -1, +1\}$ let
\begin{equation}\label{eq:def-itauxi}
\Itauxi = \Itau\bigcap\left\{ E\in\R :\, \mathrm{sign}\, \Discr = \xi\right\}~.
\end{equation}
Then $F_\tau$ will eventually be defined as the ``largest'' among the $2q$ intervals 
\[ \Itauxi \setminus I^j_{-\tau,\xi}~, \quad \xi \in \{-1, +1\}, \, j \in \{1, \cdots, q\}\]
(a formal and slightly more careful definition is given before (\ref{eq:J-tau}) below).


\begin{proposition}\label{prop:i+/i-extended} Let $0 < \lambda < 1$.   We have the following:
\begin{enumerate}
\item For any $1\leq j \leq q$ and any $\xi \in \{-1,+1\}$, $-1\leq\tau\leq\tau'\leq1$,
 $$\frac{|\Itausxi|}{|\Itauxi|} \geq 1 + \frac{\lambda^q}{4q^2}(\tau' - \tau)~.$$ 
\item There exist $1\leq j_0 \leq q$ and $\xi \in \{ -1, +1\}$, such that for any $-1\leq\tau\leq\tau'\leq1$
\begin{equation*}
| I^{j_0}_{\tau', \xi}| - |{I}^{j_0}_{\tau, \xi}|\geq \frac{1-\lambda}{ 4q^3}\lambda^q (\tau' - \tau)~.
\end{equation*}
\end{enumerate}
 \end{proposition}

In the proof of the proposition we need the following inequality due to Chebyshev. Recall that the Chebyshev polynomial (of the first kind) of degree $q$ is defined by
\[
T_q(\cos \theta) = \cos (q\theta),\ T_q(\cosh \theta) = \cosh (q\theta).
\]
 It is not hard to see that, for any $q \geq 1$ and any $x \geq 0$,  $T_q(1 + x) \leq \exp(q^2 x)$.

\begin{lemma}[Chebyshev, see e.g. pp. 67--68 \cite{timan}]\label{cl:remez} Let $[c, d]\subset [a,b]$ be two intervals and let $p(x)$ be a polynomial of degree $q$. Then
\begin{equation*}
\sup_{x\in [a,b]} |p(x)| \leq T_q\left(\frac{2}{t} - 1\right) \sup_{x\in [c,d]} |p(x)|
\leq \exp\left(q^2 \frac{2(1-t)}{t}\right)\sup_{x\in [c,d]} |p(x)|~, \quad t = \frac{d - c}{b - a}~.
\end{equation*}
\end{lemma}

%

\begin{proof}[Proof of Proposition~\ref{prop:i+/i-extended}] Observe that
\[
\max_{E\in\Itausxi}{|\Discr|} = 2 + 2\tau'\lambda^q, \quad \max_{E\in\Itauxi }{|\Discr|} = 2 + 2\tau\lambda^q.
\]
Invoking Lemma \ref{cl:remez} we get
\begin{equation*}
\frac{1 + \tau'\lambda^q}{1 + \tau\lambda^q}  =  \frac{\max_{E\in\Itausxi} {|\Discr|}}{\max_{E\in\Itauxi }{|\Discr|}}  \leq \exp\left(2q^2\left(\frac{|\Itausxi|}{|\Itauxi|}  - 1\right)\right)~.
\end{equation*}
\noindent Therefore 
\begin{equation}\label{eq:2}
\begin{split}
\frac{|\Itausxi|}{|\Itauxi|}&\geq 1 + \frac{1}{2q^2}\left(\ln\left(1 + \tau'\lambda^q\right) - \ln\left(1 + \tau\lambda^q\right) \right) \\& \geq 1 + \frac{\lambda^q(\tau' - \tau)}{2q^2(1 + \tau'\lambda^q)} \geq 1+ \frac{1}{4q^2}(\tau' - \tau)\lambda^q,
\end{split}
\end{equation}
where the second inequality follows from the Mean Value Theorem and the last inequality holds since $1 + \tau'\lambda^q \leq 2$. This finishes the proof of the first statement of the proposition. 

To prove the second part, we first note that since $\frac{1}{1+\epsilon}\leq 1 - \frac{\epsilon}{2}$ for any $0\leq\epsilon < 1$, and then conclude from (\ref{eq:2})
\begin{equation*}
\frac{|\Itauxi|}{|\Itausxi|} \leq 1 - \frac{1}{8q^2}\lambda^q(\tau' - \tau).
\end{equation*}


\noindent Thus we get
\[
|\Itausxi| - |\Itauxi| = |\Itausxi| \left( 1 -  \frac{|\Itauxi|}{|\Itausxi|} \right) \geq  |\Itausxi|  \frac{\lambda^q(\tau' - \tau)}{8q^2} \geq |I^{j}_{-1, \xi}|  \frac{\lambda^q(\tau' - \tau)}{8q^2}~,
\]
where the last inequality holds since for every $\tau' \geq \tau$ we have $\Itauxi\subset\Itausxi$.
As proved in \cite{simon_vanmouche_avron},  for any $0 < \lambda < 1$, 
\[
\left|\specIq \right| = 4 - 4\lambda,
\]
therefore, there exist $1\leq j_0 \leq q$ and $\xi\in\{ +1 , -1 \}$  such that  $I^{j_0}_{-1, \xi} \subset \specIq$ such that  $$|I^{j_0}_{-1, \xi}| \geq \frac{4 - 4\lambda}{2q}. $$ Thus we conclude that
\begin{equation*}
| I^{j_0}_{\tau', \xi}| - |{I}^{j_0}_{\tau, \xi}| = |I^{j_0}_{\tau', \xi}| \left( 1 -  \frac{|{I}^{j_0}_{\tau, \xi}|}{|I^{j_0}_{\tau', \xi}|} \right) \geq     \frac{1-\lambda}{ 4q^3}\lambda^q (\tau' - \tau)~.\qedhere
\end{equation*}
\end{proof}

Let $1 \leq j_0\leq q,\, \xi\in\{ + 1, -1\}$, be as in the second part of Proposition~\ref{prop:i+/i-extended}. For $0 < \tau \leq \frac{1}{2}$ let 
\begin{equation}\label{eq:J-tau}
F_\tau = I^{j_0}_{\tau, \xi}\setminus I^{j_0}_{-\tau, \xi} = (a_\tau, b_\tau), \quad |F_\tau| \geq \frac{(1 - \lambda)\tau}{2q^3}\lambda^q~. 
\end{equation}
Now we show that the mass assigned to $F_\tau$ by the ($\theta$-averaged) integrated density of states is not too small.

\begin{lemma}\label{cl:ids-per} For any $0 < \lambda < 1$,  we have
\begin{equation*}
\overline\rho_{\frac{p}{q}, \lambda} (F_\tau) \geq c_\tau \frac{\lambda^{\frac{q}{2}}}{q}, \quad c_\tau = \frac{\tau\arccos\tau}{\pi}~.
\end{equation*}
\end{lemma}

 \begin{proof}
Without loss of generality assume that $\xi = +1$, namely
 \[
 F_\tau =  I^{j_0}_{\tau, +1}\setminus I^{j_0}_{-\tau, +1} = (a_\tau, b_\tau).
 \]
Then, 
\[
\rho_{\frac{p}{q}, \lambda, \theta}(F_\tau) = \mathcal{N}_{\frac{p}{q}, \lambda, \theta}(b_\tau) -  \mathcal{N}_{\frac{p}{q}, \lambda, \theta}(a_\tau).
\] 
Consider the set $$A = \left\{ \theta\in [-\pi, \pi)\, :\,\, \cos\theta q > \tau\right\},\,\, |A| = 2\arccos\tau.$$
Using a result of Delyon-Souillard  \cite{ds} one obtains that
\begin{equation}\label{eq:ds}
D_{\frac{p}{q}, \lambda, \theta} (E) = 2\cos\left(2\pi q\,  \mathcal{N}_{\frac{p}{q}, \lambda, \theta}(E)\right),
\end{equation}
where $D_{\frac{p}{q}, \lambda, \theta} (E)$ is defined by (\ref{eq:Dpq}). Thus, for any $\theta \in A$ combining (\ref{eq:ds}), the definition (\ref{eq:def-itauxi}) of $I^j_{\tau}$, and  Chambers' formula (Proposition~\ref{prop1}), we get
\[
\begin{split}
&2\cos\left( 2\pi q\, \mathcal{N}_{\frac{p}{q}, \lambda, \theta}(b_\tau) \right) = 2 + 2\tau \lambda^q - 2\lambda^q\cos\theta q,\\&
2\cos\left( 2\pi q\, \mathcal{N}_{\frac{p}{q}, \lambda, \theta}(a_\tau) \right) = 2 - 2\tau\lambda^q - 2\lambda^q\cos\theta q~,
\end{split}
\]
namely
\[
\cos\left( 2\pi q\, \mathcal{N}_{\frac{p}{q}, \lambda, \theta}(b_\tau) \right) - \cos\left( 2\pi q\, \mathcal{N}_{\frac{p}{q}, \lambda, \theta}(a_\tau) \right)  = 2\tau\lambda^q.
\]
On the other hand
\begin{equation}\label{eq:trig}
\begin{split}
&\left| \cos\left( 2\pi q\, \mathcal{N}_{\frac{p}{q}, \lambda, \theta}(b_\tau) \right) - \cos\left( 2\pi q\, \mathcal{N}_{\frac{p}{q}, \lambda, \theta}(a_\tau) \right)\right|\\& = \left| 2\sin\left(\pi q \left( \mathcal{N}_{\frac{p}{q}, \lambda, \theta}(b_\tau) +  \mathcal{N}_{\frac{p}{q}, \lambda, \theta}(a_\tau) \right)\right)\sin\left(\pi q \left( \mathcal{N}_{\frac{p}{q}, \lambda, \theta}(b_\tau) -  \mathcal{N}_{\frac{p}{q}, \lambda, \theta}(a_\tau) \right)\right)\right|~.
\end{split}
\end{equation}
Since $|\sin x| \leq |x|$  for any $x$, we obtain
\[
\left|\sin\left(\pi q \left( \mathcal{N}_{\frac{p}{q}, \lambda, \theta}(b_\tau) -  \mathcal{N}_{\frac{p}{q}, \lambda, \theta}(a_\tau) \right)\right)\right|\ \leq \left| \pi q \left( \mathcal{N}_{\frac{p}{q}, \lambda, \theta}(b_\tau) -  \mathcal{N}_{\frac{p}{q}, \lambda, \theta}(a_\tau) \right)\right|~.
\] 
We also note by our selection $$\frac{j_0 - \frac{1}{2}}{q}\leq \mathcal{N}_{\frac{p}{q}, \lambda, \theta}(a_\tau) <  \mathcal{N}_{\frac{p}{q}, \lambda, \theta}(b_\tau) \leq \frac{j_0}{q},$$ thence we have
\[
\begin{split}
&\left|2\sin\left(\pi q \left( \mathcal{N}_{\frac{p}{q}, \lambda, \theta}(b_\tau) +  \mathcal{N}_{\frac{p}{q}, \lambda, \theta}(a_\tau) \right)\right)\right|\\& \leq 2\left|\sin\left(2\pi q\, \mathcal{N}_{\frac{p}{q}, \lambda, \theta}(a_\tau) \right)\right| = 2\sqrt{1 - \cos^2\left(2\pi q\, \mathcal{N}_{\frac{p}{q}, \lambda, \theta}(a_\tau) \right)} \\& \leq 2\sqrt{2\lambda^q(1 + \tau) - \lambda^{2q}(1 + \tau)^2} \leq 4\lambda^{\frac{q}{2}},
\end{split}
\]
where the last inequality holds since $0 < \tau \leq \frac{1}{2}$, and in the third inequality we used that
\[
\cos^2\left(2\pi q\, \mathcal{N}_{\frac{p}{q}, \lambda, \theta}(a_\tau)\right) = (1 - \tau\lambda^q - \lambda^q\cos(\theta q))^2 \geq (1 - \lambda^q(1 + \tau))^2.
\]
Consequently, we have using (\ref{eq:trig}):
\[
2\tau\lambda^q \leq  4\lambda^{\frac{q}{2}} \pi q \left( \mathcal{N}_{\frac{p}{q}, \lambda, \theta}(b_\tau) -  \mathcal{N}_{\frac{p}{q}, \lambda, \theta}(a_\tau) \right),
\]
namely
\[
 \mathcal{N}_{\frac{p}{q}, \lambda, \theta}(b_\tau) -  \mathcal{N}_{\frac{p}{q}, \lambda, \theta}(a_\tau) = \rho_{\frac{p}{q},\lambda,\theta}(F_\tau) \geq \frac{\tau}{2\pi q}\lambda^{\frac{q}{2}}.
\]
Thus we conclude that
\[
\overline\rho_{\frac{p}{q},\lambda}(F_\tau)\geq \frac{\tau}{2\pi q}\lambda^{\frac{q}{2}} |A| =\frac{\tau\arccos\tau}{\pi q}\lambda^{\frac{q}{2}}.
\]
 \end{proof}

Next, we show that if $\alpha$ is sufficiently close to $\frac{p}{q}$, a similar estimate holds for $\overline\rho_{\alpha,\lambda} (F_\tau) = \rho_{\alpha,\lambda} (F_\tau)$ in place of $\overline\rho_{\frac{p}{q},\lambda} (F_\tau)$.

\begin{lemma}\label{cl:ids} Let $0 < \lambda < 1$ and assume that $\alpha\in\R$ is such that 
\begin{equation}\label{eq:rho-cond}
\left| \alpha - \frac{p}{q}\right| \leq \frac{\tau^2\arccos\frac{\tau}{2}}{256\pi^2\, q^4} \lambda^\frac{3q}{2}.
\end{equation}
Then, the intervals $F_\tau$ corresponding to $\frac{p}{q}$ satisfy
\begin{equation*}
\rho_{\alpha,\lambda} (F_\tau) \geq  \frac{\tau\arccos\tau}{4\pi}\frac{\lambda^{\frac{q}{2}}}{q} = c_\tau \frac{\lambda^{\frac{q}{2}}}{q}.
\end{equation*}
\end{lemma}
\begin{proof} Let $\tau' = \frac{\tau}{2}$, then by Lemma~\ref{cl:ids-per}, we have 
\begin{equation}\label{eq:rho-tau'}
\overline\rho_{\frac{p}{q}, \lambda} (F_{\tau'}) \geq \frac{\tau'\arccos\tau'}{\pi}\frac{\lambda^{\frac{q}{2}}}{q} = \frac{\tau \arccos\frac\tau2}{2 \pi}\frac{\lambda^{\frac{q}{2}}}{q}.
\end{equation}

\noindent Denote $F_{\tau'} = (a_{\tau'}, b_{\tau'})\subset F_\tau = (a_\tau, b_\tau)$. Applying Proposition~\ref{prop:i+/i-extended} once with $-\frac{\tau}{2}, -\tau$ and once with $\tau, \frac{\tau}{2}$ in place of $\tau', \tau$, we obtain
\[
|a_{\tau'} - a_\tau|, |b_{\tau} - b_{\tau'}| \geq \frac{1 - \lambda}{4q^3}\lambda^q (\tau - \tau') = \frac{(1 - \lambda)\tau}{8q^3}\lambda^q .
\]
In the notations of Proposition~\ref{th:ky-fan}, let
\[
\kappa = \frac{(1 - \lambda)\tau}{8q^3}\lambda^q~, \quad
L = \frac{\kappa}{4\pi\lambda\left| \alpha - \frac{p}{q}\right|} = \frac{(1 - \lambda)\tau}{32\pi\, q^3 \lambda\left| \alpha - \frac{p}{q}\right|} \lambda^q~.
\]
Then Proposition~\ref{th:ky-fan} implies that
\[
\begin{split}
\rho_{\alpha,\lambda} (F_\tau) &\geq \overline\rho_{\frac{p}{q}, \lambda} (F_{\tau'}) - \frac{4}{L} \geq  \overline\rho_{\frac{p}{q}, \lambda} (F_{\tau'}) - \frac{128\pi \lambda}{(1 - \lambda)\tau}\frac{q^3}{\lambda^q}\left| \alpha - \frac{p}{q}\right| \\& \geq \frac{\tau \arccos\frac\tau2}{2 \pi}\frac{\lambda^{\frac{q}{2}}}{q} \left( 1 - \frac{\lambda}{1 - \lambda}\frac{\lambda^q}{q^3}\right) \geq \frac{\tau\arccos\tau}{4\pi}\frac{\lambda^{\frac{q}{2}}}{q},
\end{split}
\]
where the third inequality follows from (\ref{eq:rho-tau'}) and the assumption (\ref{eq:rho-cond}) on $\left| \alpha - \frac{p}{q}\right|$. 
\end{proof}

\section{Proof of Theorem~\ref{th:ids}: \eqref{eq:thm-ids1} and  \eqref{eq:thm-ids2}} The items \eqref{eq:thm-ids1} and  \eqref{eq:thm-ids2} of Theorem~\ref{th:ids} follow from the following proposition.

\begin{proposition}\label{prop:ids} If $\alpha\in\R\setminus\Q,\,\, 0 < \lambda < 1$, and $t > 2$ are such that
\begin{equation}\label{eq:prop-ids}
\beta(\alpha) > \max\left(\frac{3}{2}, \frac{t}{t - 1}\right)|\log\lambda|,
\end{equation}
then $\ids\notin \uc[\omega_t]$.
\end{proposition}

\begin{proof}[Proof of Theorem~\ref{th:ids}: \eqref{eq:thm-ids1} and  \eqref{eq:thm-ids2}] The case $\lambda = 1$ follows from Theorem~\ref{thm:hausdorff} and the second half of the Frostman's Lemma (Lemma~\ref{lem:frostman}), since $\mathrm{supp}\,\rho_{\alpha, \lambda} = S(\alpha, \lambda)$. By (\ref{eq:aubry-andre}) the case $\lambda > 1$ follows from the case $\lambda < 1$, therefore we focus on the latter.

If $\beta(\alpha) > \frac{3}{2}|\log\lambda|$, then (\ref{eq:prop-ids}) holds with any $t > 3$, hence $\ids\notin \bigcup_{t > 3} \uc\, [\omega_t]$, thus by Corollary~\ref{cor:log-t} we have $\gamma_{\alpha, \lambda}\notin \bigcup_{t > 4} \uc\, [\omega_t]$. If $\beta(\alpha) > 2|\log\lambda|$, then (\ref{eq:prop-ids}) holds with any $t > 2$, hence $\ids\notin \bigcup_{t > 2} \uc\, [\omega_t]$, thus by Corollary~\ref{cor:log-t} we have $\gamma_{\alpha, \lambda}\notin \bigcup_{t > 3} \uc\, [\omega_t]$. \end{proof}

\begin{proof}[Proof of Proposition~\ref{prop:ids}] Fix $t > 2$ satisfying (\ref{eq:prop-ids}) and let
\begin{equation}\label{eq:r}
r = \frac{1}{2(t - 1)}\max(0, 3 - t)|\log\lambda| + \nu,
\end{equation}
where $\nu > 0$ is sufficiently small to ensure that $\beta (\alpha) - r >  \frac{3}{2}|\log\lambda|$. Then we can find a sequence of rationals $\frac{P_k}{Q_k}\rightarrow\alpha$ such that
\begin{equation}\label{eq:prop-ids1}
\left| \alpha - \frac{P_k}{Q_k}\right| \leq e^{-(\beta(\alpha) - r)Q_k}.
\end{equation}
For sufficiently large $k$,  we have $Q_k \geq q_0(r, \alpha)$, and the choice of $r$ guarantees that $\frac{p}{q} = \frac{P_k}{Q_k}$ satisfies
\begin{equation*}
\left| \alpha - \frac{p}{q}\right| \leq \frac{\tau^2\arccos \tau}{256\pi^2\, q^4}\lambda^{\frac{3q}{2}},\,\quad \text{with}\quad \tau = \frac{1}{2}.
\end{equation*}
Define $F_\frac{1}{2}$ corresponding to one such $\frac{p}{q}$ using (\ref{eq:J-tau}). Then, Lemma~\ref{cl:ids} gives
\begin{equation}\label{eq:prop-ids3}
\rho_{\alpha, \lambda} (F_\frac{1}{2}) \geq \frac{\lambda^{\frac{q}{2}}}{20\, q}.
\end{equation}
\begin{claim*} If $\ids\in \uc\, [\omega_t]$ for some $t > 2$, then
\begin{equation}\label{eq:upper-bound}
\rho_{\alpha, \lambda} (F_\frac{1}{2}) \leq C_t\, q^{2t}\, \lambda^{\frac{q}{2}(t - 2)}\, e^{-rq(t - 1)}. 
\end{equation}
\end{claim*}
For sufficiently large $q$,  the lower bound (\ref{eq:prop-ids3}) with our choice (\ref{eq:r}) of $r$ contradicts the upper bound (\ref{eq:upper-bound}), therefore the assumption  $\ids\in \uc\, [\omega_t]$ for some $t > 2$ can not hold.

It remains to prove the claim.

The inequality (\ref{eq:prop-ids1}) ensures that the assumptions of Propositioin~\ref{prop:LS} hold for $\delta = \lambda^q$ and  $\frac{\widetilde p}{\widetilde q} = \frac{P_m}{Q_m}$ when $m > k$. Therefore, if $\ids\in \uc\, [\omega_t]$, then Corollary~\ref{cor:LS} and  part (i) of Remark~\ref{log-t} imply that
\begin{equation*}
\begin{split}
\rho_{\alpha, \lambda}\left( S\left( \frac{P_m}{Q_m}, \lambda\right)\bigcap F_\frac{1}{2}\right) &\leq \rho_{\alpha, \lambda}\left( S\left( \frac{P_m}{Q_m}, \lambda\right)\bigcap J_{\lambda^q}\right) \leq C_t\, q^{2t}\, \lambda^{\frac{q}{2}(t - 2)}\, e^{-rq(t - 1)},
\end{split}
\end{equation*}
where $ J_{\lambda^q}$ is the set $J_\delta$ defined by (\ref{eq:jdelta}) for $\delta = \lambda^q$. 

On the other hand, by the continuity of the spectrum (Proposition~\ref{continuity_of_spectra})
\[
S(\alpha, \lambda) \subset S\left( \frac{P_m}{Q_m}, \lambda\right) + 6\sqrt{2\lambda}\left( -\sqrt{\left|\alpha -   \frac{P_m}{Q_m}\right|}, +\sqrt{\left|\alpha -   \frac{P_m}{Q_m}\right|} \right).
\]
Hence $S(\alpha, \lambda)\setminus S\left( \frac{P_m}{Q_m}, \lambda\right)$ can be covered by (at most) $2Q_m$ intervals of length $\leq 12\sqrt{2\lambda}\sqrt{\left|\alpha -   \frac{P_m}{Q_m}\right|}$. Therefore, if  $\ids\in \uc\, [\omega_t]$, then
\[
\rho_{\alpha, \lambda}\left( S(\alpha, \lambda)\setminus S\left( \frac{P_m}{Q_m}, \lambda\right)\right) \leq C\frac{Q_m}{\log^t\frac{1}{\left|\alpha -   \frac{P_m}{Q_m}\right|}} 
\leq \frac{C Q_m^{1-t}}{(\beta(\alpha)-r)^t}~,
\]
where the last inequality follows from (\ref{eq:prop-ids1}).
Thus,
\[
\begin{split}
\rho_{\alpha, \lambda} (F_\frac{1}{2}) &= \rho_{\alpha, \lambda} (S(\alpha, \lambda)\cap F_\frac{1}{2}) 
\leq  C_t\, q^{2t}\, \lambda^{\frac{q}{2}(t - 2)}\, e^{-rq(t - 1)} +\frac{C Q_m^{1-t}}{(\beta(\alpha)-r)^t}~.
\end{split}
\]
Letting $m\rightarrow\infty$ and recalling that $t > 2$ we obtain (\ref{eq:upper-bound}). 
\end{proof}

\section{Proof of Theorem \ref{thm:homog}}

\subsection{Non-homogeneity of the spectrum}

We deduce Theorem~\ref{thm:homog}-(\ref{thm-homo}) from the following   Proposition~\ref{prop:log-homog}, which involves a relaxed notion of homogeneity which we now introduce.
\begin{defin}\label{def:loghomog} Let $\omega(s)$ be a gauge function. A set $S\subset\R$ is called  $\omega$-homogeneous if there exist $ \kappa > 0$ and $\epsilon_0 > 0$ such that for any $E\in S$ and for any $0< \epsilon \leq\epsilon_0$, we have 
\begin{equation}\label{eq:log-homog}
\omega(\left| (E - \epsilon, E + \epsilon)\cap S\right|)\geq \kappa\,\epsilon~.
\end{equation}
\end{defin}


\begin{proposition}\label{prop:log-homog} For any  $\alpha\in \R\setminus\Q$, $e^{-\frac{2\beta(\alpha)}{3}}<\lambda <1$, the spectrum $\specU$ is not $\omega_2$-homogeneous.
 \end{proposition}
 
\begin{proof}[Proof of Theorem~\ref{thm:homog}-(\ref{thm-homo}) ] If $\lambda = 1$, then by the results of Avila--Krikorian \cite{AK} and Last \cite{last3} the measure of the spectrum $|S(\alpha, \lambda)| = 0$ for any irrational $\alpha$. Thus  $S(\alpha, \lambda)$ is not $\omega_2$-homogeneous and not homogeneous. If $e^{-\frac{2\beta(\alpha)}{3}} < \lambda < 1$, then by Proposition~\ref{prop:log-homog} the spectrum $S(\alpha, \lambda)$ is not $\omega_2$-homogeneous, thus not homogeneous. By (\ref{eq:aubry-andre}), if $1 < \lambda < e^{\frac{2\beta(\alpha)}{3}}$ the spectrum  $S(\alpha, \lambda)$ is not $\omega_2$-homogeneous as well, thus not homogeneous.
\end{proof}

\begin{proof}[Proof of Proposition~\ref{prop:log-homog}] We start similarly to the proof of Proposition~\ref{prop:ids}. Choose $r > 0$ sufficietly small to ensure that $\beta(\alpha) - r > \frac{3}{2}|\log\,\lambda|$. Then choose a sequence $\frac{P_k}{Q_k}\rightarrow\alpha$ satisfying (\ref{eq:prop-ids1}), namely $\left|\alpha -  \frac{P_k}{Q_k}\right| \leq e^{-(\beta(\alpha) - r)Q_k}$. For sufficiently large $k$, we have $Q_k \geq q_0(r, \alpha)$, and the choice of $r$ guarantees that $\frac{p}{q} = \frac{P_k}{Q_k}$ satisfies
\begin{equation}\label{eq:prop-log-homog1}
\left|\alpha -  \frac{p}{q}\right| \leq \frac{\tau^2\, \arccos\,\tau}{256\, \pi^2\, q^4}\lambda^{\frac{3q}{2}},\quad \text{with}\quad \tau = \frac{1}{4}~.
\end{equation}
We shall use the sets $F_\frac{1}{2} = (a, b)$ and $F_\frac{1}{4} = (\widetilde a, \widetilde b)$ from (\ref{eq:J-tau}).  Applying the second part of Proposition~\ref{prop:i+/i-extended} once with $\tau' = -\frac{1}{4}, \tau = -\frac{1}{2}$, and once with $\tau' = \frac{1}{2}, \tau = \frac{1}{4}$, we obtain
\[
|b - \widetilde b|, |\widetilde a - a| \geq \frac{1 - \lambda}{16\, q^3}\lambda^q~.
\]
By (\ref{eq:prop-log-homog1}) the condition of Lemma~\ref{cl:ids} holds true, therefore applying Lemma~\ref{cl:ids} with $\tau = \frac{1}{4}$, we get
\[
\rho_{\alpha, \lambda} (F_\frac{1}{4}) \geq \frac{\lambda^\frac{q}{2}}{10\, q} > 0~,
\]
and in particular $F_\frac{1}{4} \cap S(\alpha, \lambda) \neq \varnothing$. Let
\[
E\in F_\frac{1}{4} \cap S(\alpha, \lambda),\quad \epsilon = \frac{1 - \lambda}{16\, q^3}\lambda^q,\quad \frac{\widetilde p}{\widetilde q} = \frac{P_m}{Q_m},\,\, m > k~,
\]
and note that $(E - \epsilon, E + \epsilon)\subset F_\frac{1}{2}\subset J_{\lambda^q}$.
By Corollary~\ref{cor:LS} (and Remark~\ref{log-t}) applied to the Lebesgue measure with $\widetilde\omega_1(\epsilon) = \epsilon$ so that $(W_3\, \widetilde \omega_1) (\epsilon)\leq C\epsilon$, we obtain
\[
\begin{split}
\left| (E - \epsilon, E + \epsilon) \bigcap S\left( \frac{P_m}{Q_m}, \lambda\right) \right|
& \leq \left| F_\frac{1}{2}\bigcap S\left( \frac{P_m}{Q_m}, \lambda\right) \right|\leq \left| J_\delta \bigcap S\left( \frac{P_m}{Q_m}\right)\right|\\&\leq \frac{C\, q^2}{\lambda^\frac{q}{2}}\exp\left(-\frac{e^{rq}}{15000 q^2 \lambda^{\frac q 2}}\right)~.
\end{split}
\]
By the continuity of the spectrum (Proposition~\ref{continuity_of_spectra}), we have
\[    
S(\alpha, \lambda) \subset  S\left( \frac{P_m}{Q_m}, \lambda\right)  + \left(-6\sqrt2\sqrt{\left| \alpha -  \frac{P_m}{Q_m}\right|}, +6\sqrt2\sqrt{\left| \alpha -  \frac{P_m}{Q_m}\right|}  \right)~.
\]
Since the set $S(\alpha, \lambda)\setminus S\left( \frac{P_m}{Q_m}, \lambda\right)$ is contained in the union of (at most) $2Q_m$ intervals and since $\left|\alpha -  \frac{P_m}{Q_m}\right| \leq e^{-(\beta(\alpha) - r)Q_m}$, we conclude that
\[
\begin{split}
|(E - \epsilon, E + \epsilon)\cap S(\alpha, \lambda)|  &\leq \left| (E - \epsilon, E + \epsilon) \bigcap S\left( \frac{P_m}{Q_m}, \lambda\right) \right| + 2Q_m\, 6\sqrt2\sqrt{\left| \alpha -  \frac{P_m}{Q_m}\right|}\\& \leq \frac{C\, q^2}{\lambda^\frac{q}{2}}\exp\left(-\frac{e^{rq}}{15000 q^2 \lambda^{\frac q 2}}\right) + 2Q_m\, e^{-(\beta(\alpha) - r)Q_m/2}~.
\end{split}
\]
The second term tends to zero as $Q_m\rightarrow\infty$, thus we conclude that
\[
|(E - \epsilon, E + \epsilon)\cap S(\alpha, \lambda)|\leq  \frac{C\, q^2}{\lambda^\frac{q}{2}}\exp\left(-\frac{e^{rq}}{15000 q^2 \lambda^{\frac q 2}}\right) \leq \exp\left(-e^\frac{rq}{2}\, \lambda^{-\frac{q}{2}}\right)~,
\]
hence
\[
\frac{1}{\log^2\frac{1}{|(E - \epsilon, E + \epsilon)\cap S(\alpha, \lambda)|}} \leq e^{-rq}\, \lambda^q < \kappa  \frac{1 - \lambda}{16\, q^3}\lambda^q  = \kappa\,\epsilon~,
\]
for any $\kappa > 0$ and sufficiently large $q$. Thus the spectrum is not $\omega_2$-homogeneous.
\end{proof} 

\subsection{Failure of Parreau--Widom condition} By (\ref{eq:aubry-andre}) we only need to prove Theorem~\ref{thm:homog}-(\ref{thm-pw}) for $0 < \lambda \leq 1$. Note that a Parreau--Widom set is always of positive Lebesgue measure \cite[pp.6]{chris2} thus in the case $\lambda = 1, \beta(\alpha) > 0$ Theorem~\ref{thm:homog}-(\ref{thm-pw}) follows from results of Avila--Krikorian \cite{AK}, Last \cite{last3} that the measure of the spectrum $|S(\alpha, \lambda)| = 0$ for any irrational $\alpha$. Therefore, from this point we assume that $0 < \lambda < 1$.

According to \cite[Theorem 1.15]{simon-pw}  the Green function $g(z)$ of the Dirichlet Laplacian in $\C\setminus\specU$ is equal to the Lyapunov exponent $\gamma_{\alpha, \lambda} (z)$ that corresponds to $H_{\alpha, \lambda,\theta}$.  To prove that the Parreau--Widom condition fails, we shall use Proposition~\ref{prop:LS} to show that at most of the points in the intervals $F_{\tau}$ which we constructed above,  the Lyapunov exponent is not too small ($\geq \operatorname{const} q^{-2} \lambda^{\frac{q}{2}}$). On the other hand, we show that the spectrum of $\op$ in $F_\tau$ has many ($\geq \operatorname{const} q^{-3} \lambda^{\frac{-(1+r)q}{2}}$) gaps. Therefore the contribution of $F_\tau$ to the sum
(\ref{eq:pw-cond}) is lower-bounded by a quantity which diverges as $q \to \infty$. 

To implement this plan, we first recall that  $\beta(\alpha) + 3\log\lambda > 0$, therefore there exists $r > 0$ and a sequence $\frac{P_k}{Q_k}\rightarrow\alpha$ such that
\begin{equation}\label{eq:pw1}
\left|\alpha -  \frac{P_k}{Q_k}\right| \leq e^{-(\beta(\alpha) - r)Q_k} \leq \lambda^{(3 + r)Q_k}~.
\end{equation}
Let $F_\frac{1}{2}$ be as defined in (\ref{eq:J-tau}) for $\frac{p}{q} =  \frac{P_k}{Q_k}$, then by (\ref{eq:J-tau}) we have 
\begin{equation}\label{f12}|F_\frac{1}{2}| \geq \frac{1 - \lambda}{4\, q^3}\lambda^q.
\end{equation} 
Using the continuity of the spectrum (Proposition~\ref{continuity_of_spectra}) and (\ref{eq:pw1}) we obtain for $m > k$ that for any $E\in F_\frac{1}{2}$ there exists $E'\in S\left( \frac{P_m}{Q_m}, \lambda\right)$ such that
\begin{equation}\label{eq:pw2}
|E- E'| < 6\sqrt{2\lambda\left| \frac{p}{q} - \frac{P_m}{Q_m}\right|} \leq 12  \lambda^{\frac{(3 + r)q}{2}}~.
\end{equation}
Therefore, by \eqref{f12} we can find a collection $\mathcal{J}_q$ of disjoint closed intervals $\mathbf{j}\subset F_\frac{1}{2}$ of length $|\mathbf{j}| = 225\,  \lambda^{\frac{(3 + r)q}{2}}$, so that
\begin{equation}\label{eq:pw3}
\#\mathcal{J}_q \geq \left\lfloor\frac{\frac{1 - \lambda}{4\, q^3}\lambda^q}{225\lambda^{\frac{(3 + r)q}{2}}}\right\rfloor
\geq \frac{(1 - \lambda)}{1000q^3}\lambda^{-\frac{(1 + r)q}{2}}~.
\end{equation}
Let $\mathbf{j}\in \mathcal{J}_q$ be an interval and divide it into nine equal closed subintervals $\mathbf{j}_1, \dots, \mathbf{j}_9$ of length $$|\mathbf{j}_k|=\frac{225}{9}\lambda^{\frac{(3 + r)q}{2}} = 25 \lambda^{\frac{(3 + r)q}{2}}, \qquad k=1,2,\cdots 9.$$ Then, by (\ref{eq:pw2}) applied to the centers of $\mathbf{j}_2, \mathbf{j}_8$ there exist $s_2\in S\left( \frac{P_m}{Q_m}, \lambda\right)\bigcap \mathbf{j}_2$ and $s_8\in S\left( \frac{P_m}{Q_m}, \lambda\right)\bigcap \mathbf{j}_8$. Another use of Proposition~\ref{continuity_of_spectra} and (\ref{eq:pw1}) imply the existence of $s'_2, s'_8 \in S(\alpha, \lambda)$ such that
\[ 
|s_2' - s_2|, |s_8' - s_8| \leq 6 \left(2\lambda \left|\alpha - \frac{P_m}{Q_m}\right|\right)^{\frac12}
\leq 12 \lambda^{\frac{(3 + r)Q_m}{2}}<  \frac{|\mathbf{j}_k|}{2}~,
\]
and in particular
\begin{equation}\label{eq:Salpha-cup-j123789}
\specU\cap (\mathbf{j}_1\cup\mathbf{j}_2\cup\mathbf{j}_3),\, \, \specU\cap (\mathbf{j}_7\cup\mathbf{j}_8\cup\mathbf{j}_9) \neq \varnothing~.
\end{equation}
On the other hand, we can apply Proposition~\ref{prop:LS} to
\[
\frac{p}{q} = \frac{P_k}{Q_k},\quad \frac{\widetilde p}{\widetilde q} = \frac{P_m}{Q_m}, \, m>k, \quad \delta = \lambda^q,\quad \epsilon = \exp\left(-\frac{e^{rq}}{15000 q^2 \lambda^{\frac q2}}\right)~.
\]
It yields that for any $E\in J_{\lambda^q}$, where $J_{\lambda^q}$ is the set $J_\delta$ defined by (\ref{eq:jdelta}) for $\delta = \lambda^q$, and for any $\theta$
\[
\gamma_{ \frac{P_m}{Q_m}, \lambda, \theta} (E + i\epsilon) \geq \frac{\delta}{9600\, q^2\, \lambda^\frac{q}{2}} = \frac{\lambda^\frac{q}{2}}{C_1\, q^2}~.
\]
In particular, the same lower bound holds true for the $\theta$-averaged Lyapunov exponent $\bar\gamma_{ \frac{P_m}{Q_m}, \lambda} (E + i\epsilon)$. Applying Proposition~\ref{prop-surace} to the Lebesgue measure with $\widetilde\omega_1(\epsilon) = \epsilon$ so that $(W_3\, \widetilde\omega_1) (\epsilon)\leq C\epsilon$ as in Remark~\ref{log-t}, with $\xi =  \frac{\lambda^\frac{q}{2}}{C_1\, q^2}$, we obtain
\begin{equation}\label{eq:pw4}
\left|\left\{ E\, |\, \, \left|\bar\gamma_{ \frac{P_m}{Q_m}, \lambda}\left(E + i\epsilon\right) - \bar\gamma_{ \frac{P_m}{Q_m}, \lambda}\left(E\right)\right|\geq \frac{\lambda^{\frac{q}{2}}}{2C_1\,q^2} \right\}\right|\leq \frac{ C_2\, q^2}{\lambda^{\frac{q}{2}}}\epsilon~.
\end{equation}
Since $|\mathbf{j}_5| = 25\lambda^{\frac{(3 + r)q}{2}} \geq  \frac{ C_2\, q^2}{\lambda^{\frac{q}{2}}}\epsilon$, (\ref{eq:pw4}) implies that there exists $E = E\left( \frac{P_m}{Q_m}, \mathbf{j} \right) \in \mathbf{j}_5$ such that
\[
 \bar\gamma_{ \frac{P_m}{Q_m}, \lambda}\left( E\left( \frac{P_m}{Q_m}, \mathbf{j} \right) \right) \geq \frac{\lambda^{\frac{q}{2}}}{4\,C_1q^2} ~.
\]
Thus, we get using (\ref{eq:pw3})
\begin{equation}\label{eq:pw5}
\begin{split}
\sum_{\mathbf{j}\in\mathcal{J}_q} \bar\gamma_{ \frac{P_m}{Q_m}, \lambda}\left(E\left(\frac{P_m}{Q_m}, \mathbf{j}\right)\right)& \geq   \frac{\lambda^{\frac{q}{2}}}{4\,C_1q^2} \# \mathcal{J}_q \geq \frac{\lambda^{\frac{q}{2}}}{4\,C_1q^2} \frac{1 - \lambda}{1000\, q^3}\lambda^{-\frac{(1 + r)q}{2}} \\&\geq \frac{1 - \lambda}{4000\, C_1\, q^5}\lambda^{-\frac{rq}{2}}~.
\end{split}
\end{equation}
This holds for $\frac{P_m}{Q_m}$ for any $m > k$. By compactness we can choose a subsequence $\frac{P_{m_l}}{Q_{m_l}}\rightarrow\alpha$, $Q_{m_l} > q$, such that for any interval $\mathbf{j}\in \mathcal{J}_q$ the sequence of points $E\left( \frac{P_{m_l}}{Q_{m_l}}, \mathbf{j}\right)$ converges to a number which we denote by $E(\alpha, \mathbf{j})$.

By a result of Bourgain-Jitomirskaya \cite{BJ}, the $\theta$-averaged Lyapunov exponent is jointly continuous in $(E, \alpha)$ for every $E$ and every irrational $\alpha$, thence
\[
 \lim_{Q_{m_l}\to\infty}\bar\gamma_{ \frac{P_{m_l}}{Q_{m_l}}, \lambda}\left(E\left(\frac{P_{m_l}}{Q_{m_l}}, \mathbf{j}\right)\right) = \bar\gamma_{\alpha, \lambda}(E(\alpha, \mathbf{j})) = \gamma_{\alpha, \lambda}(E(\alpha, \mathbf{j})) ,
\]
in particular, 
\begin{equation}\label{eq:pw6}
\gamma_{\alpha, \lambda}(E(\alpha, \mathbf{j})) \geq \frac{\lambda^{\frac{q}{2}}}{4C_1\,q^2 } > 0. 
\end{equation}
Since by (\ref{eq:bj-formula})  the Lyapunov exponent is equal to zero on the spectrum, we obtain that $E(\alpha, \mathbf{j}) \notin \specU$. Since $\mathbf{j}_5$ is closed we conclude that  $E(\alpha, \mathbf{j})\in {\mathbf{j}}_5$, and by (\ref{eq:Salpha-cup-j123789}) we obtain that different energies $E(\alpha, \mathbf{j})$ lie in different gaps of $\op$ (in particular, all these energies are distinct). From (\ref{eq:pw6}) we obtain as in (\ref{eq:pw5})
\[
\sum_{\mathbf{j}\in\mathcal{J}_q} \gamma_{ \alpha, \lambda}\left(E\left(\alpha, \mathbf{j}\right)\right) \geq   \frac{\lambda^{\frac{q}{2}}}{4\,C_1q^2} \# \mathcal{J}_q \geq \frac{1 - \lambda}{4000\, C_1\, q^5}\lambda^{-\frac{rq}{2}}~.
\] 
In particular, the left-hand side of (\ref{eq:pw-cond}) is bounded from below by this expression. Since this holds for any $q = Q_k$, we get that this sum is not bounded, therefore the Parreau--Widom condition fails. \qed

\section{Proof of Theorem~\ref{th:ids}-\eqref{thm:ids-opt}} If $ \ids\notin\uc\,[\omega]$,   then  by Corollary~\ref{cor:log-t} we have $\gamma_{\alpha, \lambda}\notin \bigcup_{t > 2} \uc\, [\omega_t]$. Thus we only need to prove that   $ \ids\notin\uc\,[\omega]$.  The case $\lambda = 1$ follows from Theorem~\ref{thm:hausdorff}, \eqref{thm:hausdorff-opt} and the second half of the Frostman's Lemma (Lemma~\ref{lem:frostman}), since $\mathrm{supp}\,\rho_{\alpha, \lambda} = S(\alpha, \lambda)$. By (\ref{eq:aubry-andre}) the case $\lambda > 1$ follows from the case $\lambda < 1$, therefore we focus on the latter.

We construct a sequence $\left\{ \frac{p_n}{q_n}\right\}$ of convergents of $\alpha$ (arising from the continuous fraction expansion) as follows. For any fixed $n_0$ the terms $q_1, \dots, q_{n_0}$ can be chosen in an arbitrary way. If $q_1, \dots, q_{n}, \, n > n_0$, are given, we find $x_n > 0$ such that for any $0 < x < x_n$
\begin{equation}\label{eq:ids-opt-cond1}
\omega(x)\log\frac{1}{x} \leq \frac{\delta}{q_n^2} = \frac{\lambda^{q_n}}{q_n^2} ~,
\end{equation}
where we set $\delta = \lambda^{q_n}$, and we choose
\begin{equation}\label{eq:ids-opt-cond2}
q_{n+1} \geq \max\left( q_n^6\lambda^{-2q_n}, \frac{C' q_n}{\lambda^{q_n}}\log\frac{5}{x_n}\right)~,
\end{equation}
where $C' > 0$ is the constant from Lemma~\ref{lem:ids-opt} below. Denote the set of $\alpha\in\R\setminus\Q$ for which the sequence of convergents satisfies \eqref{eq:ids-opt-cond2} by $\Omega(\lambda)$. Clearly, this set is $G_\delta$-dense in $\mathbb R$.

\begin{lemma}\label{lem:ids-opt} If $\alpha\in \Omega(\lambda)$,  then the sets $F_\frac{1}{2}=(a,b)$ defined for $\frac{p}{q} =  \frac{p_n}{q_n}$ as  in (\ref{eq:J-tau}) satisfy:
\begin{enumerate}
\item $\left|S\left(\frac{p_{n+1}}{q_{n+1}},\lambda\right) \bigcap F_\frac{1}{2} \right| \leq \exp\left(-\frac{q_{n + 1}\lambda^{q_n}}{C' q_n} \right)\leq \frac{1}{100} |F_\frac{1}{2}|$~.
\item $ S\left(\frac{p_{n+1}}{q_{n+1}},\lambda\right) \bigcap F_\frac{1}{2}$ is $\frac{1}{100} |F_\frac{1}{2}|$-dense in $\left(a + \frac{|F_\frac{1}{2}|}{100}, b - \frac{|F_\frac{1}{2}|}{100} \right)$, i.e.\
\[ \forall  E\in \left(a + \frac{|F_\frac{1}{2}|}{100}, b - \frac{|F_\frac{1}{2}|}{100} \right) \,\, \exists E' \in  S\left(\frac{p_{n+1}}{q_{n+1}},\lambda\right) \bigcap F_\frac{1}{2}: \,\, |E - E'|  \leq \frac{1}{100} |F_\frac{1}{2}|~.\]
\end{enumerate}
\end{lemma}
\begin{proof} First, by construction 
\begin{equation}\label{eq:cond-F}
F_\frac{1}{2}\subset  S\left(\frac{p_{n}}{q_{n}},\lambda\right), \,\, |F_\frac{1}{2}| \geq \frac{1 - \lambda}{4q_n^3}\lambda^{q_n}~. 
\end{equation}
In the notation of Proposition \ref{prop:LS1} let $\delta = \lambda^{q_n}$, $r = |\ln \lambda|$. Then from (\ref{eq:ids-opt-cond2}) we obtain
 \begin{eqnarray*}
q_{n+1} \delta \geq \lambda^{q_n} q_n^6\lambda^{-2q_n} \geq \lambda^{-q_n} = e^{r q_n}~,
\end{eqnarray*}
thus Corollary~\ref{cor:LS1} is applicable and we obtain
\begin{equation*}
\begin{split}
\left| S\left(\frac{p_{n+1}}{q_{n+1}},\lambda\right) \bigcap F_\frac{1}{2}\right|  &\leq \left| S\left(\frac{p_{n+1}}{q_{n+1}},\lambda\right) \bigcap J_{\lambda^{q_n}}\right| \leq \frac{Cq_n^2}{\lambda^{q_n}} \exp\left(-\frac{q_{n+1}\lambda^{q_n}}{C_1q_n} \right) \\ &\leq  \exp\left(-\frac{q_{n+1}\lambda^{q_n}}{C' q_n}\right) \leq\frac{1}{100} |F_\frac{1}{2}|~,
\end{split}
\end{equation*}
where $C' > 0$ is a constant, and on the last step we have used (\ref{eq:ids-opt-cond2}) and (\ref{eq:cond-F}). This proves the first statement.

By continuity of the spectrum (Proposition~\ref{continuity_of_spectra}) for any 
\[ E\in \left(a + \frac{|F_\frac{1}{2}|}{100}, b - \frac{|F_\frac{1}{2}|}{100} \right) \subset S\left( \frac{p_n}{q_n}, \lambda \right) \]
there exists 
$E' \in  S\left(\frac{p_{n+1}}{q_{n+1}},\lambda\right)$
such that
\begin{equation}\label{eq:cont-of-spqn}
\begin{split}
|E - E'|  <  6\left(2\lambda\left| \frac{p_{n+1}}{q_{n+1}} - \frac{p_n}{q_n}\right| \right)^{\frac{1}{2}} = \frac{6\sqrt{2\lambda}}{\sqrt{q_{n+1}q_n}}  \leq \frac{|F_\frac{1}{2}|}{100}, 
\end{split}
\end{equation}
where the last inequality follows from \eqref{eq:ids-opt-cond2} and \eqref{eq:cond-F} . Thus, \eqref{eq:cont-of-spqn} implies that $E' \in  S\left(\frac{p_{n+1}}{q_{n+1}},\lambda\right) \bigcap F_\frac{1}{2}$, therefore we conclude that $S\left(\frac{p_{n+1}}{q_{n+1}},\lambda\right) \bigcap F_\frac{1}{2}$ is $\frac{|F_\frac{1}{2}|}{100}$-dense in the interval $\left(a + \frac{|F_\frac{1}{2}|}{100}, b - \frac{|F_\frac{1}{2}|}{100} \right)$.
\end{proof}

We also need the following elementary claim.
\begin{claim}\label{cl:num-intervals} Let $A\subset (a, b)$ be a finite union of intervals such that
\begin{enumerate}
\item $|A|\leq \tau |b - a|$, for some $0\leq\tau\leq\ \frac{1}{3}$.
\item $A$ is $\tau |b - a|$-dense in $(a + \tau |b - a|, b - \tau |b - a|)$. 
\end{enumerate}
Then the number of intervals comprising $A$ is  $\geq\frac{1 - 2\tau}{3\tau}$.
\end{claim}
\begin{proof} Let $A = \biguplus_{c\in C}(c - a_c,\, c + a_c)$. Then, by (2)
\[
(a, b)\subset\bigcup_{c\in C} \left( c - a_c - \tau |b - a|,\,  c + a_c + \tau |b - a|\right). 
\]
Therefore, (1) yields
\[
(1 - 2\tau)|b - a| \leq  \sum_{c\in C} (2a_c + 2\tau |b - a|) \leq 3\tau |b - a|\# \{c\}, 
\]
namely, $\# C = \# \text{$\openrm$ intervals\, in\, $A$$\closerm$} \geq \frac{1 - 2\tau}{3\tau}$.
\end{proof}
By Lemma~\ref{lem:ids-opt} the conditions of Claim~\ref{cl:num-intervals} hold with $\tau = \frac{1}{100}$ and we obtain
\[
\# \left( \text{bands\, of\, }  S\left(\frac{p_{n+1}}{q_{n+1}},\lambda\right) \,\,  \text{lying entirely\, in\, } F_\frac{1}{2} \right) \geq \frac{1 - 2\tau}{3\tau} - 2 \geq 30.
\]
\noindent Let $\bar{I}\subset F_\frac{1}{2}$ be one band of $ S\left(\frac{p_{n+1}}{q_{n+1}},\lambda\right)$. From  the first part of Lemma~\ref{lem:ids-opt} we obtain
\[
|\bar{I}| \leq  \exp\left(-\frac{q_{n+1}\lambda^{q_n}}{C' q_n}\right)~.
\]
Let us define an extended interval $\bar{I}_+$ as the Minkowski sum
\[
\bar{I}_+ = \bar{I} + \left(-2  \exp\left(-\frac{q_{n+1}\lambda^{q_n}}{C' q_n}\right), + 2 \exp\left(-\frac{q_{n+1}\lambda^{q_n}}{C' q_n}\right) \right).
\]
Then,  
\[|\bar{I}_+| \leq  5 \exp\left(-\frac{q_{n+1}\lambda^{q_n}}{C' q_n}\right).\] 
 Let $\overline\rho_{\frac{p_{n+1}}{q_{n+1}},\lambda}$ be the $\theta$-averaged density of states defined as in (\ref{eq:defbarrho}), corresponding to $\{ H_{\frac{p_{n+1}}{q_{n+1}}, \lambda, \theta}\}_\theta$. Since $\bar I$ is a band of $ S\left(\frac{p_{n+1}}{q_{n+1}},\lambda\right)$ we have
\[
\overline\rho
_{\frac{p_{n+1}}{q_{n+1}},\lambda} (\bar{J}) = \frac{1}{q_{n+1}}. 
\]
Now apply  Propsition~\ref{th:ky-fan} with $\kappa=2 \exp\left(-\frac{q_{n+1}\lambda^{q_n}}{C' q_n}\right)$ and  $L = 10 q_{n+1}$. The condition (\ref{eq:ids-opt-cond2}) implies that $|\alpha - \frac {p_{n+1}}{q_{n+1}}| \leq \lambda^{2q_{n+1}}$, hence (for sufficiently large $n$)
\[
\kappa=2 \exp\left(-\frac{q_{n+1}\lambda^{q_n}}{C' q_n}\right) \geq 4 \pi\lambda \left|\alpha - \frac{p_{n+1}}{q_{n+1}}\right| L~.
\] 
Using the condition \eqref{eq:ids-opt-cond2} on $q_{n + 1}$, the definition of $L$, and applying Proposition~\ref{th:ky-fan} we obtain
\begin{eqnarray*}\rho_{\alpha,\lambda}(\bar{I}_+) \geq& \overline\rho_{\frac{p_{n+1}}{q_{n+1}},\lambda} (\bar{I}) -\frac{4}{L} 
\geq \frac{1}{q_{n+1}}- \frac{2}{5q_{n + 1}} \geq   \frac{1}{2q_{n+1}}~.
\end{eqnarray*}
By the condition \eqref{eq:ids-opt-cond2} on $q_{n + 1}$ we obtain
$$ 5 \exp\left(-\frac{q_{n+1}\lambda^{q_n}}{C' q_n}\right) \leq x_n,$$
whence
\[
\omega\left( 5 \exp\left(-\frac{q_{n+1}\lambda^{q_n}}{C' q_n}\right) \right) \log\frac{1}{5 \exp\left(-\frac{q_{n+1}\lambda^{q_n}}{C' q_n}\right)} \leq \frac{\lambda^{q_n}}{q_n^2}~,
\]
therefore by \eqref{eq:ids-opt-cond1}, we have
\[
 \frac{1}{2q_{n+1}} \geq \frac{q_n^2}{2 q_{n + 1}\lambda^{q_n}} \omega\left( 5 \exp\left(-\frac{q_{n+1}\lambda^{q_n}}{C' q_n}\right) \right) \log\frac{1}{5 \exp\left(-\frac{q_{n+1}\lambda^{q_n}}{C' q_n}\right)}~.
\]
Consequently, we have
\[
\begin{split}
\rho_{\alpha,\lambda}(\bar{I}_+)  &\geq  \frac{q_n^2}{2 q_{n + 1}\lambda^{q_n}} \omega\left( 5 \exp\left(-\frac{q_{n+1}\lambda^{q_n}}{C' q_n}\right) \right) \frac{q_{n + 1}\lambda^{q_n}}{C'' q_n} \\& \geq \frac{q_n}{C'''} \omega\left( 5 \exp\left(-\frac{q_{n+1}\lambda^{q_n}}{C' q_n}\right) \right)\geq  \frac{q_n}{C'''}\omega(|\bar I_+|)~,
\end{split}
\]
and in particular
\[
\lim_{n\to\infty}\frac{\rho_{\alpha,\lambda}(\bar{I}_+)}{\omega(|\bar I_+|)} = \infty~.
\]
This concludes the proof of Theorem~\ref{th:ids}-\eqref{thm:ids-opt}. \qed

\section*{Acknowledgements}
The authors thank S.~Jitomirskaya for useful discussions on the homogeneity of the spectrum. 
A.~Avila was supported by the ERC Starting Grant \textquotedblleft
Quasiperiodic\textquotedblright. This work was initiated while M.~Shamis was a postdoc at Weizmann Institute, supported by the ISF grant 147/15, and while Q.~Zhou was a postdoc at  
Universit\'e Paris Diderot supported by ERC Starting Grant \textquotedblleft
Quasiperiodic\textquotedblright. 
Q.~Zhou was also partially supported by  National Key R\&D Program of China (2020YFA0713300), NSFC grant (12071232), the Science Fund for Distinguished Young Scholars of Tianjin (No. 19JCJQJC61300) and Nankai Zhide Foundation. We thank the referees for suggesting several improvements and corrections.

\end{document}